\documentclass[11pt]{article}
\usepackage[margin=1.1in]{geometry}
\usepackage{amsmath,amssymb,amsthm}
\usepackage{graphicx}
\usepackage[hidelinks]{hyperref}
\hypersetup{
  pdftitle={Minimum-makespan completion and vertex selection leave the Wang-Sitters constant at 11/6},
  pdfauthor={Adam Y. Shavit},
  pdfsubject={Approximation algorithms; graph balancing; scheduling on unrelated machines},
  pdfkeywords={graph balancing, restricted assignment, approximation algorithm,
               linear programming relaxation, rounding, Wang-Sitters, makespan}
}

\newtheorem{theorem}{Theorem}
\newtheorem{lemma}[theorem]{Lemma}
\newtheorem{proposition}[theorem]{Proposition}
\newtheorem{corollary}[theorem]{Corollary}
\theoremstyle{definition}
\newtheorem{definition}[theorem]{Definition}
\theoremstyle{remark}
\newtheorem{remark}{Remark}
\newtheorem{measurement}[theorem]{Measurement}

\newcommand{\OPT}{\mathrm{OPT}}
\newcommand{\ALG}{\mathrm{ALG}}

\title{Minimum-makespan completion and vertex selection\\
leave the Wang--Sitters constant at $11/6$}
\author{Adam Y. Shavit\\
\small Hunter College and the Graduate Center, CUNY\\
\small \texttt{as1127@hunter.cuny.edu} \quad ORCID 0009-0008-1235-0995}
\date{\today}

\newcommand{\supid}{arXiv:2609.03890}
\begin{document}
\maketitle

\begin{abstract}
Wang and Sitters' rounding scheme for graph balancing (each job allowed on at
most two machines) keeps the makespan, the busiest machine's load, within $11/6$
of the optimum, and the companion note \cite{Paper1} shows that no smaller
constant holds. One suspected cause is Step~3's freedom to match leftover jobs
to slots in any valid way. We show that removing that freedom does not improve
the constant. An idealized oracle that always picks the valid matching with the
smallest makespan still has worst-case constant exactly $11/6$ against the
optimum. Both natural $7/4$ bounds for it, against the relaxation's target and
against the optimum, are false. Requiring Step~1 to return a vertex (corner) of
the relaxation does not help when that vertex is chosen adversarially. The loss
therefore does not come from Step~3's freedom alone. We also prove a reduction
confining every overload above $7/4$ to two shapes, give a seven-machine
instance on which the natural two-phase repair fails, and prove a strict $7/4$
bound on the generalized three-path family.
\end{abstract}

\section{Introduction}

Wang and Sitters' rounding scheme for graph balancing \cite[\S3]{WS16}
guarantees makespan at most $\tfrac{11}{6}T$, and strictly below it by the
companion note's ceiling \cite{Paper1}, when applied to a feasible
solution of the relaxation of Ebenlendr, Kr\v{c}\'al and Sgall \cite{EKS14}. It
leaves two choices unspecified: Step~1 may return \emph{any} feasible fractional
solution, and Step~3 \emph{any} valid Shmoys--Tardos slot matching \cite{ST93}.
A companion note \cite{Paper1} measures what that latitude permits: the supremum
of $\ALG/\OPT$ over all permitted executions is $\tfrac{11}{6}$, approached and
never attained, and no threshold $\beta \in (\tfrac12,1)$ improves the
guarantee against the relaxation's own value, $\tfrac23$ being the unique
minimizer. Section~\ref{sec:prelim} restates what this note needs from that
one, in statements without proofs. This note asks how much of the $\tfrac{11}{6}$ loss arbitrary matching selection causes.

We show that optimizing the matching does not improve the worst-case constant.
Even when Step~3 returns a makespan-minimizing valid matching, the
threshold-relative constant remains exactly $\tfrac{11}{6}$, and the
corresponding $\tfrac74$ guarantee against the true optimum is false.
Restricting Step~1 to vertex solutions does not remove the obstruction either:
a family of vertices again approaches $\tfrac{11}{6}$. The obstruction is
therefore structural to this slot-rounding framework when Step~1 is adversarial,
rather than an artifact of arbitrary matching selection, and it persists in both
of these more selective versions of the same scheme; a Step-1 rule that chooses
its vertex well is not covered.

\paragraph{Contributions.} Three theorems answer the same question in three progressively more restricted settings of the scheme.

\begin{enumerate}
\item \textbf{Optimizing Step~3 does not improve $\tfrac{11}{6}$}
(Section~\ref{sec:74}). We distinguish two natural $\tfrac74$ candidate guarantees for the best-matching variant, one relative to the relaxation's threshold and one relative to the optimum. Neither is stated as a conjecture in \cite{EKS14} or \cite{WS16}, and both are false. Separating them is the first step: the threshold-relative one falls to a $19$-machine instance at a least feasible $T$ on which \emph{every} valid matching reaches $44/25 > \tfrac74$ (Theorem~\ref{thm:74false}). Theorem~\ref{thm:116} then fixes the constant any
such bound must reach at exactly $\tfrac{11}{6}$ --- the guarantee the plain
scheme already has --- so no rule for choosing the matching can do better
against the threshold.

\item \textbf{The oracle is not a $\tfrac74$-approximation}
(Section~\ref{sec:oracle}). Against the \emph{optimum} the question is a
different one, and the answer is the same. Theorem~\ref{thm:74bfalse} gives each
of the previous family's unit jobs its own partner machine: every forced load is
unchanged, the optimum drops to $1$, and the oracle is pushed to
$\tfrac{85}{48} > \tfrac74$ at $g = 4$. The bound is over \emph{every} valid
matching, so it applies to any selection rule whatever, and
Corollary~\ref{cor:oracle} puts the oracle's worst ratio to the optimum at
exactly $\tfrac{11}{6}$: a rule that always picks the best matching guarantees
no more than the polynomial-time scheme already does.

\item \textbf{Restricting Step~1 to an adversarially chosen vertex does not improve $\tfrac{11}{6}$ either}
(Section~\ref{sec:vertex}). Neither refuting point is a vertex, so a variant
selecting a vertex solution of the relaxation --- which can be done in
polynomial time, and is also the kind of solution a simplex implementation
returns --- was the last escape for an adversarial Step~1. Theorem~\ref{thm:pinvertex} closes it
with a family whose points are vertices from the start and whose forced load is
the same; Proposition~\ref{prop:compressed} rebuilds that family on $5g+2$
machines. Because that family sits at $T = \OPT = 1$, it settles both scales at
once: Corollary~\ref{cor:vertexconst} puts the vertex-restricted constant at
$\tfrac{11}{6}$ against the threshold and against the optimum alike, and
Corollary~\ref{cor:anyrule} extends it to every rule for choosing the matching,
randomized rules included. Theorem~\ref{thm:invariance} collects the four
settings into one statement.
\end{enumerate}

\noindent
The table compares the four settings, their constants, and the supporting results:

\smallskip
\centerline{\setlength{\tabcolsep}{9pt}\renewcommand{\arraystretch}{1.15}%
\begin{tabular}{llccl}
\hline
Step~1 returns & Step~3 returns & $\sup \ALG/\OPT$ & Attained? & Where\\
\hline
any feasible $x$ & any valid matching & $11/6$ & no & \cite{Paper1}\\
any feasible $x$ & a best valid matching & $11/6$ & no & Cor.~\ref{cor:oracle}\\
a vertex & any valid matching & $11/6$ & no & Thm.~\ref{thm:pinvertex}\\
a vertex & a best valid matching & $11/6$ & no & Thm.~\ref{thm:pinvertex}\\
\hline
\end{tabular}}
\smallskip

\noindent
Every row carries the same supremum, $\tfrac{11}{6}$, and in no row is it
attained: neither optimizing Step~3 nor restricting Step~1 to vertices of the
relaxation improves the worst-case constant. The third and fourth rows are one
theorem, because its bound is over every valid matching and therefore covers the
best one; that Theorem~\ref{thm:pinvertex}'s family approaches
$\tfrac{11}{6}$ at vertices without reaching it, and that nothing else reaches it
either, is the strict ceiling of \cite{Paper1} applied at vertices, which are
feasible points like any other. Theorem~\ref{thm:invariance} states the four
rows as one statement, and Figure~\ref{fig:invariance} draws them.

\paragraph{Other results and their locations.} Section~\ref{sec:variant}
states the oracle and shows that selecting its matching is NP-hard for a
\emph{supplied} fractional solution --- weaker than it looks, since Step~1 is
not adversarial, so the variant's own complexity stays open. For a supplied
solution the hardness is strong already at two slots per machine
(Theorem~\ref{thm:twoslot}).
Section~\ref{sec:classify} confines every overload above $\tfrac74$ to two
explicit shapes (Theorem~\ref{thm:reduction}) and shows the natural two-phase
proof cannot work (Proposition~\ref{prop:twophase}).
Section~\ref{sec:threepath} bounds the generalized three-path family --- the
extremal family already known for this relaxation --- strictly below
$\tfrac74$ (Proposition~\ref{prop:familyceiling}).

\paragraph{Proofs and measurements.}
Theorems and propositions below are proved, except where a statement says
otherwise, and no statement in the body now rests on a measurement.
Definition~\ref{conj:74} records two candidate guarantees, and both are
disproved here: form~(a)
by Theorem~\ref{thm:74false} and form~(b) by Theorem~\ref{thm:74bfalse}, on
different instances. Neither was found by searching, and the searches reported in
Section~\ref{sec:vertex} would not have found either. A statement labelled
\emph{Measurement} reports what the deposited code returned at the stated
parameters; it is not a proof. Appendix~\ref{app:computational} carries the
computational record --- controls, independent implementations, search sizes and
the deposited script behind each measured claim that has one --- so that the body is not
interleaved with it.

\section{Preliminaries}\label{sec:prelim}

This section fixes the objects and states, without proof, the three results of
\cite{Paper1} that the arguments below use. Everything in it is that note's;
nothing here is new, and the proofs are there.

\paragraph{The relaxation.} The scheme is run at a \emph{target} $T$ and asks
whether the jobs can be scheduled within it. Write $p_j$ for the size of job
$j$, which is \emph{positive}. For a given $T$ the relaxation below is either
feasible or not, so the \emph{least feasible target}, written
$T_{\mathrm{LP}}$, is the smallest $T$ at which the relaxation is feasible. That
minimum is attained and not merely approached, which \cite{Paper1} proves;
binary search on $T$ is how this note computes the value, not what defines it.
It is at most the true optimum $\OPT$, because an
integral schedule of makespan $\OPT$ is itself a feasible fractional solution at
$\OPT$. \textbf{The two need not be equal, and keeping them apart is what most
of this note is about.} A target that is fixed is normalized to $T = 1$; where
a passage computes $T_{\mathrm{LP}}$ itself, as in Section~\ref{sec:threepath}
and in the measurements, the target varies.

The relaxation is that of \cite{EKS14}, also used by \cite{WS16}. The variable
$x_{ij} \in [0,1]$ is the fraction of job $j$ placed on machine $i$, and is
supported on the at most two machines job $j$ is allowed to use. Its three
constraint families are
\[
\sum_i x_{ij} = 1 \ \ \forall j, \qquad
\sum_j x_{ij} p_j \le 1 \ \ \forall i, \qquad
\sum_{j : p_j > 1/2} x_{ij} \le 1 \ \ \forall i ,
\]
together with $x_{ij} = 0$ whenever $p_j > 1$: every job is fully placed, no
machine is loaded above the target, and no machine carries more than one big
job's worth of fraction. The third family is what separates this relaxation from
the assignment relaxation. It is \emph{not} the configuration linear program,
which is strictly stronger, and nothing here is a claim about rounding that one.
This relaxation's integrality gap is exactly $7/4$: the upper bound is the
algorithm of \cite{EKS14} and the lower bound their \emph{three-path family},
three long odd paths joining two vertices with weights alternating $1$ and
$1/2 - \varepsilon$ and a dedicated load $1/4$ everywhere. Unless
P${}={}$NP, no polynomial-time algorithm for graph balancing has ratio strictly
smaller than $3/2$ \cite{AJMOZ11}. Call a job
\emph{big} when its size exceeds $T/2$.

\paragraph{The three steps.} All three are Wang and Sitters' \cite[\S3]{WS16},
not \cite{EKS14}'s; the two are different algorithms with different ratios.
\begin{enumerate}
\item[\textbf{Step 1.}] Return \emph{a} feasible solution $x$ of the relaxation
above.
\item[\textbf{Step 2.}] Assign big job $j$ to machine $i$ whenever
$x_{ij} \ge \beta$, with $\beta = 2/3$. Such a job is scheduled on $i$ and takes
no further part: delete its fractions from $x$ on both of its machines. No
machine receives two, because the relaxation's third constraint gives each
machine at most one big job's worth of fraction and $\beta > 1/2$.
\item[\textbf{Step 3.}] Build the \emph{slot structure} of \cite{ST93} from
what is left of $x$, and return \emph{any} valid matching of the surviving jobs
to slots.
\end{enumerate}
Steps~1 and~3 each return \emph{a} feasible object rather than \emph{the} one.
Wang and Sitters prove that every run finishes within $\tfrac{11}{6}T$.

\paragraph{The slot structure and the half-open convention.} On each machine,
order its fractional jobs by nonincreasing size and cut the fractional mass into
unit slots, the last one possibly shorter; a job is adjacent to every slot its own
fraction touches. Slots and job fractions are half-open intervals
$[\,\cdot\,,\cdot\,)$, and \emph{touches} means nonempty intersection, so a job
whose fraction ends exactly where a slot begins is \emph{not} adjacent to that
slot. This is \cite{ST93}'s own convention --- their construction adds the edge
carrying a job into the next slot only if $\sum_{j \le j_s} x_{ij} > s$
\cite[p.~464]{ST93}, a strict inequality --- and it is not cosmetic: the
$\tfrac{11}{6}$ ceiling below is false under the closed reading, at a
three-job witness \cite{Paper1} exhibits. Two consequences are used here. The
number of slots on machine $i$ is $\lceil M_i \rceil$ where $M_i$ is its
surviving fractional mass, so a machine of mass exactly $2$ gets exactly two
slots, which is the count the pigeonhole of Theorem~\ref{thm:116} runs on. And
a \emph{valid matching} assigns each remaining job to one adjacent slot, at most
one job per slot; one exists by Hall's theorem.

\paragraph{The execution space.} Because two of the three steps do not determine
their own output, the object of study is fixed before anything is stated about it.
Write $P(I)$ for the set of feasible solutions of the relaxation on instance $I$
at threshold $T$ and, for $x \in P(I)$, write $\mathcal{M}(I, x)$ for the set of
valid slot matchings available in Step~3 after Step~2 has run on $x$. The
\emph{execution space} of the scheme on $I$ is
\[
  \mathcal{E}(I) \;=\; \{\,(x, M) \;:\; x \in P(I),\ M \in \mathcal{M}(I, x)\,\},
\]
and $\ALG(I, x, M)$ is the makespan the scheme produces on that execution. The
constant \cite{Paper1} computes is
\[
  R_{\mathrm{WS}} \;=\; \sup_{I}\ \sup_{(x,M)\,\in\,\mathcal{E}_{T_{\mathrm{LP}}(I)}(I)}
  \frac{\ALG(I, x, M)}{\OPT(I)} ,
\]
where the subscript records that the execution space is taken at the least
feasible target, which is where the scheme is run: the numerator is a makespan
at that target and the denominator is the true optimum.
Three objects must be kept apart, and this note refers to them by these names
throughout.

\smallskip
\centerline{\begin{tabular}{lll}
\hline
Name & Step~1 returns & Step~3 returns\\
\hline
the \emph{scheme} & any $x \in P(I)$ & any $M \in \mathcal{M}(I,x)$\\
an \emph{implementation} & one $x$, solver-determined & one $M$, routine-determined\\
the \emph{oracle} & any $x \in P(I)$ & a makespan-minimizing $M$\\
\hline
\end{tabular}}
\smallskip

\noindent
Every theorem below is about the first or the third. Nothing here is a theorem
about the second: an implementation realizes one execution per instance, and
which one depends on a solver and a matching routine that the published
specification does not fix. A report of what a solver actually returned is
evidence about one implementation, and this note labels it as a measurement,
never as a bound.

\paragraph{Per-machine notation.} Fix a machine $i$ and a valid matching. Write
$b_i$ for the size of the job assigned to $i$ in Step~2 (zero if none); $L'_i$
for the relaxation load on $i$ of the jobs that entered the slot structure. Then
$L'_i \le 1 - \beta b_i$, since the job Step~2 took contributed at least
$\beta b_i$ to machine $i$'s load constraint. Write $M_i$ for the total fraction
on $i$ of those same jobs and $K_i = \lceil M_i \rceil$ for the number of slots;
$s_z$ for the largest size adjacent to slot $z$; and $w_z$ for the
size-weighted mass of slot $z$, so that $\sum_{z \le K_i} w_z = L'_i$.

\paragraph{Results imported from \cite{Paper1}.} The three statements below are
proved there and used here without proof.

\medskip\noindent
\textbf{The chain bound \cite{Paper1}.} \emph{For every machine $i$ carrying at
least one slot job and every valid matching:} (i) $w_{K_i} > 0$, and
$s_z \le w_{z-1}$ for $2 \le z \le K_i$, \emph{hence}
$\sum_{z=2}^{K_i} s_z \le L'_i - w_{K_i}$; (ii) \emph{the load of $i$ is at most}
$b_i + s_1 + L'_i - w_{K_i}$.

\medskip\noindent
\textbf{The $11/6$ ceiling \cite{Paper1}.} \emph{Let the relaxation be feasible
at $T$, let $x$ be any feasible solution, and let any valid matching be taken in
Step~3 with $\beta = 2/3$. Then every machine's load is strictly less than}
$\tfrac{11}{6}T$.

\medskip\noindent
\textbf{The per-machine trichotomy \cite{Paper1}.} \emph{Fix a valid matching
and a machine $i$ carrying at least one slot job, and normalize $T = 1$. Exactly
one of the following holds, and each carries its own bound.}
\begin{itemize}
\item[(N0)] \emph{$b_i = 0$ and the job matched to slot $1$ is absent or of size
at most $\tfrac12$; then} $\mathrm{load}(i) < \tfrac32$.
\item[(N1)] \emph{$b_i = 0$ and a big job $q$ is matched to slot $1$; then}
$\mathrm{load}(i) < 2 - \tfrac{x_{iq}}{2}$.
\item[(S)] \emph{$b_i > 0$; then} $\mathrm{load}(i) < \tfrac32 + \tfrac{b_i}{3}$.
\end{itemize}
\emph{The ceiling above is what these three give together, since
$b_i \le 1$ and a big job that survives Step~2 has share below $\beta = \tfrac23$
on each of its two machines, hence $x_{iq} > \tfrac13$.} Theorem~\ref{thm:reduction} uses the three
bounds rather than the ceiling they imply, which is why they are imported
separately.

\medskip\noindent
\textbf{The $11/6$ family \cite{Paper1}.} \emph{For every
$\delta \in (0, 1/6)$ there are a three-job, three-machine instance with
$\OPT = T = 1$, a feasible relaxation solution and a valid matching whose
makespan is} $\tfrac{11}{6} - \delta$. \emph{Together with the ceiling this
fixes} $R_{\mathrm{WS}} = 11/6$, \emph{a supremum attained by no instance.} A
second, six-job family there is tight only against the relaxation's target,
where its ratio to the true optimum tends to $\tfrac{11}{8}$ along
the family and is below it at every member; we call it the
\emph{six-job family} below.

\section{The threshold-relative constant is exactly \texorpdfstring{$\tfrac{11}{6}$}{11/6}}\label{sec:74}

This section proves the first of the note's three central theorems, the benchmark for the other two. It says that choosing the Step-3 matching perfectly
--- by a rule returning a makespan-minimizing valid matching
--- does not move the worst-case constant against the threshold at all: it stays
at $\tfrac{11}{6}$, where \cite{WS16}'s analysis already put every run of the
plain scheme. Two statements have to be separated before that can be said
cleanly, and the separation is itself a result, so this section states both
$\tfrac74$ candidate guarantees, refutes the threshold-relative one, and then
pins the constant.

Two natural $\tfrac74$ candidate guarantees for this framework, one relative to
the threshold and one relative to the optimum, are not the same statement;
Theorem~\ref{thm:74false} below separates them. We
state both, mark each with its status, and record which way the
implication runs.

\emph{Attribution.} Neither form appears explicitly in \cite{EKS14} or
\cite{WS16}. The nearest published statement is \cite[\S5]{EKS14}: ``Even for
the case of graphs there is the remaining gap between $1.5$ and $1.75$. It would
be nice to have a tight(er) bound,'' which records the gap as open and says
nothing about slot matchings or about the best-matching variant.
\cite{WS16}'s only conjecture concerns a $3/2$-approximation for restricted
assignment with interval processing sets, a different problem. We therefore
introduce (a) and (b) as two natural formalisations of the $\tfrac74$ question
for this rounding framework, and not as anyone's conjecture.

\begin{definition}[two candidate guarantees]\label{conj:74}
The following two natural guarantees differ in what the makespan is measured
against.
\begin{enumerate}
\item[(a)] \emph{Threshold-relative form.} Let $x$ be a feasible
relaxation solution at threshold $T$. Then some valid slot matching has
makespan at most $\tfrac74 T$. Equivalently, the best-matching variant of
Section~\ref{sec:variant} has makespan at most $\tfrac74 T$. \textbf{False}, by
Theorem~\ref{thm:74false}.
\item[(b)] \emph{Optimum-relative form.} For every instance, the
best-matching variant run at the least feasible threshold
$T_{\mathrm{LP}}$ is a $\tfrac74$-approximation: its makespan is at most
$\tfrac74 \OPT$. \textbf{False}, by Theorem~\ref{thm:74bfalse}, though
not refuted by Theorem~\ref{thm:74false}: the two forms are refuted by different
instances, which is the point of separating them.
\end{enumerate}
\end{definition}

\noindent
\emph{The difference between the two forms.}
Form (a) compares the rounding's output to the \emph{threshold the
relaxation was solved at}. Form (b) compares it to the \emph{true
optimum}. The two coincide only when the relaxation is tight.

Since the algorithm binary searches to the least feasible threshold and
$T_{\mathrm{LP}} \le \OPT$ --- any integral schedule of makespan $\OPT$ being a
feasible fractional solution at $\OPT$ --- form (a) would imply form (b). The
converse need not hold, because a rounding may exceed
$\tfrac74 T_{\mathrm{LP}}$ while remaining below $\tfrac74 \OPT$.
Theorems~\ref{thm:74false} and~\ref{thm:74bfalse} exploit precisely this
distinction: the two forms fall to different instances, and
Theorem~\ref{thm:74bfalse}'s is not a perturbation of
Theorem~\ref{thm:74false}'s but a one-line change to the family that determines
the threshold-relative constant.

\begin{remark}[what the integrality gap does and does not bound]
\label{rem:gapscope}
Form (a) invites the justification that it is the strongest
statement available for any procedure rounding this relaxation, since the
relaxation's own integrality gap is exactly $7/4$. That inference is
invalid, and it is wrong independently of the counterexample; the
counterexample only makes it consequential.

The premise is very nearly right, and worth stating precisely. For the
relaxation this note actually uses,
Ebenlendr, Kr\v{c}\'al and Sgall put the integrality gap of (LP2) and (LP3) at
exactly $1.75$ \cite[\S4.1]{EKS14}. That is the ceiling the inference appeals
to, and it is correct. The neighbouring statement about the \emph{configuration}
linear program is a different and strictly stronger relaxation, and is worth
keeping apart: Verschae and Wiese record that once each job is required to have
the same processing time on both its machines --- which is exactly graph
balancing --- the configuration linear program's gap is \emph{at most} $7/4$, attributing
the bound implicitly to \cite{EKS08}, and is $2$ once that requirement is
dropped \cite[p.~374]{VW14}. That upper half has since been improved to $1.749$
by \cite[Thm.~1]{JR19}, so it is not the current ceiling for that relaxation
either.

The integrality-gap ceiling does not imply form (a). The
integrality gap bounds $\OPT / T_{\mathrm{LP}}$. It says nothing about what a \emph{particular} rounding,
confined to the slot structure, can reach relative to $T$. On the instance
of Theorem~\ref{thm:74false} the gap is comfortable, $\OPT/T_{\mathrm{LP}} = 3/2
\le 7/4$, while every valid slot matching is at $44/25$. That instance shows the most general finding here: slot-matching rounding cannot meet this relaxation's ceiling.
Theorem~\ref{thm:116} then gives the shortfall exactly: against the
threshold this rounding's constant is $\tfrac{11}{6}$ and the gap's is
$\tfrac74$, so the two differ by $\tfrac1{12}$, and no tightening of the
gap would have shown it.
\end{remark}

Form (b) would \emph{not} have given a
$7/4$-approximation algorithm: no polynomial-time implementation of the
oracle is known, and \cite{EKS14} already achieves $1.75$ in polynomial
time by other means. Its interest was that an algorithm of little more than a page plus a
matching-selection rule would match that ratio, and that a polynomial
selection rule, if one existed, would then be a genuinely simpler
$1.75$-approximation. Theorem~\ref{thm:74bfalse} rules that out for the
algorithm as stated, where Step~1 may return any feasible solution: no
selection rule reaches $7/4$ there, because the bound it violates holds for
\emph{every} valid matching. It does not rule out a Step-1 rule that returns a
\emph{vertex}, which is polynomial; Remark~\ref{rem:74bvertex} separates
that case out, and Theorem~\ref{thm:pinvertex} shows that the vertex
restriction alone does not restore $7/4$, though it does not exclude a rule that
picks a better vertex. We report the structure we can prove, and the two instances that
refute the two forms.

Two counting statements are used twice each below, so they are separated out
first. Both are pigeonhole arguments; stating them separately makes the counterexamples reproducible rather than merely verifiable: they say what a
construction has to arrange, and the arithmetic of arranging it is then all that
is left.

\begin{lemma}[slot deficiency forces jobs back]\label{lem:deficiency}
Let $S$ be a set of slots and let $J$ be a set of jobs, each of which is adjacent
to no slot outside $S$ except, possibly, slots of a single machine designated for
it. Then every valid matching places at least $|J| - |S|$ jobs of $J$ on their
designated machines.
\end{lemma}

\begin{proof}
A valid matching gives each slot at most one job, so at most $|S|$ jobs of $J$ are
matched into $S$ and at least $|J| - |S|$ are matched outside it. A job of $J$
matched outside $S$ occupies a slot of its designated machine, and in particular
has one.
\end{proof}

\begin{lemma}[deficiency collision]\label{lem:collision}
Let $A$ and $B$ be sets of machines drawn from a set of $m$ machines, with
$|A| \ge d_1$ and $|B| \ge d_2$. If $d_1 + d_2 > m$ then $A \cap B \ne \emptyset$.
\end{lemma}

\begin{proof}
Otherwise $A$ and $B$ are disjoint and $d_1 + d_2 \le |A| + |B| \le m$.
\end{proof}

\noindent
The two together are the shape of the lower bound that follows: a deficiency
forces a class of jobs back onto designated machines, a second deficiency forces
another class back, and when the two counts exceed the number of designated
machines some machine receives one of each. A construction must then arrange that a machine holding one of each also carries something it cannot shed.

\begin{theorem}[form (a) is false]\label{thm:74false}
There is an instance of graph balancing on $19$ machines with $39$ jobs,
and a feasible relaxation solution at $T = 1$ with $T = 1$ least feasible,
for which \emph{every} valid slot matching has makespan at least
$44/25 = 1.76 > \tfrac74$. Its true optimum is $\OPT = 3/2$, so on this
instance $\ALG/\OPT = 88/75 \approx 1.173$.
\end{theorem}

\emph{Strategy.} The instance is three copies of one six-machine gadget,
sharing a single partner machine $R$. Two counting arguments each force a
number of jobs to sit at their own blocker; the two counts sum to more
than the number of blockers, so some blocker takes both, and its own
dedicated job has nowhere else to go. No search is involved, and
Figure~\ref{fig:collision} draws the two counts and the collision.

\begin{proof}
Write $G_k$ for the
nine \emph{blockers}, $c_k$ for the nine \emph{crammers} and $R$ for the
shared partner, $k = 1,\dots,9$. A blocker is a machine the counting forces a
job back onto; a crammer is a machine that job could otherwise have gone to. Each blocker gets three jobs: a
\emph{big} $q_k$ of size $1$, a \emph{straddler} $u_k$, the only job
supported on both a blocker and $R$, and a \emph{dedicated} $v_k$,
supported on its blocker alone. Twelve \emph{fillers} sit on crammers
only, four to each of three copies. In full:
\[
\begin{array}{llll}
q_k & \text{size } 1 & x = \{G_k: \tfrac{17}{50},\ c_k: \tfrac{33}{50}\}
& k = 1,\dots,9\\[2pt]
u_k & \text{size } \tfrac12 & x = \{G_k: \tfrac{39}{50},\ R: \tfrac{11}{50}\}
& k = 1,\dots,9\\[2pt]
v_k & \text{size } \tfrac{13}{50} & x = \{G_k: 1\} & k = 1,\dots,9\\[2pt]
j_{1,t},\ j_{2,t} & \text{size } \tfrac{2537}{10000}
& x = \{c_{3t-2}: \tfrac12,\ c_{3t-1}: \tfrac12\} & t = 1,2,3\\[2pt]
j_{3,t} & \text{size } \tfrac{2537}{10000}
& x = \{c_{3t-2}: \tfrac{33}{100},\ c_{3t}: \tfrac{67}{100}\} & t = 1,2,3\\[2pt]
j_{4,t} & \text{size } \tfrac{2537}{10000}
& x = \{c_{3t-1}: \tfrac{33}{100},\ c_{3t}: \tfrac{67}{100}\} & t = 1,2,3.
\end{array}
\]
This $x$ is feasible at $T = 1$: maximum load $999958/1000000$, maximum
big budget $33/50$, every job supported on at most two machines, every
size at most $1$. It is \emph{least} feasible because $q_k$ has size $1$
and the relaxation zeroes $x_{ij}$ whenever $p_j > T$. The largest big
share is $33/50 < \beta = 2/3$, so Step~2 assigns nothing and
$b_i \equiv 0$ throughout. Both of those depend on the convention of
Section~\ref{sec:prelim}: the straddlers have size exactly $\tfrac12$ and so are
not big, which is what keeps $G_k$'s big-job budget at $\tfrac{17}{50}$ rather
than $\tfrac{56}{50} > 1$, and what stops Step~2 from assigning a share of
$\tfrac{39}{50} \ge \beta$. Measurement~\ref{meas:116} moves
Theorem~\ref{thm:116}'s family off the same boundary for the same reason.

Now fix any valid matching.

\emph{(i) Every big job is supported only at its own blocker and one crammer.} At $G_k$ the slot
layout is $q_k$ on $[0,\tfrac{17}{50})$, $u_k$ on
$[\tfrac{17}{50},\tfrac{28}{25})$, $v_k$ on
$[\tfrac{28}{25},\tfrac{53}{25})$. That is three slots. The fact that matters here concerns
supports, not loads: $q_j$ has zero share on $G_k$ for $j \ne k$, so
$\operatorname{supp}(q_j) = \{G_j, c_j\}$ and no big job is adjacent to any slot
of a blocker other than its own. The supports settle it directly; a ``cap''
$\tfrac74 - \mathrm{tail}$ on what a slot-$1$ occupant may carry does not,
since $\mathrm{tail}$ bounds what the later slots \emph{may} hold, not what
they do, so exceeding the cap implies nothing.

\emph{(ii) At least three of the $q_k$ sit at their own blocker.} Apply
Lemma~\ref{lem:deficiency} with $S$ the crammers' slots and $J$ the nine $q_k$
together with the twelve fillers. Each crammer carries two slots, so
$|S| = 18$; the fillers are supported on crammers alone, so they have no
designated machine; and by (i) a $q_k$ meets no slot outside $S$ except at its
own blocker $G_k$, which is therefore its designated machine. Since
$|J| - |S| = 21 - 18 = 3$, at least three of the $q_k$ sit at their own blocker.

\emph{(iii) At least seven straddlers sit at their own blocker.} The
fractional mass at $R$ is $99/50$, so $R$ carries $K_R = 2$ slots. Apply
Lemma~\ref{lem:deficiency} again, with $S$ those two slots and $J$ the nine
$u_k$, each designated to its own blocker, the only other machine it is
supported on: $9 - 2 = 7$ of them sit at their own blocker.

\emph{(iv) Conclusion.} Take $A$ for the blockers holding their own $q_k$ and
$B$ for those holding their own $u_k$, so $|A| \ge 3$ and $|B| \ge 7$ among
$m = 9$ blockers. Since $3 + 7 > 9$, Lemma~\ref{lem:collision} gives a blocker
$G_k$ holding both its own $q_k$ and its own $u_k$. Its dedicated $v_k$ is
supported on $G_k$ alone and has nowhere else to go, so the three jobs
occupy that machine's three slots, and the load there is
\[
1 + \tfrac12 + \tfrac{13}{50} \ =\ \tfrac{44}{25} \ =\ 1.76 \ >\ \tfrac74 .
\]

\emph{(v) No machine exceeds $\tfrac{44}{25}$.} Only $q_k$, $u_k$ and $v_k$ are
supported on a blocker $G_k$, so its load is at most
$1 + \tfrac12 + \tfrac{13}{50} = \tfrac{44}{25}$. A crammer carries two slots (step (ii)) and
so holds at most two jobs; only $q_k$ among them has size $1$, and the rest are
fillers of size $\tfrac{2537}{10000}$, so its load is at most
$1 + \tfrac{2537}{10000} < \tfrac{44}{25}$. The
slot count is needed here: the first two crammers of each copy support $q_k$
and three fillers, which together weigh $\tfrac{17611}{10000} > \tfrac{44}{25}$.
The shared machine $R$ carries two slots and supports only straddlers, of
size $\tfrac12$, so its load is at most $1$. With (iv), every valid matching has
makespan exactly $\tfrac{44}{25}$.
The bound $\OPT = \tfrac32$ is Lemma~\ref{lem:ce74opt}, so on this instance
every valid matching has ratio exactly $\tfrac{88}{75}$ to the optimum.
\end{proof}

\begin{figure}[htbp]\centering
\includegraphics[width=0.98\textwidth]{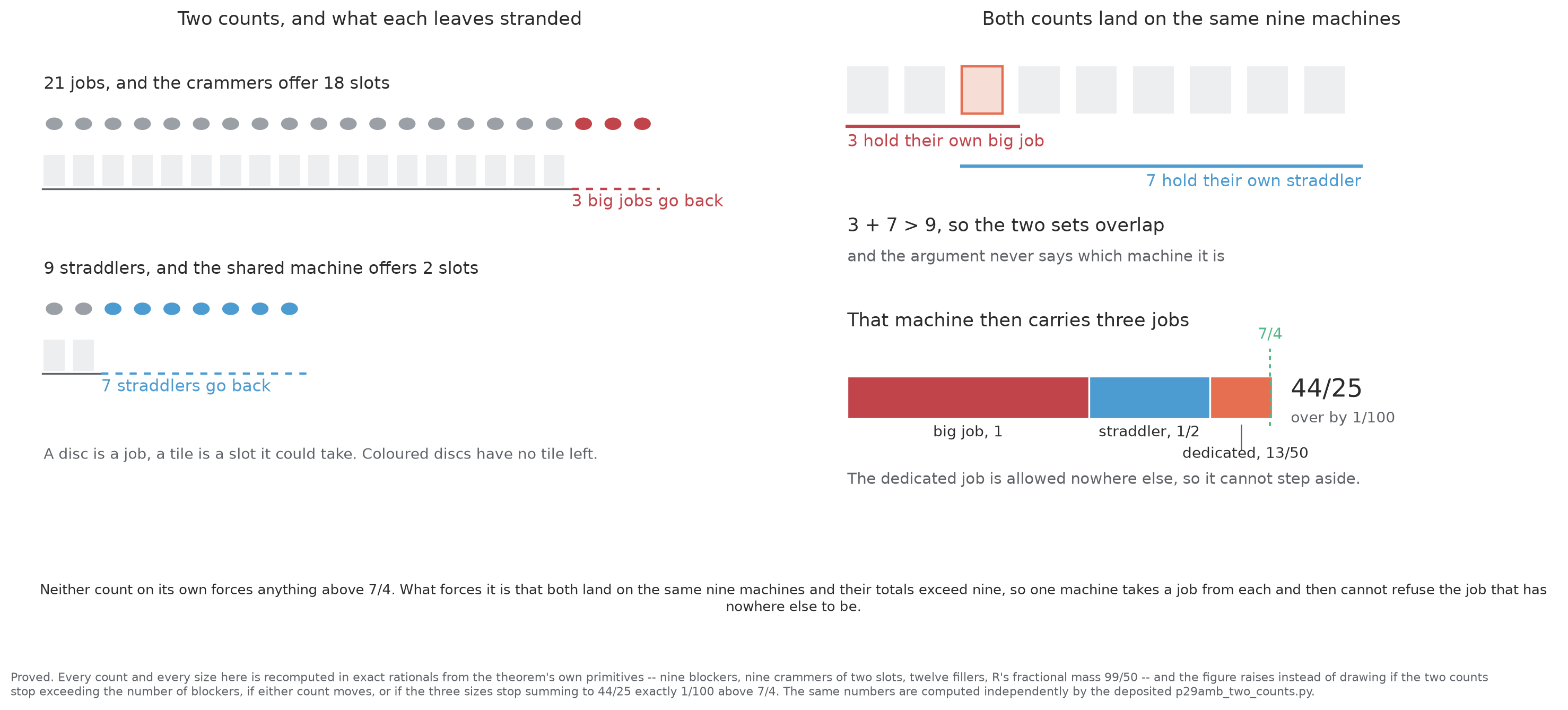}
\caption{Two deficiencies, and the collision between them. \emph{Left:}
twenty-one jobs compete for the crammers' eighteen slots, so
Lemma~\ref{lem:deficiency} returns three big jobs to their own blockers; nine
straddlers compete for the two slots of the shared machine, so seven straddlers
return to theirs. \emph{Right:} both sets lie among nine blockers and
$3 + 7 > 9$, so by Lemma~\ref{lem:collision} some blocker holds one of each ---
which blocker is not determined, and the argument does not need it --- and its
dedicated job, allowed nowhere else, joins them for a load of $\tfrac{44}{25}$,
over $\tfrac74$ by $\tfrac1{100}$. Neither count on its own forces anything
above $\tfrac74$.}
\label{fig:collision}
\end{figure}

\noindent
\emph{Sensitivity of the construction.} The counterexample lies on a slot-count boundary, on which the counting argument of step~(ii) depends: six crammers carry fractional mass $\tfrac{199}{100}$ and the three $c_{3t}$ carry exactly $2$, so all nine round up to two slots rather than three. Moving
$10^{-4}$ of $j_{3,1}$'s share from $c_1$ onto $c_3$ keeps the point feasible at
$T = 1$ and lifts one crammer's mass above $2$; that crammer gains a slot, the
deficiency --- the excess of jobs over the slots they may occupy --- falls from
$3$ to $2$, the count reads $2 + 7 = 9$ rather than $10 > 9$, and the perturbed
instance admits a valid matching within $\tfrac74$. Theorem~\ref{thm:74false}
therefore establishes the existence of a violating feasible point, and makes no
robustness claim about a neighbourhood of that point, nor any claim about how
much of the feasible region violates $\tfrac74$. The discontinuity belongs to
slot counting rather than to this construction: a slot count is a ceiling, so
any argument that counts slots is discontinuous at integer mass, and the two
boundary readings the instance depends on are \cite{WS16}'s and \cite{EKS14}'s
own (Section~\ref{sec:prelim}). Two tie-breaks here are genuinely free: the order
among the twelve equal-sized fillers, and the order among the nine equal-sized
straddlers sharing $R$. Reversing both changes $46$ of the $97$ slot-graph edges
and neither verdict; the fillers alone account for $30$ of the $46$
(\texttt{p29lever\_slotgraph\_tiebreak\_2026-09-08.py}).

\noindent
The point exhibited here is not a vertex of the relaxation. That is why the
vertex-restricted form is stated separately in Remark~\ref{rem:74bvertex}, and
why it is settled by Theorem~\ref{thm:pinvertex} rather than by any
perturbation of this instance.

\begin{lemma}[the optimum of Theorem~\ref{thm:74false}'s instance]\label{lem:ce74opt}
The instance of Theorem~\ref{thm:74false} has $\OPT = \tfrac32$.
\end{lemma}

\begin{proof}
\emph{Upper bound.} The schedule below is feasible at $\tfrac32$. Machines $0$ to $8$ are the
blockers $G_1, \dots, G_9$, machines $9$ to $17$ the crammers $c_1, \dots, c_9$,
and machine $18$ is $R$. Jobs $3k-3$, $3k-2$ and $3k-1$ are $q_k$, $u_k$ and $v_k$,
and jobs $27 + 4(t-1)$ to $30 + 4(t-1)$ are $j_{1,t}, \dots, j_{4,t}$. Each entry
gives a machine, the jobs it takes, and its load.

\smallskip
\centerline{\begin{tabular}{ll}
$0$: {0, 2}, load $63/50$ & $1$: {4, 5}, load $19/25$ \\
$2$: {7, 8}, load $19/25$ & $3$: {9, 11}, load $63/50$ \\
$4$: {13, 14}, load $19/25$ & $5$: {16, 17}, load $19/25$ \\
$6$: {18, 20}, load $63/50$ & $7$: {22, 23}, load $19/25$ \\
$8$: {25, 26}, load $19/25$ & $9$: {27, 28, 29}, load $7611/10000$ \\
$10$: {3, 30}, load $12537/10000$ & $11$: {6}, load $1$ \\
$12$: {31, 32, 33}, load $7611/10000$ & $13$: {12, 34}, load $12537/10000$ \\
$14$: {15}, load $1$ & $15$: {35, 36, 37}, load $7611/10000$ \\
$16$: {21, 38}, load $12537/10000$ & $17$: {24}, load $1$ \\
$18$: {1, 10, 19}, load $3/2$ \\
\end{tabular}}
\smallskip

\noindent
Every load is at most $\tfrac32$ and machine $18$ meets it, so
$\OPT \le \tfrac32$.

\emph{Lower bound.} Suppose some schedule had makespan strictly below
$\tfrac32$. The shared machine $R$ then holds at most two straddlers, since
three of them have size $\tfrac32$ exactly; and a straddler is allowed only on
$R$ and on its own blocker, so at least seven sit on their blockers. No blocker
holds both its own $q_k$ and its own $u_k$: its dedicated $v_k$ is allowed
nowhere else, and the three together make $\tfrac{44}{25}$. So each $q_k$ that
sits on its blocker forces $u_k$ onto $R$, and $R$ takes at most two of them:
at most two of the nine $q_k$ sit on blockers, and at least seven sit on
crammers. The crammers come in three triples of three, so some triple holds all
three of its own $q_k$. A crammer holding a unit job cannot also take two
fillers, since $1 + 2 \cdot \tfrac{2537}{10000} > \tfrac32$, so each of those
three crammers takes at most one filler --- and that copy's four fillers are
supported on those three crammers and nowhere else. Four jobs into three places
is the contradiction. Hence $\OPT \ge \tfrac32$, so $\OPT = \tfrac32$ and on
this instance $\ALG/\OPT = 88/75$.
\end{proof}

\paragraph{The instance is not asymptotically extremal.}
Three of its parameters are slack. The blocker's fractional load is
$99/100$ rather than $1$; the shared machine carries share mass
$99/50$ where $K_R = 2$ permits $2$; and the fillers use $33/100$ and
$67/100$, which puts $67/100 + 67/100 = 67/50$ on the third crammer of
each copy where exact thirds would put $4/3$. Tightening all three raises
the forced load, with supremum $16/9$. The family below supersedes that bound:
its fifth member already exceeds $16/9$, and it does
more than raise the number, since it determines the constant outright. The
two are not nested, and one thing carries over from the first. Step~2 assigns nothing in it, while the family below needs Step~2 to assign a job
at every blocker, so among the instances here on which Step~2 is inactive
$16/9$ is still the largest forced load we know.

\begin{figure}[htbp]\centering
\includegraphics[width=0.95\textwidth]{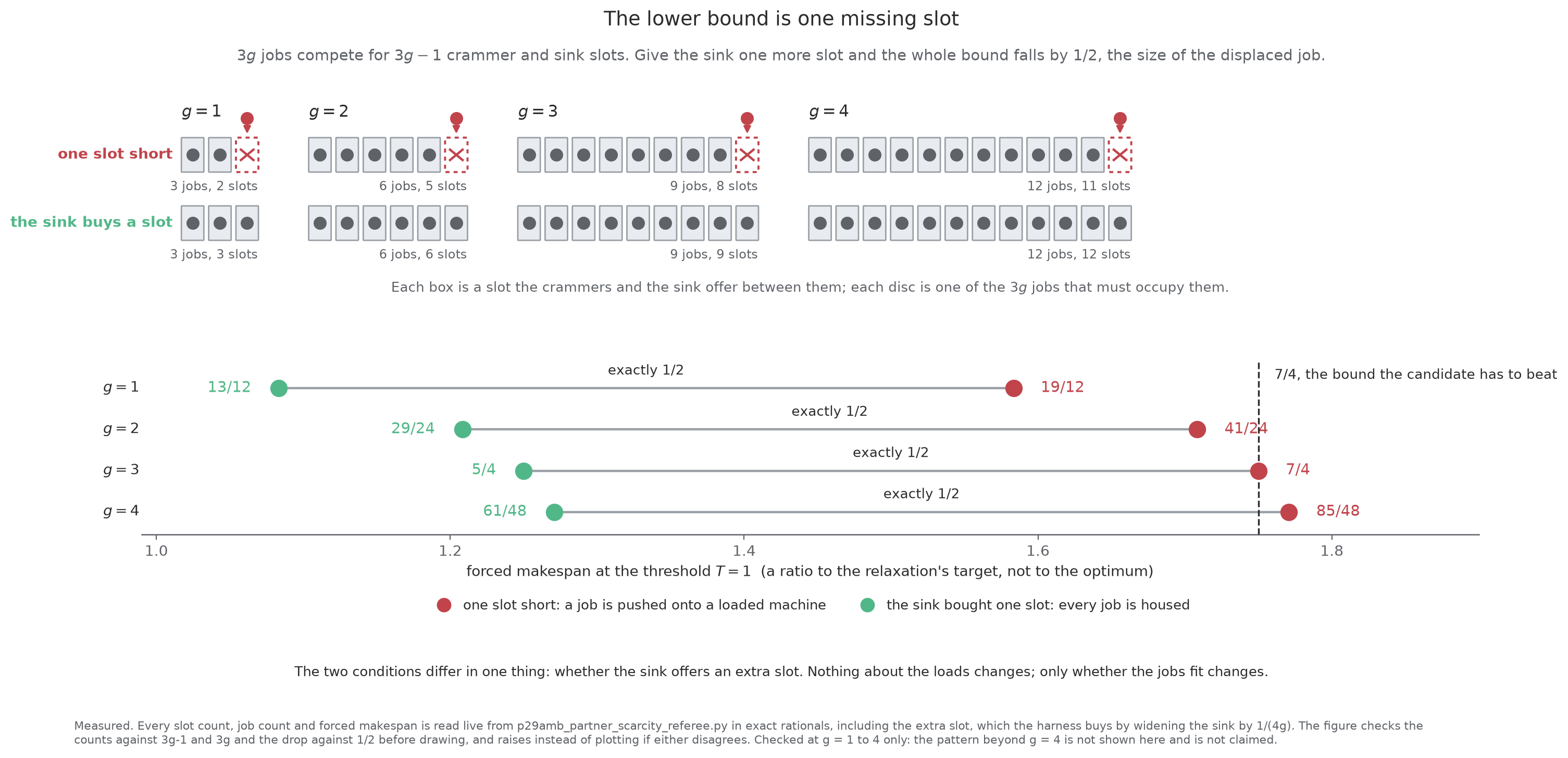}
\caption{The lower bound is one missing slot. Across $g=1$ to $4$,
$3g$ jobs compete for $3g-1$ crammer and sink slots --- deficiency exactly
one --- and one job is left with nowhere to go. Give the sink a single extra
slot and the slots equal the jobs: the forced makespan drops from
$19/12, 41/24, 7/4, 85/48$ to $13/12, 29/24, 5/4, 61/48$, a fall of exactly
$1/2$ at every $g$, the size of the displaced job.
The values are exact minima over every valid matching, drawn at $g \le 4$;
Theorem~\ref{thm:116} carries them to every $g$.}
\label{fig:pigeonhole}
\end{figure}

\begin{theorem}[the threshold-relative constant is exactly $\tfrac{11}{6}$]\label{thm:116}
Let $c^\star$ be the least constant such that, for every instance, every
$T$ at which the relaxation is feasible and every feasible $x$ at $T$,
some valid slot matching has makespan at most $c^\star T$. Then
\[
c^\star \ = \ \tfrac{11}{6} .
\]
The constant is exactly the guarantee \cite{WS16} already prove, and no
instance attains it: on every instance every valid matching finishes
\emph{strictly} below $\tfrac{11}{6}T$.

\emph{One half is imported and one is new.} The upper half --- that every valid matching finishes strictly
below $\tfrac{11}{6}T$ on every instance and every feasible $x$ --- is the
companion note's ceiling \cite{Paper1}, which is a prerequisite for this
statement and is cited rather than reproved. The lower half, the family below,
is new here. Together, the two halves prove that the value is the supremum, not merely an upper bound.

\emph{The two halves have opposite quantifier structures.} The upper half says
\[
  \forall I \ \ \forall T \ \ \forall x \in P(I,T) \ \ \exists M \in \mathcal{M}(I,x)
  : \ \ALG(I,x,M) < \tfrac{11}{6}\,T ,
\]
in words: on every instance, at every feasible target, and for every solution
Step~1 may return, \emph{some} valid matching is within the constant. The lower
half says
\[
  \forall g \ \ \exists (I_g, T_g, x_g) \ \ \forall M \in \mathcal{M}(I_g, x_g)
  : \ \ALG(I_g, x_g, M) \ \ge \ \bigl(\tfrac{11}{6} - \tfrac{1}{4g}\bigr) T_g ,
\]
in words: for every $g$ there is an instance and a feasible solution on which
\emph{every} valid matching pays almost the constant. The second is what makes
the first tight, and it is also why no selection rule helps --- it quantifies
over all matchings, so it binds whichever one a rule picks.

The lower half is a family. For every integer $g \ge 1$ there is an
instance of graph balancing on $3g + \lceil 2g/3 \rceil + 1$ machines with
$7g$ jobs, and a feasible relaxation solution at $T = 1$ with $T = 1$
least feasible, for which \emph{every} valid slot matching has makespan at
least $\tfrac{11}{6} - \tfrac1{4g}$. At $g = 1$ that instance has five
machines, seven jobs and $\OPT = \tfrac{13}{12}$, so the best-matching
variant of Section~\ref{sec:variant}, run at the least feasible threshold on
that solution, finishes at no better than
$\tfrac{19}{13} \OPT \approx 1.4615\,\OPT$ there.

The family depends on two boundary readings of Section~\ref{sec:prelim}, both
\cite{WS16}'s and \cite{EKS14}'s own: \emph{big} means size strictly above
$T/2$, so the half-sized jobs are not big; and Step~2's test is non-strict, so a
unit job at share exactly $\beta = \tfrac23$ is assigned. Under the opposite
reading of either, the construction as built fails.
\end{theorem}

\emph{Strategy.} Step~2's test is non-strict, so a big job whose share is
exactly $\beta$ is assigned to its machine and leaves one unit of base load
there. That machine keeps two slots, and a half-sized job on it reaches
only the first. The lower bound is then a count rather than a load argument:
the machines that could otherwise absorb the half-sized jobs carry one
slot fewer than there are jobs to place, so in every valid matching one of
those jobs returns to its blocker.

\begin{proof}
Fix $g \ge 1$ and set $\varepsilon = \tfrac1{2g}$ and
$\rho = \tfrac13 - \tfrac{\varepsilon}{2} = \tfrac13 - \tfrac1{4g}$. The
machines have the three roles they had in Theorem~\ref{thm:74false}'s gadget:
a blocker is the machine that ends up overloaded, the crammers are where the jobs that
would relieve it must otherwise go, and the partners and the sink absorb
share so that the relaxation stays feasible.
Take $2g$ \emph{blockers} $B_0, \dots, B_{2g-1}$, $\lceil 2g/3 \rceil$
\emph{partners} $P_p$, $g$ \emph{crammers} $C_1, \dots, C_g$, one
\emph{sink} $C_0$, and the jobs
\[
\begin{array}{llll}
Q_k & \text{size } 1 & x = \{B_k : \tfrac23,\ P_{\lfloor k/3 \rfloor} : \tfrac13\}
& k = 0, \dots, 2g-1\\[2pt]
A_k & \text{size } \tfrac12
& x = \{B_k : \varepsilon,\ C_{\lfloor k/2 \rfloor + 1} : 1 - \varepsilon\}
& k = 0, \dots, 2g-1\\[2pt]
D_k & \text{size } \rho & x = \{B_k : 1\} & k = 0, \dots, 2g-1\\[2pt]
f_j & \text{size } \varepsilon
& x = \{C_j : 2\varepsilon,\ C_0 : 1 - 2\varepsilon\} & j = 1, \dots, g .
\end{array}
\]
Three jobs sit on each blocker and one on each crammer, which is $7g$ jobs
on the stated number of machines. We call the $Q_k$ the \emph{unit jobs},
after their size. The \emph{forced load} of an instance is the makespan every
valid matching must reach on it; the point of the construction is to make that
number large, and Figure~\ref{fig:pigeonhole} shows the counting that does it.

\emph{Feasibility at $T = 1$, and least feasibility.} Blocker $B_k$ carries
$\tfrac23 + \tfrac12 \varepsilon + \rho = 1$. Each partner takes the
$\tfrac13$ shares of at most three of the $Q_k$, so its load is at most
$1$ and its big-job constraint reads $3 \cdot \tfrac13 \le 1$; that is
what fixes the number of partners at $\lceil 2g/3 \rceil$. Crammer $C_j$
carries $2 \cdot \tfrac12 (1 - \varepsilon) + 2\varepsilon^2 =
1 - \varepsilon + 2\varepsilon^2 \le 1$, and the sink carries
$g \varepsilon (1 - 2\varepsilon) = \tfrac12 (1 - \tfrac1g) \le \tfrac12$.
Every job is supported on at most two machines and every share sums to
$1$; the $D_k$ are supported on one, which the load argument below uses. Only
the $Q_k$ are big --- a job is big when its size \emph{exceeds} $T/2$, and
$A_k$ has size exactly $\tfrac12$ --- so the crammers' big-job constraints
are vacuous. Since $p_{Q_k} = 1$ and the relaxation zeroes $x_{ij}$
whenever $p_j > T$, no smaller threshold is feasible.

\emph{Step 2, and the slot structure.} $Q_k$'s share at $B_k$ is exactly
$\beta = \tfrac23$ and the test is non-strict, so $Q_k$ is assigned there
and $b_{B_k} = 1$; no other machine receives a Step-2 job. The jobs still fractional at $B_k$ are $A_k$, at share $\varepsilon$, and $D_k$, at share $1$, so $M_{B_k} = 1 + \varepsilon$ and $K_{B_k} = 2$. Since
$\tfrac12 > \rho$, the nonincreasing order puts $A_k$ on
$[0, \varepsilon)$ and $D_k$ on $[\varepsilon, 1 + \varepsilon)$: $A_k$ is
adjacent to slot $1$ of its blocker alone, and $D_k$ to both slots.
Crammer $C_j$ carries its two $A$ jobs at $1 - \varepsilon$ each and $f_j$
at $2\varepsilon$, so $M_{C_j} = 2$ exactly and $K_{C_j} = 2$; the sink
carries $g(1 - 2\varepsilon) = g - 1$ exactly, so $K_{C_0} = g - 1$. The
partners carry no fractional job at all.

\emph{The count.} Apply Lemma~\ref{lem:deficiency} with $S$ the crammer and sink slots and $J$ the $2g$ jobs $A_k$ together with the $g$ fillers $f_j$. The crammers and the sink carry $2g + (g-1) = 3g - 1$ slots between them,
so $|S| = 3g-1$; a filler has no share on any blocker, so those machines are the
only ones it is adjacent to a slot of and it has no designated machine; and
$A_k$'s only slot outside $S$ is at its own blocker $B_k$, its designated
machine. Since $|J| - |S| = 3g - (3g-1) = 1$, some $A_k$ is matched at its own
blocker, and slot $1$ is the only slot of $B_k$ it is adjacent to.

\emph{The load.} $D_k$ is supported on $B_k$ alone, so it is matched
there; slot $1$ is taken, so it takes slot $2$. Machine $B_k$ then carries
\[
b_{B_k} + p_{A_k} + p_{D_k} \ = \ 1 + \tfrac12 + \rho
\ = \ \tfrac{11}{6} - \tfrac1{4g} .
\]
That display is the three terms of the Wang--Sitters analysis --- a unit
job, a half and a third --- realized on one machine, short of the third by
$\tfrac1{4g}$, and no valid matching finishes below it.

\emph{The count's breaking point.} Two share masses are exact integers,
and both must be: if a crammer's exceeded $2$, or the sink's exceeded
$g-1$, that machine would gain a slot, the set would carry $3g$ of them,
and the pigeonhole would be gone. Conservation is what fixes the
\emph{totals}, and for the sink --- a single machine --- that is all that is
needed. The $A_k$ put $2g\varepsilon = 1$ of their share on blockers and the
remaining $2g - 1$ on crammers; the fillers put $2g\varepsilon = 1$ more
there, giving $2g$ across $g$ crammers, and their remaining $g - 1$ on the
sink. Across the crammers, conservation alone would not be enough: a
$\tfrac52$--$\tfrac32$ split between two of them conserves the same total
and yields $3 + 2 = 5$ slots where the count needs $4$. Two choices together make each crammer's mass exactly $2$: each crammer carries exactly two $A$ jobs, and exactly
one filler leaves exactly $2\varepsilon$ there. Either can be broken while the
total is conserved. Keeping two $A$ jobs on every crammer but giving one filler
its whole share to its crammer and another none is still feasible at $T = 1$, and
it splits the crammers' slot counts $3$ and $2$ against the needed $4$: at
$g = 2$ the forced load falls from $\tfrac{41}{24}$ to $\tfrac54$. So the construction needs both
choices, and naming only the $A$ jobs names only one of them.

\emph{The constant.} The forced loads increase strictly in $g$ with
supremum $\tfrac{11}{6}$, attained by no member, so no constant below
$\tfrac{11}{6}$ is admissible and $c^\star \ge \tfrac{11}{6}$.
The $11/6$ ceiling of \cite{Paper1} gives the other direction: on every
instance every valid matching --- in particular some valid matching ---
finishes strictly below $\tfrac{11}{6}T$, so $c = \tfrac{11}{6}$ is
admissible. Hence $c^\star = \tfrac{11}{6}$, and the value itself is
reached by no instance.

\emph{The optimum at $g = 1$.} There the instance is $Q_0, Q_1$ of size
$1$, $A_0, A_1$ of size $\tfrac12$, $D_0, D_1$ of size $\tfrac1{12}$ and
one filler of size $\tfrac12$, on five machines. Sending
$Q_0 \mapsto B_0$, $Q_1 \mapsto P_0$, $A_0 \mapsto C_1$,
$A_1 \mapsto B_1$, $f_1 \mapsto C_1$ and each $D_k$ to its own blocker
gives loads $\tfrac{13}{12}, \tfrac7{12}, 1, 1, 0$, so
$\OPT \le \tfrac{13}{12}$. Conversely $Q_0$ and $Q_1$ cannot both go to
$P_0$, which would load it to $2$, so some $Q_k$ shares its blocker with
$D_k$, which is supported nowhere else, for load $1 + \tfrac1{12}$. Hence
$\OPT = \tfrac{13}{12}$ and the forced $\tfrac{19}{12}$ is
$\tfrac{19}{13} \OPT$.
\end{proof}

\noindent
\emph{Three appearances of $\tfrac{11}{6}$, from one family.}
\cite{Paper1}'s $\tfrac{11}{6}$ is a supremum of $\ALG/\OPT$ for
the plain algorithm; this one is a supremum of $\ALG/T$; and
Corollary~\ref{cor:oracle}'s is a supremum of $\ALG/\OPT$ for the oracle.
Their agreement is not a coincidence of three separate
arguments: Theorem~\ref{thm:74bfalse} collapses them, since on that family every valid
matching finishes at $\tfrac{11}{6} - \tfrac1{4g}$ with $\OPT = 1$, so the
\emph{worst} matching gives the plain algorithm's supremum and the
\emph{best} gives the oracle's, while the same family read against
$T = 1$ gives $c^\star$. One family and
the $11/6$ ceiling of \cite{Paper1} determine all three. The other two
constructions keep economy rather than independence:
\cite{Paper1}'s witness needs three machines where this one needs
$5g+1$, and Theorem~\ref{thm:116} is the statement in which the threshold,
not the optimum, is the quantity the makespan is compared against.
Figure~\ref{fig:numbers} keeps the threshold-relative
constant apart from the two optimum-relative ones, and draws those two on a
single segment because they coincide. Theorem~\ref{thm:74false} is the same
pair of scales disagreeing: $44/25$ against the threshold and $88/75$
against the optimum, on one instance.

\begin{figure}[htbp]\centering
\includegraphics[width=0.92\textwidth]{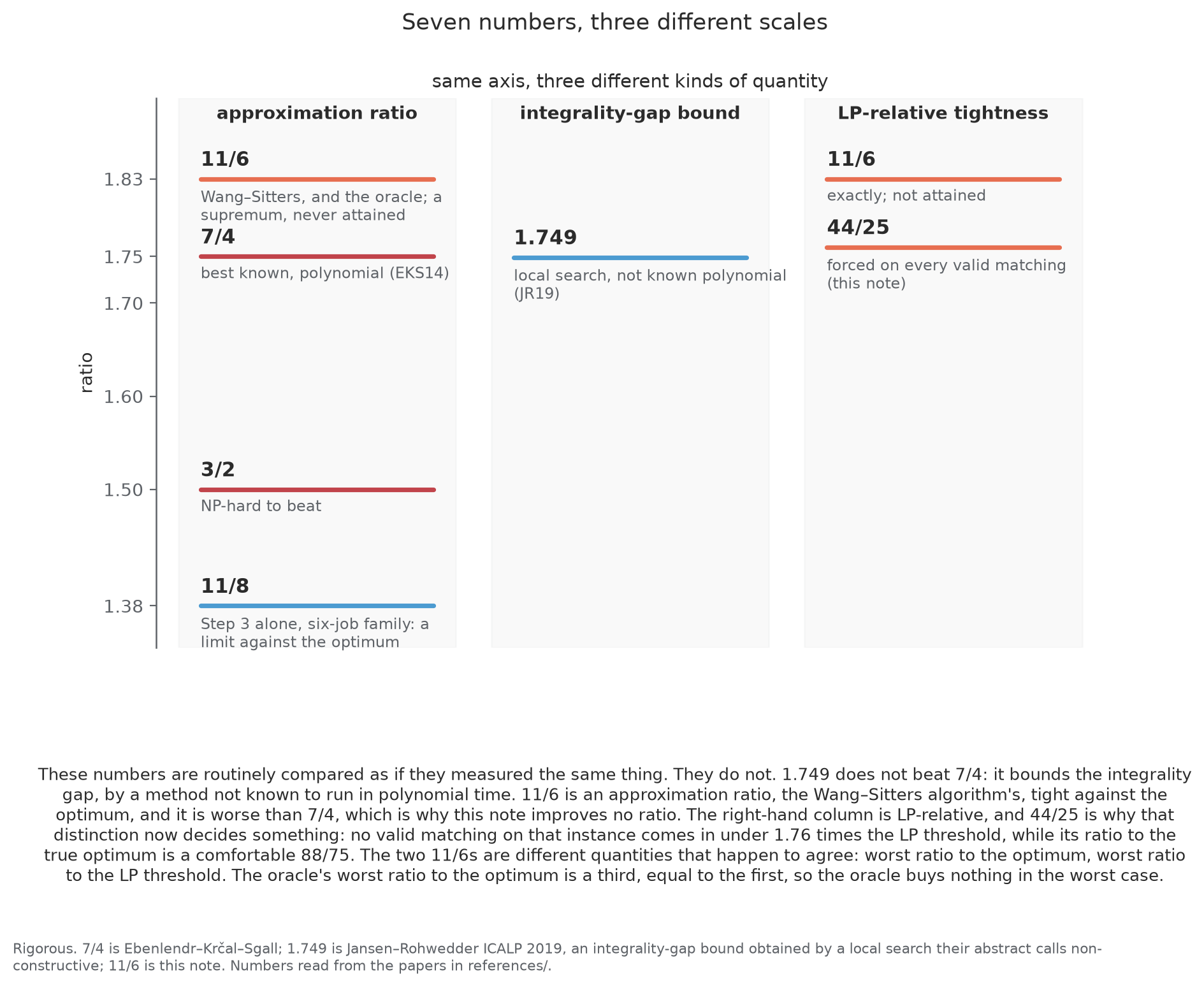}
\caption{Seven numbers that are routinely compared as though they measured
one quantity. They do not. Only the left column is an approximation ratio,
so $1.749$ does not improve on $7/4$: it bounds the integrality gap, by a
local search not known to run in polynomial time. The right column is
tightness against the relaxation, not against the optimum, which is why
this note improves no approximation ratio. That column is also where
Theorem~\ref{thm:74false} belongs: $44/25$ is forced on every valid matching
\emph{relative to the threshold}, while the same instance sits at $88/75$
relative to the true optimum. Reading one as the other is the error this
figure exists to prevent. The right column also carries $11/6$, at the same height as the left
column's: Theorem~\ref{thm:116} settles the threshold-relative constant at
that value. The two are different quantities that happen to agree, one a
worst ratio to the optimum and the other a worst ratio to the threshold, and
nothing forces two such constants to coincide; $44/25$ against $88/75$ is
the same pair of scales disagreeing on one instance. A third quantity now
joins them at $11/6$: by Corollary~\ref{cor:oracle} the best-matching
oracle's worst ratio to the optimum is that too, drawn on the same segment
as the plain algorithm's rather than beside it, since the two values
coincide. Three constants agreeing without any of them being the same
measurement is what this figure exists to separate, and here the
agreement is itself the finding, because it shows the oracle gains
nothing.}
\label{fig:numbers}
\end{figure}

\noindent
\emph{The bound is attained, and one matching suffices to show it.} The count
above forces at least one $A_k$ out of the crammer--sink set; it forces no more
than one. So the family's minimum can be read off a single matching chosen to
send exactly that one job back to its blocker and no others. This is a separate obligation
from the theorem, and the reason to discharge it is that the oracle of
Section~\ref{sec:variant} takes the \emph{minimum} over valid matchings:
``every valid matching finishes at least $\tfrac{11}{6} - \tfrac1{4g}$'' bounds
that minimum from below, and a statement about what the oracle pays needs it
bounded from above as well.

\begin{lemma}[the family is rigid, and its optimum]\label{lem:116attained}
Write $F_g = \tfrac{11}{6} - \tfrac1{4g}$ and $O_g = \tfrac43 - \tfrac1{4g}$. For
every $g \ge 1$, \emph{every} valid slot matching on the family of
Theorem~\ref{thm:116} has makespan exactly $F_g$, and $\OPT = O_g$. So the
best-matching variant of Section~\ref{sec:variant} finishes at $F_g$ there,
neither better nor worse than any other rule, and the family's ratio to its own
optimum is
\[
  \frac{F_g}{O_g} \;=\; \frac{22g-3}{16g-3} .
\]
Both quantities in that ratio are therefore pinned at every $g$ rather than
bracketed: the family admits no better matching and no better schedule.
\end{lemma}

\begin{proof}
\emph{Every valid matching finishes at $F_g$.} Theorem~\ref{thm:116}'s count
forces some $A_k$ onto its own blocker, which carries $Q_k$ from Step~2 and the
dedicated $D_k$ as well, so some machine is at $1 + \tfrac12 + \rho = F_g$. For
the converse, bound every machine. A blocker carries $Q_k$ from Step~2 and has
two slots, and the only jobs with positive share on it are $A_k$ and $D_k$, so
its load is at most $1 + \tfrac12 + \rho = F_g$. A crammer takes no Step-2 job,
has two slots, and every job with share on it --- its two $A$ jobs and its
filler --- has size at most $\tfrac12$, so its load is at most $1$. The sink has
$g-1$ slots and only fillers, of size $\varepsilon = \tfrac1{2g}$, so its load is
at most $\tfrac{g-1}{2g} < \tfrac12$. The partners hold no slot. So no machine
exceeds $F_g$ under any valid matching, and one exists by Hall's theorem.

\emph{The optimum.} A schedule putting two unit jobs on one partner has makespan
$2$, so below that each partner carries at most one and the
$\lceil 2g/3 \rceil$ partners take fewer than $2g$ of them. Some $Q_k$ is
therefore on its own blocker, which also carries $D_k$ --- allowed nowhere else
--- for a load of $1 + \rho = O_g$; hence $\OPT \ge O_g$. For $g \ge 2$ a
schedule meets it: choose one unit job for each partner and send it there,
sending that job's own $A_k$ to its blocker; send every other unit job to its
blocker and its $A_k$ to that job's crammer; send every filler to the sink,
whose share $1 - 2\varepsilon$ is positive exactly when $g \ge 2$; and every
$D_k$ to its blocker. The loads are $1$ on a partner, $1 + \rho$ or
$\tfrac12 + \rho$ on a blocker, at most $1$ on a crammer and $\tfrac12$ on the
sink, so the makespan is $O_g$. At $g = 1$ the schedule exhibited in
Theorem~\ref{thm:116}'s proof gives $\tfrac{13}{12} = O_1$. So $\OPT = O_g$ at
every $g$.
\end{proof}

\noindent
So the forced load is not a floor that some matching might exceed: on this
family every valid matching returns it, and the oracle has nothing to choose
between. That is what makes the family bind every Step-3 rule at once, and it
holds at every $g$ rather than over an enumerated range.

\begin{measurement}\label{meas:116}
Two of the theorem's hypotheses look like conventions, and each is moved off its
boundary. The unit job's share is exactly $\beta$, which a \emph{strict} Step-2
test would not assign; raising it a little and reducing $\rho$ to match makes
Step~2 assign under a strict test and leaves every other step intact. $A_k$'s
size is exactly $\tfrac12$, which reading \emph{big} as $p_j \ge T/2$ would make
big; lowering it a little puts it below the threshold under either reading. Each
perturbation lowers the forced load by exactly its own size and nothing else,
and the two hold at once: raising the share by $10^{-2}$ and lowering the size
by $5 \cdot 10^{-3}$ turns $g = 2$'s forced $\tfrac{41}{24}$ into $127/75$,
still on every valid matching. Section~\ref{sec:prelim}'s half-open adjacency is
not needed either: for $g = 2, \dots, 10$ the closed reading gives an
edge-for-edge identical slot graph, so the lower bound does not rest on the
convention its paired upper bound does.

Lemma~\ref{lem:116attained} proves rigidity and the optimum at every $g$, so
neither is measured here; what the enumeration adds is a check of both against
an independent computation. At $g = 1, \dots, 6$ the minimum over valid
matchings equals the maximum and equals $F_g$, and the exact optima are
$\tfrac{13}{12}$, $\tfrac{29}{24}$, $\tfrac54$, $\tfrac{61}{48}$,
$\tfrac{77}{60}$ and $\tfrac{31}{24}$ --- the six values $O_g$ predicts ---
giving ratios $\tfrac{19}{13} \approx 1.4615$ down to
$\tfrac{43}{31} \approx 1.3871$, largest at $g = 1$ and decreasing.
Appendix~\ref{app:computational} carries the verification record.
\end{measurement}

\section{The oracle is not a \texorpdfstring{$\tfrac74$}{7/4}-approximation}\label{sec:oracle}

The second central theorem. Section~\ref{sec:74} settles the constant against
the \emph{threshold}; against the \emph{optimum} the question is a different
one, because the two scales need not agree, and a family that is extremal on
one can be unremarkable on the other. The family below is one line away from
Theorem~\ref{thm:116}'s, and that one line moves the optimum without moving the
forced load.

Theorem~\ref{thm:116}'s family was built to pin the \emph{threshold}-relative
constant, and on that scale the optimum plays no role at all. So nothing in the
construction was chosen to make the optimum small, and one thing in it makes
the optimum large. That thing is the partners' scarcity. Each partner absorbs the
$\tfrac13$ shares of up to three unit jobs, because its big-job constraint binds
at $3 \cdot \tfrac13 = 1$, and $\lceil 2g/3 \rceil$ of them is the fewest that
will hold $2g$ such shares. The same scarcity forces a unit job back onto a
blocker \emph{in the optimum too}, which is the whole reason that family has
$\OPT = \tfrac43 - \tfrac1{4g}$ rather than $1$.
Figure~\ref{fig:familycurves} plots the two families against each other and
shows why a larger instance of the first would not have sufficed: its ratio is
largest at its first member and falls, while the modified family's rises.

Give each unit job its own partner and the optimum falls to $1$, while the forced
load does not change at all. The lower bound is a count over crammer and sink slots,
and partners appear nowhere in it.

\begin{figure}[htbp]\centering
\includegraphics[width=0.95\textwidth]{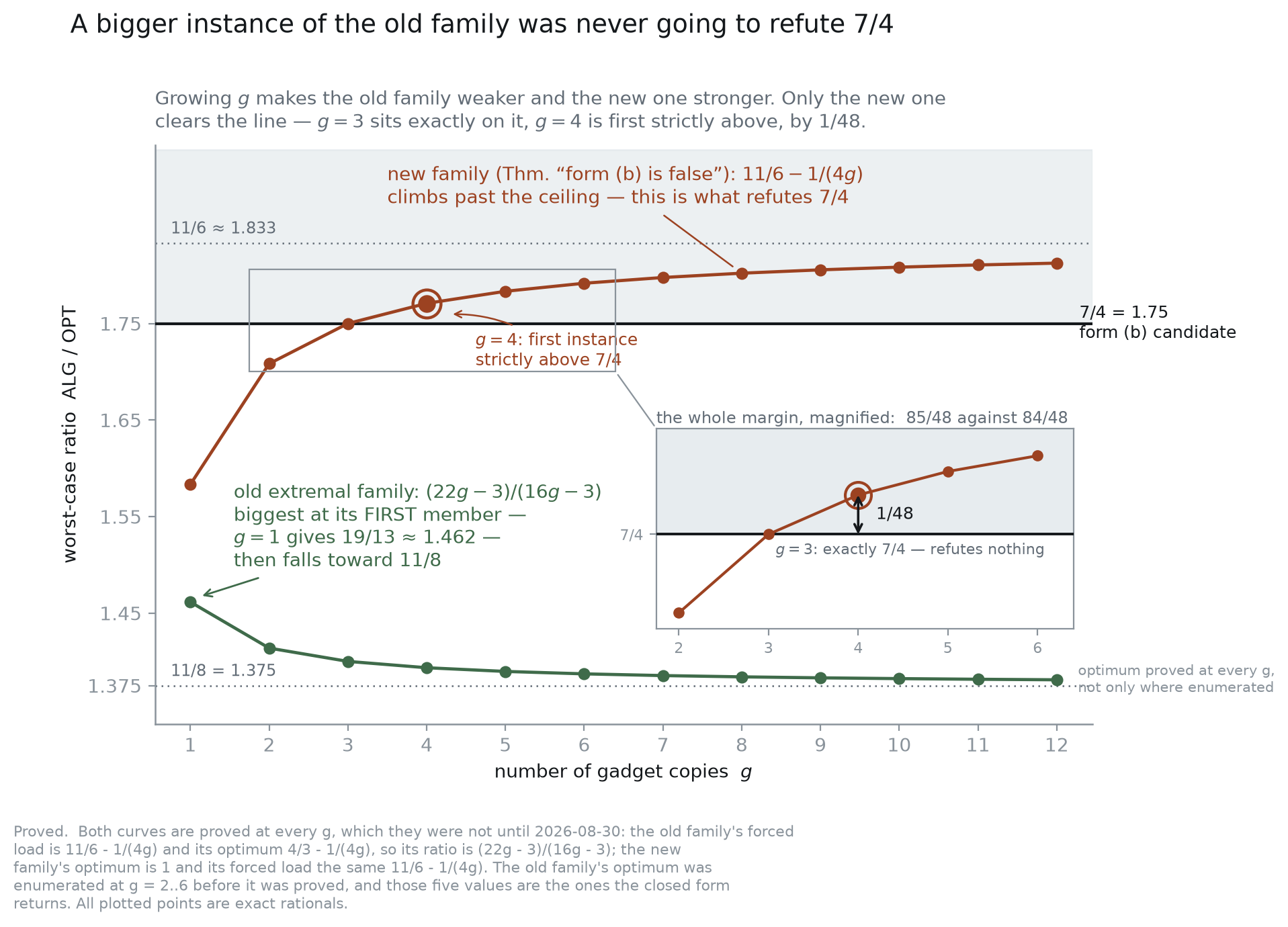}
\caption{Refuting form~(b) needed a different construction rather than a
larger instance. Theorem~\ref{thm:116}'s family has ratio $(22g-3)/(16g-3)$ (lower curve),
which is largest at its \emph{first} member, $19/13$ at $g=1$, and falls toward
$11/8$; the modified family's $\tfrac{11}{6} - 1/(4g)$ (upper curve) rises,
equals $7/4$ exactly at $g=3$, and first exceeds it at $g=4$ by $1/48$. Every
plotted point is an exact rational.
Both optima are proved at every $g$: the modified family's is $1$ by
Theorem~\ref{thm:74bfalse}, and Theorem~\ref{thm:116}'s is
$\tfrac43 - 1/(4g)$ by Lemma~\ref{lem:116attained}, which also makes that family
rigid, so the hollow markers the figure once used for unenumerated parameters no
longer mark anything.}
\label{fig:familycurves}
\end{figure}

\begin{theorem}[form (b) is false]\label{thm:74bfalse}
For every $g \ge 1$ there is an instance of graph balancing on $5g + 1$
machines with $7g$ jobs, and a feasible relaxation solution at $T = 1$ with
$T = 1$ least feasible, whose optimum is $\OPT = 1$, such that \emph{every}
valid slot matching has makespan at least $\tfrac{11}{6} - \tfrac1{4g}$ and
some valid matching has makespan exactly that. Hence the best-matching variant
of Section~\ref{sec:variant}, run at the least feasible threshold on that
solution, finishes at
\[
\Bigl(\tfrac{11}{6} - \tfrac1{4g}\Bigr)\,\OPT .
\]
At $g = 4$ that is $\tfrac{85}{48} > \tfrac{84}{48} = \tfrac74$, so form~(b)
of the optimum-relative candidate in Definition~\ref{conj:74} is
\textbf{false}. The margin is $\tfrac1{48}$, and
the preceding member $g = 3$ sits at exactly $\tfrac74$, so the family refutes
the ``at most $\tfrac74$'' reading from $g = 4$ on and no earlier. The family
stands on the same two boundary readings as Theorem~\ref{thm:116}'s, and on no
others.
\end{theorem}

\emph{Strategy.} One line of Theorem~\ref{thm:116}'s construction changes and
nothing else does; Figure~\ref{fig:onelinechange} puts the two side by side at
$g = 4$, where the change first carries the ratio past $\tfrac74$. No constraint anywhere tightens --- the partners get slacker
and no other machine is touched --- so feasibility is inherited rather than
rechecked; every step of the count mentions only machines that did not move,
so the lower bound is inherited too.
The optimum, by contrast, changes, and it changes because the modification removes the
only obstruction to a perfect assignment.

\begin{figure}[htbp]\centering
\includegraphics[width=0.98\textwidth]{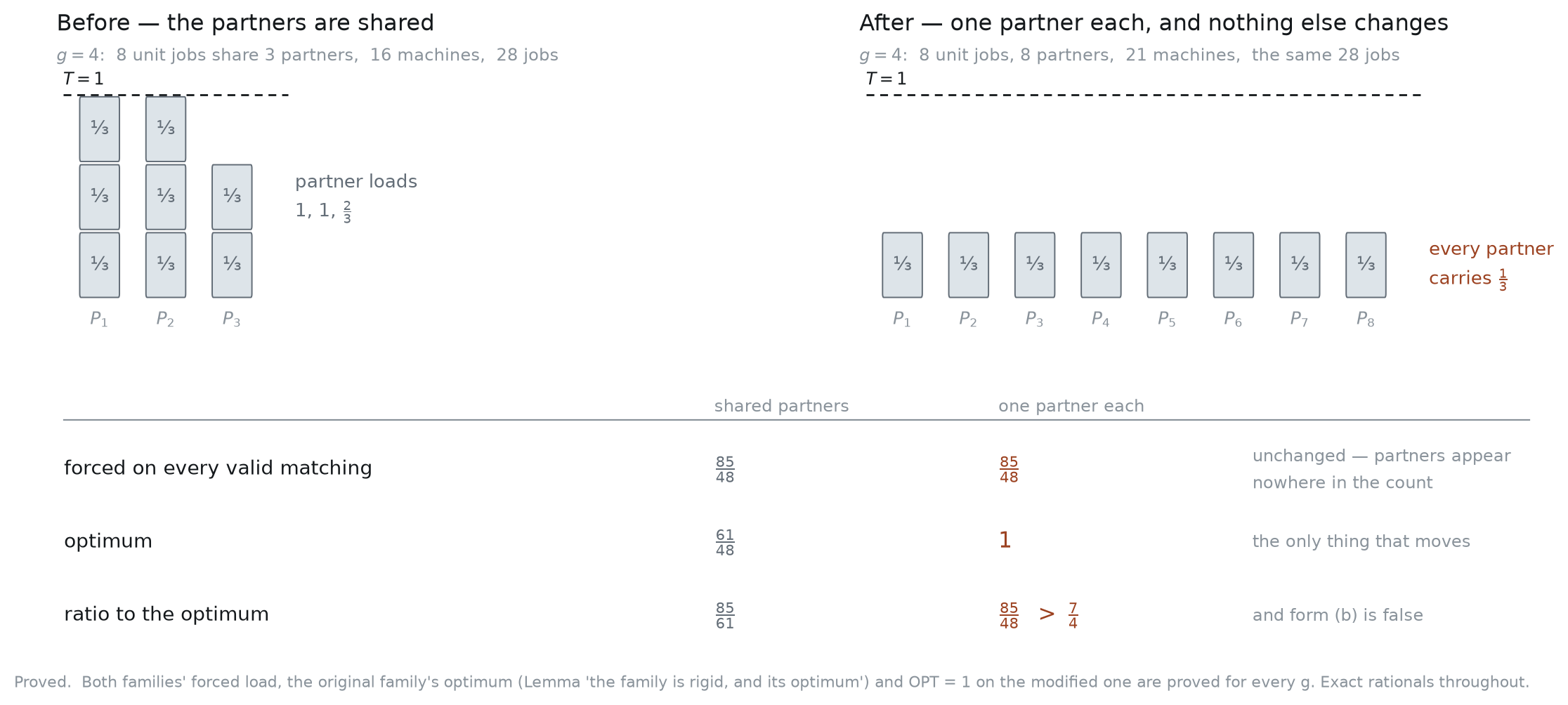}
\caption{The one line, at $g = 4$ --- the first member of the modified family
strictly above $7/4$. Only the partners move: Theorem~\ref{thm:116}'s
$\lceil 2g/3 \rceil$ of them carry up to three $\tfrac13$ shares apiece, for a
load of at most $1$ ($1$, $1$ and $\tfrac23$ at $g = 4$), and after the change each of the $2g$ unit jobs has a
partner of its own carrying one share, for a load of $\tfrac13$. No constraint
tightens, so feasibility and the lower bound are inherited; the forced load is
the same $\tfrac{11}{6} - 1/(4g)$ on both. The optimum moves, from
$\tfrac{61}{48}$ to $1$, and with it the ratio, from $\tfrac{85}{61}$ to
$\tfrac{85}{48} > \tfrac74$, set against the
original family's optimum $\tfrac43 - 1/(4g)$, which
Lemma~\ref{lem:116attained} proves at every $g$ and which was enumerated at
$g = 2$ to $6$ before it was proved; the forced load on both families and
$\OPT = 1$ on the modified one are proved for every $g$ as well.}
\label{fig:onelinechange}
\end{figure}

\begin{proof}
Take Theorem~\ref{thm:116}'s family and replace its $\lceil 2g/3 \rceil$
partners by $2g$ of them, one per unit job, so that
$Q_k$ has $x = \{B_k : \tfrac23,\ P_k : \tfrac13\}$. The blockers, crammers,
sink, the jobs $A_k$, $D_k$, $f_j$, their sizes and all their shares are
unchanged, as are $\varepsilon = \tfrac1{2g}$ and
$\rho = \tfrac13 - \tfrac1{4g}$. The machine count is
$2g + 2g + g + 1 = 5g + 1$ and the job count is unchanged at $7g$.

\emph{Feasibility at $T = 1$, and least feasibility.} Every machine other than a
partner carries exactly what it carried before, since no share on it moved, so
the blocker, crammer and sink computations of Theorem~\ref{thm:116} stand
verbatim. A partner now carries one $\tfrac13$ share instead of up to three: its
load and its big-job constraint both read $\tfrac13 \le 1$, where before each
was at most $1$: exactly $1$ for a partner holding three shares, which is
every partner only when $3 \mid g$ and not otherwise. The direction matters, not the value: no constraint tightens anywhere, which is all
that inheriting feasibility requires. Least feasibility is unchanged and does not depend on the partners:
$p_{Q_k} = 1$, the relaxation zeroes $x_{ij}$ whenever $p_j > T$, and
$\sum_i x_{iQ_k} = 1$ then fails for every $T < 1$.

\emph{Step~2 and the slot structure are untouched.} $Q_k$'s share at $B_k$ is
still exactly $\beta = \tfrac23$ and its share at its partner is $\tfrac13$ in
either regime, so Step~2 assigns $Q_k$ to $B_k$ and to no partner; no machine
receives two Step-2 jobs, since two shares of at least $\tfrac23$ would sum to
$\tfrac43$ and break the big-job constraint. A job assigned in Step~2 does not
enter the slot structure, on either of its machines, so the partners carry no
fractional job and hold no slot, exactly as in Theorem~\ref{thm:116}, where
they hold none either. Every other machine's residual mass and slot count is
the one computed there: $M_{B_k} = 1 + \varepsilon$ with $K_{B_k} = 2$,
$M_{C_j} = 2$ with $K_{C_j} = 2$, and $M_{C_0} = g-1$ with $K_{C_0} = g-1$.

\emph{The lower bound is inherited.} Theorem~\ref{thm:116}'s count is over the
$2g$ jobs $A_k$ and the $g$ fillers $f_j$, which are supported on the crammers,
the sink, and, for $A_k$, its own blocker; against the $3g - 1$ slots the
crammers and sink carry between them. Not one of those objects moved, and no
$A_k$ or $f_j$ has any share on a partner, so no partner slot could receive one
even if a partner had a slot. The count and the load computation that follows it
therefore hold word for word, and every valid matching has makespan at least
$1 + \tfrac12 + \rho = \tfrac{11}{6} - \tfrac1{4g}$. Lemma~\ref{lem:116attained}'s
matching places nothing on any partner, so it remains valid here and attains
that value.

\emph{The optimum is $1$.} For $g \ge 2$, send each $Q_k$ to its own partner,
each $A_k$ to its crammer, each filler to the sink, and each $D_k$ to its
blocker. Every job goes to a machine on which it has positive share: the
filler's sink share is $1 - 2\varepsilon$, positive exactly when $g \ge 2$. The
loads are $1$ on each partner, $1$ on each crammer --- two jobs of size
$\tfrac12$, the fillers having gone to the sink --- $g\varepsilon = \tfrac12$ on
the sink, and $\rho < \tfrac13$ on each blocker.

At $g = 1$ that assignment is unavailable, because $\varepsilon = \tfrac12$
makes the filler's sink share zero and the sink is then not in its support. One
substitution repairs it: send $f_1$ to its crammer $C_1$ and $A_1$ to its own
blocker $B_1$. Then $C_1$ carries $A_0$ and $f_1$, both of size $\tfrac12$, for
load $1$; $B_1$ carries $A_1$ and $D_1$, for load $\tfrac12 + \tfrac1{12} =
\tfrac7{12}$; each partner carries $1$; and $B_0$ carries $\tfrac1{12}$.

Either way the maximum is $1$, so $\OPT \le 1$; and $\OPT \ge p_{Q_0} = 1$
because a job of size $1$ must go somewhere. So $\OPT = 1$ at every $g$, by an
exhibited schedule and a trivial bound rather than by search, but not by the
\emph{same} schedule at $g = 1$, and saying otherwise would put a job on a
machine it is not allowed to use.

\emph{The number.} $\tfrac{11}{6} - \tfrac1{4g}$ exceeds $\tfrac74$ exactly when
$\tfrac1{12} > \tfrac1{4g}$, that is when $g \ge 4$. At $g = 4$ the ratio is
$\tfrac{85}{48}$ against $\tfrac74 = \tfrac{84}{48}$.
\end{proof}

\begin{corollary}[the oracle's constant]\label{cor:oracle}
Let the best-matching variant be run at the least feasible threshold with
$\beta = \tfrac23$. Then, over all instances and all feasible solutions Step~1
may return,
\[
\sup \ \frac{\ALG}{\OPT} \ = \ \frac{11}{6},
\]
and the supremum is attained by no instance.
\end{corollary}

\begin{proof}
For the upper bound, the $11/6$ ceiling of \cite{Paper1} puts every machine's load
strictly below $\tfrac{11}{6}T$ for any feasible $x$ and any valid matching, and
the variant takes a valid matching, so its makespan is strictly below
$\tfrac{11}{6}T$. The variant is run at $T = T_{\mathrm{LP}}$, and
$T_{\mathrm{LP}} \le \OPT$ because any integral schedule of makespan $\OPT$ is a
feasible fractional solution at $\OPT$: a machine finishing by $\OPT$ carries
at most one job of size exceeding $\OPT/2$, so the big-job constraint holds
there. Hence $\ALG < \tfrac{11}{6}\OPT$ on every instance, which also shows the
value is attained by none. For the lower bound, Theorem~\ref{thm:74bfalse}'s
family gives ratios $\tfrac{11}{6} - \tfrac1{4g}$, increasing in $g$ with
supremum $\tfrac{11}{6}$.
\end{proof}

\noindent
\emph{Three constants coincide.} Three constants in this note now coincide:
the threshold-relative constant $c^\star$ of Theorem~\ref{thm:116}, the plain
algorithm's worst ratio to the optimum in \cite{Paper1}, and the
oracle's worst ratio to the optimum here. They are three different quantities, all determined by one
family, as the passage after Theorem~\ref{thm:116} records, and
Figure~\ref{fig:numbers}'s warning against reading one as another is
unaffected. But the third of them is the damaging one: \textbf{a makespan-minimizing valid
matching does not improve the worst-case constant against the optimum}. It
improves plenty of individual instances --- that is not the claim, and the
phrase ``the oracle does not help'' would run the two together. Removing
Step~3's freedom entirely, which we do not know how to do in polynomial time,
leaves the
constant that freedom was thought to account for exactly where it was.
Removing Step~3's freedom therefore does not improve the worst-case constant.
One possibility remains at this point in the argument --- that the obstruction
needs Step~1 to return a feasible point that is not a vertex --- and
Section~\ref{sec:vertex} closes it for an adversarial Step~1; a rule that chooses
its vertex well is not covered (Section~\ref{sec:open}). That is Remark~\ref{rem:gapscope}'s point
carried from the threshold scale to the optimum scale, and no further.

\begin{measurement}\label{meas:74b}
The forced load is identical under both partner regimes at $g = 1, \dots, 7$:
$\tfrac{19}{12}$, $\tfrac{41}{24}$, $\tfrac74$, $\tfrac{85}{48}$,
$\tfrac{107}{60}$, $\tfrac{43}{24}$, $\tfrac{151}{84}$, each equal to
$\tfrac{11}{6} - \tfrac1{4g}$; only the optimum moves. The relaxation is
feasible at $T = 1$ in both regimes throughout, every constraint checked machine
by machine rather than asserted, and the modified family's optimum is decided
exhaustively over every assignment of every job to a machine in its own support
at $g \le 6$, always $1$; the largest of those enumerates
$1{,}073{,}741{,}824$ assignments. All of it is in exact rational arithmetic,
which the half-open convention requires: adjacency turns on exact equality, so
floating point is not merely imprecise here but inadmissible.

One more fact belongs here, because it bears on how the theorem should be read.
It is reported from a run made with the deposited algorithm module and a
linear-programming solver, and no deposited script regenerates it. Solving
Step~1 with a solver returns, on this family, a \emph{benign} point: at
$g = 1, \dots, 4$ the solver's own solution is one on which the best-matching
variant finishes at exactly $\OPT$, and the least feasible threshold it reports
is $T_{\mathrm{LP}} = 1$, as the theorem says. The adversarial point must
therefore be supplied explicitly, which is the same thing \cite{Paper1} records
for its own $11/6$ family. This does not weaken
Theorem~\ref{thm:74bfalse}: Step~1 returns \emph{a} feasible solution, and a
guarantee must hold for every one it might return. But it does say what the
theorem is and is not about --- the scheme as specified, not the behaviour of a
particular solver --- and, taken with Remark~\ref{rem:74bvertex}, it is a second
sign that what survives here is a question about \emph{which} feasible solution
Step~1 is allowed to return.
\end{measurement}

\section{An adversarial vertex in Step 1 buys nothing}\label{sec:vertex}

The third central theorem, and the one that closes, against an adversarial choice
of vertex, the escape route the other two leave open. Both refutations exhibit a point that is not a vertex of the
relaxation polytope, so a variant selecting a vertex solution --- which can be
done in polynomial time, and is the kind of solution a simplex implementation
returns --- might have escaped them. Against an adversarial choice of vertex it does
not: there is a family whose points are vertices from the start and whose forced
load is the same. Whether a rule that chooses its vertex well escapes is open
(Section~\ref{sec:open}).

\begin{remark}[what Theorem~\ref{thm:74bfalse} does not reach]\label{rem:74bvertex}
The point exhibited in Theorem~\ref{thm:74bfalse} is not a vertex of the
relaxation, and this has to be said here for the same reason the remark after the $11/6$ family in \cite{Paper1} says it there. Form~(b) does not restrict Step~1 ---
Section~\ref{sec:prelim}'s Step~1 returns \emph{a} feasible solution --- so the
theorem refutes the candidate guarantee as posed. A variant pinning Step~1 to a vertex is
untouched by it; Theorem~\ref{thm:pinvertex} is what reaches that case, and the
evidence below is what made the restriction look promising before it did.

The comparison goes against the refuting point. Measuring the
rank of the tight constraints against the number of variables in exact rational
arithmetic, Theorem~\ref{thm:116}'s point falls short of full rank by
$1, 5, 7, 10$ at
$g = 1, \dots, 4$ and this one by $1, 6, 9, 12$: the two are equal at
$g = 1$, and from $g = 2$ on the refuting point is \emph{further} from a
vertex, not closer. Nor does it survive being pushed to
one. Walking from it to genuine vertices by purification --- moving along a
randomly chosen direction of the null space of the tight constraints until a new
constraint binds, and repeating until that null space is empty --- reaches $180$
certified vertices at $g = 2, 3, 4$, and at $178$ of them the forced load
falls to $\tfrac43 - \tfrac1{4g}$; at the remaining two, both at $g = 2$, it
falls further still, to $1$. That null space has dimension $6$, $9$
and $12$ at those three parameters, so $180$ is a \emph{sample} of the vertices
of a face and not an enumeration of them; ``every one'' means every one reached,
and the paragraph after this says what that means.

The reason is not the one first guessed, and it says something about vertices
rather than about this family. The pigeonhole
does \emph{not} fail: at most of these vertices the count still forces a job out
of the crammer--sink set, sometimes more of them than it did before, and Step~2
still assigns a unit job to almost every blocker. The needed \emph{coincidence} fails. The forced load $\tfrac{11}{6} - \tfrac1{4g}$ is the three
terms of the Wang--Sitters analysis, $1 + \tfrac12 + \rho$, landing on one
machine, and $A_k$ reaches its blocker's first slot only because it holds the
share $\varepsilon$ there. That share exists for no other purpose: it carries
$\tfrac1{2g}$ of a job and adds one edge to the slot graph. Adjacency is
positivity here --- $A_k$ occupies $[0, \varepsilon)$ of its blocker's residual
mass, which meets the first slot whenever $\varepsilon > 0$ --- so the edge
survives any share however small, and vanishes only when the share is exactly
zero. That is what purification does to it. At every vertex reached, \emph{no}
half-sized job is adjacent to any slot of a machine carrying a Step-2 job, and
wherever the count still forces a job out the load is $1 + \rho$, with the
$\tfrac12$ missing.

So what a vertex restriction does to this construction is sharper than ``it breaks
the count''. Every family in this note creates a slot-graph edge with a vanishing
positive share, and a vertex is precisely the kind of point that cannot carry
such a share unless a tight constraint holds it there.

The complementary search is no more decisive. Asking instead how large the oracle's ratio can be at a vertex of a
\emph{random} instance --- each built around a job of size exactly $1$, so that
the least feasible threshold and the optimum are both pinned at $1$ --- gives
$\tfrac{17}{12} \approx 1.417$ as the largest, over the pool
Measurement~\ref{meas:vertexoracle}(a) reports, with most of that pool sitting
at exactly $1$. Proposition~\ref{prop:vertex}'s own witness gives
the oracle $26571/20000 \approx 1.329$, and the smaller certified vertex
recorded in the deposit gives it $39/29 \approx 1.345$. Theorem~\ref{thm:pinvertex}'s
vertices give it $\tfrac{11}{6} - \tfrac1{4g}$, which is $\tfrac{85}{48} \approx 1.771$
at $g = 4$.

Neither search establishes much. The purification search's best vertex sits $\tfrac{23}{48}$ \emph{below} $\tfrac74$
at $g = 4$ and the random search's best sits $\tfrac13$ below it, so neither is
near the bound being tested. And the random search's maximum tracks how many
vertices were drawn rather than anything structural. A second experiment,
kept separate from the first because it draws from a different pool, splits
the search into four bands of increasing instance size
(Measurement~\ref{meas:vertexoracle}(b)): the best is $\tfrac43$ in the three
smaller bands and only $\tfrac76$ in the largest, which is the band with the
fewest certified vertices, because large instances are harder to build
feasibly. No band produced $\tfrac{17}{12}$; that value comes from the single
larger pool of part~(a) and from nowhere else.
That is what a maximum over a sample does, and it establishes no plateau. From
these searches alone one could claim only that the vertex-restricted supremum
is at least $\tfrac{17}{12}$, and neither search comes near $\tfrac74$;
Theorem~\ref{thm:pinvertex} below settles it at $\tfrac{11}{6}$. The
purification search found no vertex that keeps this refutation, which does not
exclude an unsampled vertex of the same instance, and it says nothing about
whether some other vertex, in some other family, does. This note
certifies two vertices at which the best valid matching stays well below
$\tfrac74$: $26571/20000 \approx 1.329$ at Proposition~\ref{prop:vertex}'s, and
$39/29 \approx 1.345$ at the smaller one in the deposit. Two points are not a
family and determine no supremum.

The successor question was \emph{restricted to vertices of the relaxation, is
the best-matching variant a $\tfrac74$-approximation?} Theorem~\ref{thm:pinvertex}
answers it, negatively, and the searches above are why the answer is worth
stating: neither of them found the family, and one of them was built to.
\end{remark}

\begin{theorem}[a vertex restriction buys nothing in the worst case]\label{thm:pinvertex}
For every $g \ge 2$ there is an instance of graph balancing on $10g + 3$
machines and a point $x$ of the relaxation at $T = \OPT = 1$ such that $x$ is a
\emph{vertex} of that polytope and \emph{every} valid Step-3 matching has
makespan at least $\tfrac{11}{6} - \tfrac1{4g}$, with equality attained. At $g = 4$ that is
$\tfrac{85}{48} > \tfrac74$. Vertexhood is proved analytically for every $g \ge 2$, by
the spanning-tree rank argument below; the exact eliminations reported in
Measurement~\ref{meas:pinvertex} confirm it at particular $g$ and carry no part of
the proof. The family stands on the same two boundary readings as
Theorem~\ref{thm:116}'s.

Consequently, restricting Step~1 to return a vertex leaves the supremum of
$\ALG/\OPT$ at $\tfrac{11}{6}$, and the answer to
Remark~\ref{rem:74bvertex}'s question is no: the best-matching variant is not a
$\tfrac74$-approximation on vertices either.
\end{theorem}

\noindent
Theorem~\ref{thm:pinvertex} stands as stated, and Proposition~\ref{prop:compressed}
below gives a smaller version: it deletes the $5g+2$ private dump machines
introduced in the proof, reroutes their zero shares to the blockers and to one
shared machine, and keeps a vertex with the same forced makespan on $5g+2$
machines. We prove the larger construction first because its rank
argument is the more transparent one.

\begin{proof}
The instance is Theorem~\ref{thm:74bfalse}'s family with \emph{padding jobs}
added to every machine whose load row must be made tight. Because that family and
its parameters are needed by name below, we write the padded instance out in
full. Keep $\varepsilon = \tfrac1{2g}$ and $\rho = \tfrac13 - \tfrac1{4g}$ from
Theorem~\ref{thm:74bfalse}, and set
\[
  \nu \;=\; \varepsilon(1 - 2\varepsilon), \qquad
  \sigma \;=\; \frac{1 - g\nu}{2} \;=\; \frac14 + \frac1{4g} .
\]
In words: $\nu$ is the size of the pad that fills a crammer and $\sigma$ the
size of each of the two pads that fill the sink; both are chosen to be exactly
what is left over. For $g \ge 2$ all four quantities are positive and
$\nu < \varepsilon < \sigma < \tfrac12$, an ordering used twice.

The machines, $10g+3$ of them, are the $2g$ blockers $B_k$, the $2g$ partners
$P_k$, the $g$ crammers $C_j$ and the sink $C_0$ of
Theorem~\ref{thm:74bfalse}, together with $5g+2$ \emph{dump} machines, one
private to each padding job. The jobs, $12g+2$ of them, are
Theorem~\ref{thm:74bfalse}'s own $7g$ --- unchanged in size and in share ---
and $5g+2$ pads:
\[
\begin{array}{lccl}
  \text{job} & \text{size } p_j & \text{shares } x_{ij} & \text{count}\\[2pt]
  Q_k   & 1              & B_k : \tfrac23,\ P_k : \tfrac13 & 2g\\
  A_k   & \tfrac12       & B_k : \varepsilon,\ C_{\lfloor k/2\rfloor+1} : 1-\varepsilon & 2g\\
  D_k   & \rho           & B_k : 1 & 2g\\
  f_j   & \varepsilon    & C_j : 2\varepsilon,\ C_0 : 1-2\varepsilon & g\\
  R_{k,t} & \tfrac13     & P_k : 1,\ Z^P_{k,t} : 0 & 4g\\
  S_j   & \nu            & C_j : 1,\ Z^C_j : 0 & g\\
  S_{0,t} & \sigma       & C_0 : 1,\ Z^0_t : 0 & 2
\end{array}
\]
The pads sit two to a partner, one to a crammer and two on the sink: $5g+2$
pads over the $3g+1$ machines whose rows must be made tight. Each pad is
admissible on its target and on its own private dump, with share $1$ on the
first and $0$ on the second.

\emph{The reason the pad changes nothing it must not.} Both of a pad's shares sit at
their bounds, so it contributes no edge to the \emph{fractional} support. Its
size is at most $\max(\tfrac13, \sigma) \le \tfrac38 < \tfrac12$, so it is never big and Step~2
never takes it. It goes to its private dump in the optimum, which is otherwise
empty, so $\OPT$ stays $1$. And on a machine whose post-Step-2 fractional mass
is an \emph{integer}, a share-$1$ job that sorts strictly outside the existing
jobs raises the mass and the slot count by exactly one and occupies the new
slot alone, leaving every other adjacency fixed, because a common integer shift
maps slots to slots. That integrality is a real hypothesis rather than
bookkeeping: the one machine with fractional mass, the blocker at
$1 + \varepsilon$, is the one machine that receives no pad.

\emph{The reason it is a vertex.} We use the standard characterisation, stated here
because everything below is organised around it: a feasible point $x$ is a
vertex of a polyhedron exactly when the constraints holding with equality at
$x$ have rank equal to the number of variables. (If the rank is smaller there
is a nonzero $d$ in the kernel of the tight rows; every slack constraint stays
satisfied for $|t|$ small, so $x \pm td$ are both feasible and $x$ is not
extreme. Conversely a rank-$N$ tight system has $x$ as its unique solution.)
One direction of that equivalence is all we need, and it is the easy one: it
suffices to exhibit $N$ independent tight rows. In particular nothing below
depends on our having listed \emph{every} tight constraint, since extra rows
can only raise a rank and the rank is already capped at $N$; the same remark
disposes of the box constraints $x_{ij} \le 1$, which are implied by a job's
placement equality together with nonnegativity and so lie in the span of rows
we already have.

The variables are the $x_{ij}$ over \emph{allowed} pairs --- the instance's
pairs, not the support's, so a pad's zero share at its dump is a variable like
any other. Counting the table row by row,
\[
  N \;=\; \underbrace{4g}_{Q} + \underbrace{4g}_{A} + \underbrace{2g}_{D}
        + \underbrace{2g}_{f} + \underbrace{10g+4}_{\text{pads}}
    \;=\; 22g + 4 ,
\]
the $10g+4$ being two variables for each of the $5g+2$ pads.

\emph{Step 1: which constraints are tight, and in particular that every
non-dump load row is.} The padding exists to make this step hold, so we check it by computing the four loads rather than asserting it. Blocker $B_k$
carries $\tfrac23 \cdot 1$ from $Q_k$, $\varepsilon \cdot \tfrac12$ from $A_k$
and $1 \cdot \rho$ from $D_k$, that is
$\tfrac23 + \tfrac1{4g} + \tfrac13 - \tfrac1{4g} = 1$. Partner $P_k$ carries
$\tfrac13$ from $Q_k$ and $\tfrac13$ from each of its two pads, so $1$; this is
the equation the construction was built to obtain, and without the partner pads
it would read $\tfrac13 < 1$. Crammer $C_j$ carries its two half-sized jobs at
$(1-\varepsilon)\tfrac12$ each, its filler at $2\varepsilon \cdot \varepsilon$
and its pad at $\nu$, totalling
$(1-\varepsilon) + 2\varepsilon^2 + \varepsilon - 2\varepsilon^2 = 1$. The sink
$C_0$ carries the $g$ fillers at $(1-2\varepsilon)\varepsilon = \nu$ each and
its two pads at $\sigma$ each, totalling $g\nu + 2\sigma = 1$ by the definition
of $\sigma$. Every dump carries one job at share $0$, so its load is $0$ and
its row is slack. Hence exactly $5g+1$ load rows are tight.

The other tight rows are the $12g+2$ placement equalities, tight by definition,
and the $5g+2$ nonnegativities $x_{ij} = 0$, one per pad at its dump. No other
nonnegativity is tight: this needs $g \ge 2$ twice, for the filler's sink share
$1 - 2\varepsilon = 1 - \tfrac1g$ and for the crammer pad's size
$\nu = \varepsilon(1-2\varepsilon)$, both of which vanish at $g = 1$. And no
big-job row is tight: only the unit jobs $Q_k$ exceed $T/2$ --- the $A_k$ have
size \emph{exactly} $\tfrac12$, and $\rho$, $\varepsilon$, $\tfrac13$, $\nu$,
$\sigma$ are all below $\tfrac12$ --- and they sit at shares $\tfrac23$ on a
blocker and $\tfrac13$ on a partner, so the largest big-job sum anywhere is
$\tfrac23 < 1$. The tight system is therefore
$(12g+2) + (5g+2) + (5g+1) = 22g+5$ rows on $22g+4$ columns.

\emph{Step 2: the pads split off, and the split is triangular rather than
diagonal.} Let $W$ be the coordinate subspace spanned by the $10g+4$ pad
variables. For a pad $r$ with target $i$ and dump $z$, its tight nonnegativity
row is the coordinate vector $e_{(r,z)}$ and its placement equality is
$e_{(r,i)} + e_{(r,z)}$; the two together span $\{e_{(r,i)}, e_{(r,z)}\}$.
Ranging over the $5g+2$ pads, those $10g+4$ rows span exactly $W$, so they are
independent and contribute $\dim W = 10g+4$ to the rank. It is worth saying
what is and is not claimed here, because the natural shorter phrasing is false:
the pad rows meet no other column, but other rows do meet the pad columns ---
each target machine's load row carries the pad's size in the pad column. The
system is block \emph{triangular}, not block diagonal, and triangularity is
enough. Writing $\pi$ for the projection that forgets the pad coordinates,
\[
  \operatorname{rank}(\text{all tight rows})
  \;=\; (10g+4) \;+\; \operatorname{rank}\bigl(\pi(\text{the remaining rows})\bigr),
\]
since the pad rows span $\ker \pi$. In words:
the pads are nailed down by their own two constraints, so the rest of the
system may be read as though the pads were not there.

\emph{Step 3: what remains is the incidence matrix of a tree.} Let $H$ be the
bipartite graph whose nodes are the $5g+1$ non-dump machines together with
Theorem~\ref{thm:74bfalse}'s $7g$ jobs $Q_k, A_k, D_k, f_j$, with one edge per
allowed machine--job pair among them. Then $H$ has
$(5g+1) + 7g = 12g+1$ nodes and, counting $2$ edges from each $Q_k$, $2$ from
each $A_k$, $1$ from each $D_k$ and $2$ from each $f_j$,
$4g + 4g + 2g + 2g = 12g$ edges. It is connected, which we check by walking
outward from the sink: $C_0$ is joined to every filler $f_j$ and $f_j$ to its
crammer $C_j$, reaching all $g$ crammers; $C_j$ is joined to $A_{2j-2}$ and
$A_{2j-1}$ and each $A_k$ to its blocker $B_k$, reaching all $2g$ half-sized
jobs and all $2g$ blockers; and $B_k$ is joined to $D_k$ and to $Q_k$, and
$Q_k$ to its partner $P_k$, reaching the third jobs, the unit jobs and all
$2g$ partners. Every node is on that list, so $H$ is connected; a connected
graph with one fewer edge than node is a tree. Nothing in this paragraph
depends on $g$ beyond $g \ge 2$, which is the sense in which the argument is
uniform.

The remaining rows are the $7g$ placement equalities of $Q_k, A_k, D_k, f_j$
and the $5g+1$ tight load rows. Projected off the pad coordinates, the
placement row of job $j$ has entry $1$ on each edge of $H$ at $j$, and the load
row of machine $i$ has entry $p_j$ on the edge $(i,j)$; every other entry is
zero, because no pad column survives $\pi$ and no non-pad job is allowed on a
dump. So $\pi(\text{remaining rows})$ is a $(12g+1) \times 12g$ matrix indexed
by the nodes and edges of $H$. Scale each column $(i,j)$ by $1/p_j$, legitimate
because the relaxation's job sizes are strictly positive and column scaling
does not change rank; the machine rows become $0/1$, and job row $j$ becomes
$1/p_j$ on each of its edges, which a further row scaling by $p_j$ also makes
$0/1$. The resulting matrix is exactly the unsigned node--edge incidence matrix of $H$.

\emph{Step 4: a tree's incidence columns are independent.} Suppose
$\sum_e c_e (\chi_{u(e)} + \chi_{v(e)}) = 0$, the sum over the edges of a
finite tree with ends $u(e), v(e)$. A finite tree with at least one edge has a
leaf $u$; reading the $u$-coordinate of the sum gives $c_e = 0$ for the unique
edge $e$ at $u$. Delete $u$ and $e$, which leaves a tree, and repeat. Every
coefficient vanishes, so the columns are independent and the rank is the number
of edges, $12g$. (Independence is what is needed, and it is why the argument is
insensitive to the sign convention: for a bipartite graph the signed and
unsigned incidence matrices agree up to negating the rows on one side.)

\emph{Conclusion.} $\operatorname{rank} = (10g+4) + 12g = 22g+4 = N$, so the
tight constraints have rank equal to the number of variables and $x$ is a
vertex, for every $g \ge 2$. The $22g+5$ tight rows therefore satisfy exactly
one dependency up to scale, and it is visible: adding the $5g+1$ load rows and
subtracting the size-weighted placement equalities --- with each pad's
equality corrected by its own nonnegativity --- cancels every entry, since each
variable is counted once from each side. That single redundancy is the tree's
extra node, and it is why untightening one machine leaves the rank full while
untightening two does not.

\emph{The reason the bound survives.} The deficiency is unchanged at every $g$. The
pads add $g + 2$ competing jobs and $g + 2$ crammer-and-sink slots, so the
increments cancel: $4g + 2$ jobs against $4g + 1$ slots, against the unpadded
$3g$ against $3g - 1$. Theorem~\ref{thm:74bfalse}'s counting therefore applies
with the pads carried on both sides, and the bound is attained rather than
merely met.

Figure~\ref{fig:vertexpad} draws the move this proof makes against the one it
deliberately does not. The point is not walked to a vertex, because walking is
what destroys the edge the bound rests on; the instance is built so that its
point is a vertex already.
\end{proof}

\begin{figure}[htbp]
\centering
\includegraphics[width=0.97\textwidth]{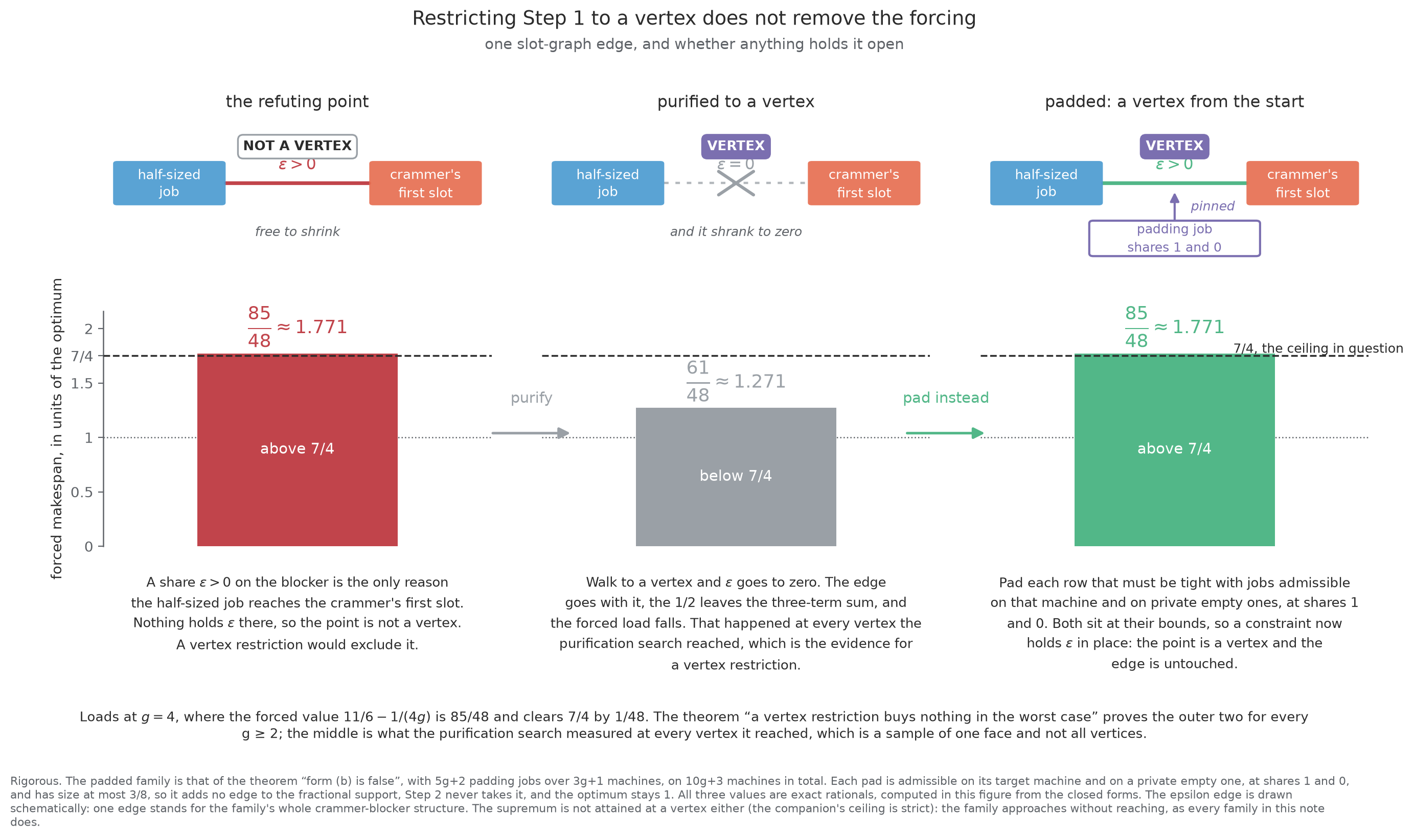}
\caption{The mechanism that defeats the vertex restriction, and the reason
the restriction looked as though it should hold. The whole construction turns on one slot-graph edge: the
half-sized job $A_k$ reaches its blocker's first slot only because it holds a
share $\varepsilon > 0$ there, and adjacency here is positivity, so the edge
survives any share however small and vanishes only when the share is exactly
zero. \emph{Left:} at Theorem~\ref{thm:74bfalse}'s point nothing holds
$\varepsilon$ where it is, which is why the point is not a vertex.
\emph{Middle:} purification drives $\varepsilon$ to zero, the edge goes, the
$\tfrac12$ leaves the three-term sum, and the forced load falls from
$\tfrac{11}{6} - \tfrac1{4g}$ to $\tfrac43 - \tfrac1{4g}$ --- $85/48$ to
$61/48$ at $g = 4$. That is real, and it is the evidence that made a vertex
restriction look promising. \emph{Right:} Theorem~\ref{thm:pinvertex} does
not push this point to a vertex. It adds padding jobs, whose two shares sit
at their bounds, so a tight constraint now holds $\varepsilon$ in place; the
point is a vertex from the start and the edge is untouched. The single edge
drawn stands for the family's whole crammer--blocker structure.}
\label{fig:vertexpad}
\end{figure}

\noindent
\emph{The checks and their range.} Every step of the proof is verified in
exact rational arithmetic against the deposited builder at $g = 2, \dots, 40$,
and all three conclusions --- vertexhood by rank, $\OPT = 1$, and the minimum
over every valid matching --- are checked outright at $46$ values of $g$ up to
$g = 225$. These are limits of the \emph{check}, not of the theorem; the
independent rank computation and the four controls are in
Appendix~\ref{app:computational}.

The bound is a minimum over \emph{every} valid matching, not the value of one
rule, so it holds for any Step-3 rule whatever, the oracle of
Section~\ref{sec:variant} included.

\noindent
\emph{Two boundary readings the construction depends on.} The half-sized jobs $A_k$ have size \emph{exactly}
$T/2$, so the construction needs \emph{big} to mean size $> T/2$ rather than
$\ge T/2$; under the second reading the crammer's big-job row would read
$2 - 2\varepsilon > 1$ and the point would not be feasible at all. And the unit
jobs sit at share \emph{exactly} $\beta = \tfrac23$, so the construction needs
Step~2's test to be non-strict; under a strict test Step~2 declines them and the
forced load falls to $\tfrac53 < \tfrac74$. Both are the published readings and
both were checked in the sources rather than inferred:
\cite[Step~2]{WS16} reads ``If $j$ is a big job and $x_{ij} \ge \beta$, then
assign $j$ to $i$'', and \cite{EKS14} defines the big edges by
$p_e > 0.5$. Theorem~\ref{thm:74bfalse}'s unpadded family stands on the same
two equalities, so neither is introduced here; what is new is only that they are
now written down.

\noindent
\emph{The bound does not depend on the half-open convention, and that is worth
saying because the $11/6$ ceiling of \cite{Paper1} does.} A pad has share exactly
$1$, so its interval begins and ends on slot boundaries, which is the
configuration the convention exists to settle. Half-open, the pad occupies its
own new slot alone. Closed, the padding \emph{creates} adjacencies that do not
exist unpadded: each crammer gains a third slot and the filler's fraction ends
exactly at its boundary. Counted correctly the deficiency is $1$ either way ---
$4g + 2$ competing jobs against $4g + 1$ crammer-and-sink slots, the pads adding
$g + 2$ of each so the increments cancel --- and the minimum over valid
matchings is unchanged, enumerated under both readings at $g \le 20$
(\texttt{p29amb\_\allowbreak crux\_\allowbreak enum.py 20}; the $g = 4$
witness is decided under both readings by
\texttt{p29amb\_\allowbreak crux\_\allowbreak halfopen.py}, and the two
counts below at $g = 2, \dots, 40$ by
\texttt{p29amb\_\allowbreak crux\_\allowbreak mscount.py}). The narrower
count $3g$ against $3g - 1$ is the unpadded one and holds only under the
half-open reading.

\noindent
\emph{Theorem~\ref{thm:pinvertex} does not contradict
Remark~\ref{rem:74bvertex}'s searches, and the difference is the point.} \emph{Purifying} Theorem~\ref{thm:74bfalse}'s own
family to a vertex destroyed the lower bound at every vertex the search reached:
$180$ certified vertices at $g = 2, 3, 4$, with forced load $\tfrac43 - \tfrac1{4g}$
at $178$ of them (at $g = 4$ the fall is from $\tfrac{85}{48}$ to
$\tfrac{61}{48}$, the $\tfrac12$ gone) and $1$ at the other two. That is a sample
of purification paths, not every vertex of the old instance. Padding does not
purify that instance; it builds a different one, in which the constraint that
pins the small share is present from the start.

\begin{measurement}[what remains measured here]\label{meas:pinvertex}
Two neighbours of Theorem~\ref{thm:pinvertex} are checked and not proved.
\emph{(a)} The perturbed family that pushes strictly inside both of
\cite{WS16}'s boundary tests --- Step~2's share at exactly $\beta$, and
\emph{big} at size exactly $T/2$ --- reaches $86407/48000$ at $g = 8$ with both
margins at $1/1000$
(\texttt{p29amb\_\allowbreak pinned\_\allowbreak vertex\_\allowbreak robust.py}).
The padding lemma does not transfer to it: the
perturbation moves the crammers and the sink off integer mass and changes which
jobs are big,
so this is measured at the parameters stated and is not claimed beyond them.
\emph{(b)} At $g = 1$ the family degenerates, two jobs vanish, and nothing is
claimed.
\end{measurement}

\noindent
\emph{Both restrictions hold on one construction.} The family's fractional
solution is a vertex, so the lower bound rests on no non-vertex point in Step~1;
and the bound is over \emph{every} valid matching, so it rests on no adverse
choice in Step~3. The $\tfrac{11}{6}$ obstruction survives both restrictions
imposed together, which is more than either of the two preceding theorems says
on its own. A vertex restriction still leaves Step~1 a choice among vertices, so
the statement is about the supremum over the vertices Step~1 may return; it does
not exclude a rule that picks a better vertex on the same instance.

\begin{proposition}[the dumps can be shared]\label{prop:compressed}
For every $g \ge 2$ the $5g+2$ private dump machines of
Theorem~\ref{thm:pinvertex} can be deleted, once each partner pad's zero share is
rerouted to its own blocker and every other pad's zero share to one shared
machine. The result is an
instance on $5g+2$ machines and $12g+2$ jobs, with $T = \OPT = 1$ and a point
$x$ that is a vertex of the relaxation, on which every valid slot matching has
makespan exactly $\tfrac{11}{6} - \tfrac1{4g}$.
\end{proposition}

\begin{proof}
Send each partner pad's zero share to its own blocker $B_k$ rather than to a
private machine, and every crammer and sink pad's zero share to one shared
machine $Z$. Nothing else changes, and no share moves off zero, so the
fractional point, the slot structure, the deficiency count and the forced load
are the ones computed in Theorem~\ref{thm:pinvertex}'s proof. The rank argument
is unchanged as well: each pad still contributes exactly two variables and two
tight rows --- its placement equality and its nonnegativity at the machine where
its share is zero --- which span those two coordinates, and projecting them away
leaves the same $(12g+1) \times 12g$ incidence matrix of the same tree. The
tight load rows are still the $5g+1$ of the blockers, partners, crammers and
sink, since $Z$ carries only zero shares and its row is slack. So the rank is
again $(10g+4) + 12g = 22g+4$, the number of variables, and $x$ is a vertex.

Only the integral schedule changes. Send $Q_k$ to $P_k$, $A_k$ to its crammer,
and $D_k$ together with both partner pads $R_{k,1}, R_{k,2}$ to $B_k$, whose
load is then $\rho + \tfrac23 = 1 - \tfrac1{4g}$; send every filler to the sink,
for $g\varepsilon = \tfrac12$, and every remaining pad to $Z$, whose load is
$g\nu + 2\sigma = 1$ by the definition of $\sigma$. The makespan is $1$, and a
job of size $1$ forces $\OPT \ge 1$.
\end{proof}

\noindent
The compression is worth stating because the machine count is what
Section~\ref{sec:open}'s rate question counts: it halves the leading constant.
The construction was also rebuilt from Section~\ref{sec:prelim}'s definitions
and checked outright at $g = 2, 3, 4$ --- feasibility, rank $22g+4$ of $22g+4$,
the minimum over every valid matching, and the schedule above --- with the
unpadded point coming back short of full rank each time, which is what the
padding is for.

\begin{corollary}[the vertex-restricted constants]\label{cor:vertexconst}
Let $c^\star_{\mathrm{vert}}$ be the least constant such that, for every
instance, every $T$ at which the relaxation is feasible and every \emph{vertex}
$x$ of the relaxation at $T$, some valid slot matching has makespan at most
$c^\star_{\mathrm{vert}} T$. Then
\[
  c^\star_{\mathrm{vert}} \;=\; \tfrac{11}{6} \;=\; c^\star ,
\]
and the same value is the supremum of $\ALG/\OPT$ over vertex Step-1 solutions,
whether Step~3 returns an arbitrary valid matching or a makespan-minimizing one.
None of these suprema is attained.
\end{corollary}

\begin{proof}
\emph{Upper.} The $11/6$ ceiling of \cite{Paper1} holds at every feasible point
and a vertex is one, so on a vertex every valid matching --- in particular a
minimizing one --- finishes strictly below $\tfrac{11}{6}T$. Running at the
least feasible threshold and using $T_{\mathrm{LP}} \le \OPT$ turns that into
$\ALG < \tfrac{11}{6}\OPT$. So $\tfrac{11}{6}$ is admissible on both scales and
reached on neither.

\emph{Lower.} Theorem~\ref{thm:pinvertex} gives, for every $g \ge 2$, a vertex
at $T = \OPT = 1$ on which \emph{every} valid matching finishes at least
$\tfrac{11}{6} - \tfrac1{4g}$; that $T = 1$ is least feasible there because the
unit jobs have size $1$ and the relaxation zeroes $x_{ij}$ whenever $p_j > T$.
Since $T = \OPT = 1$, the same number bounds both ratios, and since the bound is
over every valid matching it bounds the minimizing one too. Letting
$g \to \infty$ gives $\tfrac{11}{6}$ from below on both scales.
\end{proof}

\noindent
Measurement~\ref{meas:thrvertex} below reaches the threshold-relative half of
this corollary by enumeration on a smaller construction; the corollary does not
depend on it, and a finite enumeration could not have supplied a supremum in any
case.

\begin{corollary}[no Step-3 selection rule improves the constant]\label{cor:anyrule}
Let $S$ be any rule --- deterministic or randomized, of any computational power
--- that, given the slot structure induced by a Step-1 solution, returns a valid
matching, or a probability distribution over valid matchings. Replacing Step~3
by $S$ leaves the worst-case constant at $\tfrac{11}{6}$, against the threshold
and against the optimum alike, and whether Step~1 may return any feasible
solution or only a vertex.
\end{corollary}

\begin{proof}
On the family of Theorem~\ref{thm:pinvertex} every valid matching finishes at
least $\tfrac{11}{6} - \tfrac1{4g}$, so whatever $S$ returns does too. For a
randomized $S$ every matching in the support satisfies the bound, so it holds
with probability one and not merely in expectation. Since $T = \OPT = 1$ it is a
bound on both scales. Let
$g \to \infty$. The ceiling of \cite{Paper1} bounds any valid matching, $S$'s
output included, from above.
\end{proof}

\begin{theorem}[the constant is invariant under all four settings]\label{thm:invariance}
For $X \in \{\text{feasible}, \text{vertex}\}$ and $M \in \{\text{arbitrary},
\text{minimum}\}$, write $R_{X,M}$ for the supremum of $\ALG/\OPT$ over all
instances and all executions in which Step~1 returns a solution of kind $X$ and
Step~3 a matching of kind $M$, the scheme being run at the least feasible
threshold with $\beta = \tfrac23$. Then
\[
  R_{\text{feasible},\text{arbitrary}} \;=\;
  R_{\text{feasible},\text{minimum}} \;=\;
  R_{\text{vertex},\text{arbitrary}} \;=\;
  R_{\text{vertex},\text{minimum}} \;=\; \tfrac{11}{6},
\]
and none of the four suprema is attained.
\end{theorem}

\begin{proof}
The first is \cite{Paper1}'s $R_{\mathrm{WS}}$. The second is
Corollary~\ref{cor:oracle}. The third and fourth are
Corollary~\ref{cor:vertexconst}. Non-attainment is the strict ceiling of
\cite{Paper1} in each case, since every execution counted by any of the four is
a feasible solution together with a valid matching.
\end{proof}

\noindent
The preceding theorem states the note's conclusion, and Figure~\ref{fig:invariance} draws it: within the Wang--Sitters slot-rounding framework, neither freedom accounts for the constant when Step~1 is adversarial, whether it may return any feasible point or only a vertex, separately or together; a Step-1 rule that chooses its point well is not covered (Section~\ref{sec:open}).

\begin{figure}[htbp]\centering
\includegraphics[width=0.98\textwidth]{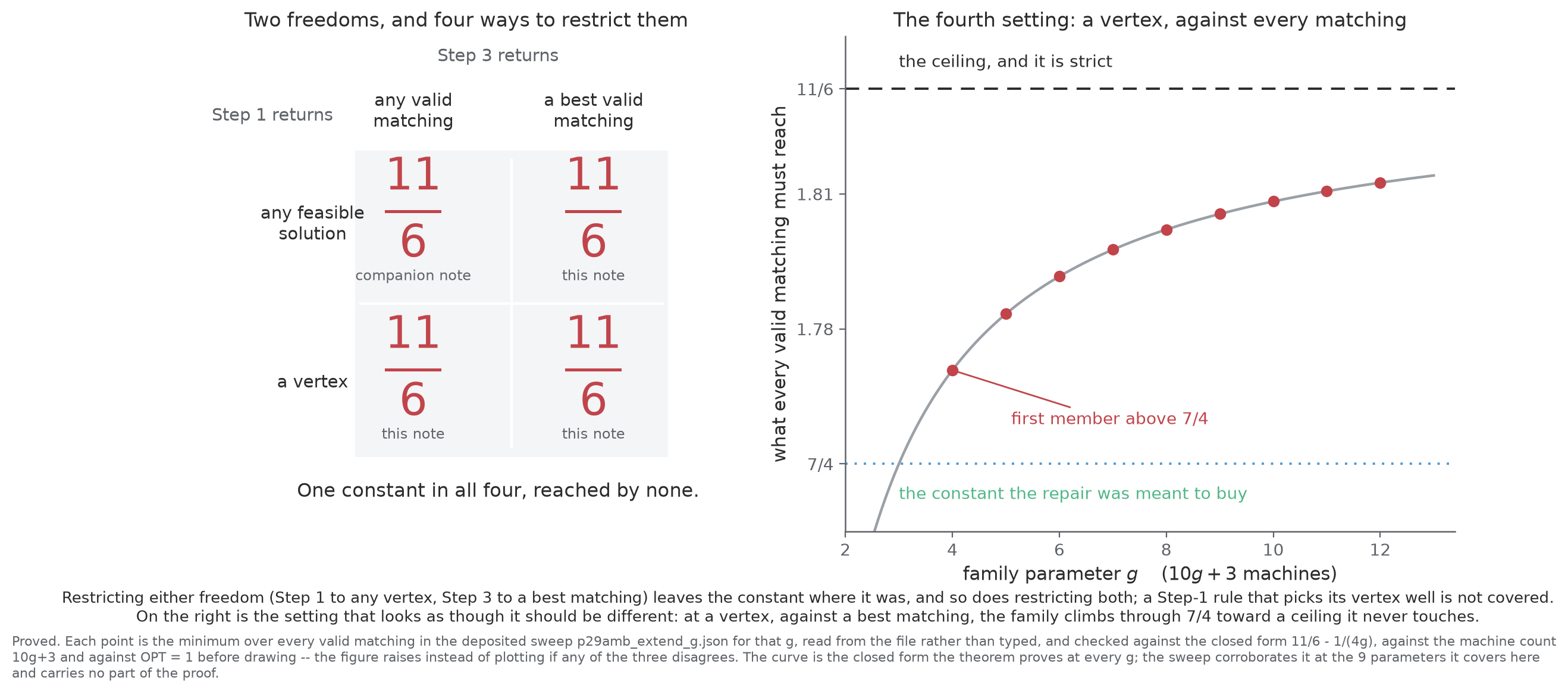}
\caption{The four settings of the two freedoms, and the one constant they
share. Step~1 may return any feasible solution or only a vertex; Step~3 any
valid matching or a best one. Every combination has supremum $\tfrac{11}{6}$,
and none of them attains it. The right panel is the setting that looks as though
it ought to be different: at a vertex, against a makespan-minimizing matching,
Theorem~\ref{thm:pinvertex}'s family is forced to $\tfrac{11}{6} - 1/(4g)$,
which rises with $g$, passes $\tfrac74$ at $g = 4$, and approaches the ceiling
without reaching it, the ceiling being strict. The points are the minima over
every valid matching recorded in the deposited certificates, at the nine
parameters $g = 4, \dots, 12$ drawn (the deposit covers $46$ values, up to
$g = 225$); the curve is the closed form the theorem proves at every
$g \ge 2$.}
\label{fig:invariance}
\end{figure}

\noindent
\emph{The reach of this result, and its limit.} The supremum of $\ALG/\OPT$
over \emph{vertices} is therefore no smaller than over all feasible points.
Restricting Step~1 to a vertex is a real restriction --- Wang and Sitters chose the Shmoys--Tardos rounding partly
because it \emph{does not} require one, calling that an advantage because the
technique then applies to any fractional assignment \cite{WS16} --- and
giving that up buys nothing here.

\begin{measurement}[the threshold-relative constant is unchanged at vertices]\label{meas:thrvertex}
Theorem~\ref{thm:116}'s family does not take the padding of
Theorem~\ref{thm:pinvertex} directly, and the obstruction is exact rather than a
failed search: three blockers \emph{share} a partner, and that partner's load
row and big-job row are the same equation, so that partner contributes one equation
for three unknowns. Granting every load row tight still leaves the active set short by
$2g - \lceil 2g/3 \rceil$, and no pad can tighten a blocker's big-job budget.
Giving each unit job its own partner removes the sharing, and then the padding
applies.

Enumerated in exact rational arithmetic at $g = 2, \dots, 8$, on
$41$ machines at $g = 8$: the minimum over every valid matching is
$\tfrac{41}{24}$, $\tfrac74$, $\tfrac{85}{48}$, $\tfrac{107}{60}$,
$\tfrac{43}{24}$, $\tfrac{151}{84}$, $\tfrac{173}{96}$, each exactly
$\tfrac{11}{6} - \tfrac1{4g}$, at points certified vertices by rank $138$ of
$138$ variables under both exact rational Gauss--Jordan and Bareiss
fraction-free elimination. Because $\OPT$ plays no part on this scale the pads
need no private machines, so the instances are half the size of
Theorem~\ref{thm:pinvertex}'s. Corollary~\ref{cor:vertexconst} already settles $c^\star_{\mathrm{vert}} = \tfrac{11}{6}$ without this measurement; the measurement adds that the smaller dump-free construction carries the same forced loads at vertices over the range enumerated. The measurement establishes no supremum of its own: a finite enumeration cannot, and
the largest value it reaches is $\tfrac{173}{96} \approx 1.802$. No general-$g$
derivation was attempted for this variant; Theorem~\ref{thm:pinvertex}'s proof
is for the optimum-relative family and does not cover it. Controls fire: one extra sink slot
collapses the lower bound by exactly $\tfrac12$, removing the partner padding costs
vertexhood, and dropping one pad leaves a vertex while dropping two does not.
\end{measurement}

\begin{proposition}[a certified vertex witness]\label{prop:vertex}
There are an instance on thirteen machines with twenty-three jobs, all
sizes integer multiples of $10^{-4}$, and a point $x$ of the relaxation at
$T = \OPT = T_{\mathrm{LP}} = 2$ which is a \emph{vertex} of that polytope, such that
the \emph{worst} valid Step-3 matching has makespan $36571/10000$. Hence,
for the scheme as \cite{WS16} state it---Step~3 taking any valid
matching---restricting Step~1 to return a vertex leaves the supremum of
$\ALG/\OPT$ in $[36571/20000,\ 11/6] = [1.82855,\ 1.8333\ldots]$, and
this instance \emph{attains} the lower end exactly, whereas the supremum
$11/6$ of \cite{Paper1} is attained by no instance.
\end{proposition}

\begin{proof}
The instance and the point are deposited in
\texttt{p29amb\_\allowbreak vertex13\_\allowbreak witness.json}, written by
\texttt{p29amb\_\allowbreak vertex13\_\allowbreak deposit.py} from the audit
artifact. Every step is finite and was checked in exact rational
arithmetic by programs sharing no code:
\texttt{p29amb\_vertex13\_verify.py} rebuilds the polytope, the rank, the
slot structure and the matching enumeration from
Section~\ref{sec:prelim} alone, and
\texttt{p29amb\_vertex13\_rank\_bareiss.py} redecides $\OPT$ by a
constraint solver and by a dynamic program over the machine graph's
components, and the rank by Bareiss elimination. Before either verdict
was believed, \texttt{p29amb\_vertex13\_poscontrol.py} ran the same
pipeline against this note's earlier $49/29$ witness --- deposited
2026-08-03 and verified then by a separate program, not a result from the
literature --- and rediscovered it:
$\OPT = 1$, rank $26$ of $26$, worst valid matching $49/29$ and best
$39/29$, under both size-tie-break orders. Logs are deposited beside
each script.
$\OPT = 2$: job~$0$ has size exactly $2$, so $\OPT \ge 2$; a schedule of
makespan $2$ exists, found by exhaustive branch and bound over all
$2^{23}$ assignments in integer arithmetic and confirmed independently by
a constraint solver. $T_{\mathrm{LP}} = 2$: the point certifies feasibility, and for
$T < 2$ the relaxation forces $x_{i0} = 0$ at every machine against
$\sum_i x_{i0} = 1$. The point is feasible, with all twenty-three job
equalities exact, every load at most $2$ and every big-job budget at most
$1$. It is a vertex because the constraints active at it---twenty-three
equalities, five tight inequalities (the loads on machines $0,1,3,4$ and
the big-job budget on machine~$4$) and nineteen coordinates at
zero---have rank $46$, the number of variables, which is the extreme-point
characterisation; the rank was computed twice, by exact Gauss--Jordan and
by fraction-free Bareiss elimination on the integer-scaled matrix. The two
runs do not use the same row set, and saying so matters because the redundancy
claim below is exactly what the difference tests: Bareiss
(\texttt{p29amb\_vertex13\_rank\_bareiss.py}) runs on all sixty-six active rows,
the nineteen coordinates at $1$ included, while the forty-seven-row rank ---
those nineteen dropped --- is \texttt{p29amb\_vertex13\_minimal\_rank.py}, by
Gauss--Jordan. Both return $46$. The nineteen coordinates equal to $1$ are
omitted as redundant---a job placed integrally has its other coordinates
among the zero rows already, so $x_{ij} = 1$ is the job equality minus
them---and including them gives the same rank, which was also computed.
The
slot structure was rebuilt from the definition in
Section~\ref{sec:prelim}, and the value by enumerating all $135$ valid
matchings, whose makespans run from $26571/10000$ to $36571/10000$. All
twenty-three sizes are distinct, so no tie-break arises.
\end{proof}

\noindent
One feature of this witness needs separating out, because it is
stronger than the statement needs. The makespan $36571/10000$ is what the
\emph{faithful} algorithm returns on the solver's own solution, and that
matching is provably the worst of the $135$. So the witness requires no
adversarial choice at Step~3 at all: it is what our implementation's
augmenting-path rule returns on this point --- Step~3 as defined takes
\emph{any} valid matching, and on this witness they run from
$26571/10000$ to $36571/10000$ --- not
what an adversary could make it do.

The witness does show that the Step-2 dodge survives the restriction. Job~$0$ has size exactly $2$ with shares $3363/10000$ and
$6637/10000$, both below $\beta = 2/3$ but the larger by a margin of only
about $0.003$, so Step~2 does not place it. The witness binds precisely
for $\beta \in (6637/10000,\ 1]$, which contains $2/3$, and collapses at
$\beta \le 6637/10000$, where job~$0$ is pinned to machine~$3$.

The proposition bounds a supremum from below and settles nothing above it:
$36571/20000 \approx 1.82855$ lies below $\tfrac{11}{6}$, and the interval
between them is closed by Theorem~\ref{thm:pinvertex} rather than here.

A smaller vertex witness exists: eight machines and fourteen jobs, at
$T = \OPT = 1$, with worst valid matching $49/29 \approx 1.6897$, active
set of rank $26$. It is dominated as a bound on the \emph{worst} valid
matching, $49/29$ against $36571/20000$, but not as one on the best: its
oracle value $39/29 \approx 1.345$ exceeds the value
$26571/20000 \approx 1.329$ of Proposition~\ref{prop:vertex}'s witness, so it is the larger of the note's two
certified-vertex values for the oracle. It is recorded in the deposit.

\section{Choosing the matching, and the cost of choosing well}\label{sec:variant}

The three theorems above are about a procedure this section defines: the
\emph{oracle}, which runs Steps 1 and 2 unchanged and returns a
makespan-minimizing valid matching in Step~3. It is stated here rather than
earlier because the theorems do not need its definition to be motivated ---
their bounds are over \emph{every} valid matching, so they cover any selection
rule --- and because what this section adds is a cost rather than a bound: the
selection is NP-hard for a supplied fractional solution.

The families of \cite{Paper1} exploit one ambiguity: Step~3 may take
\emph{any} valid matching, and each of those families makes one of them bad.
The families above go further: on them every valid matching is bad. Choosing the rounding rather than accepting any is not new:
Shchepin and Vakhania \cite{SV05} improve the Lenstra--Shmoys--Tardos
factor from $2$ to $2 - 1/m$ that way, and show theirs ``gives the best
possible approximation ratio that can be achieved using the rounding
approach.'' That phrase is narrower than it sounds, and our separate paper
\cite{Shavit26} depends on the point: the supporting argument is a closing
paragraph of \cite[\S4]{SV05} whose witness is a single job on $m$
identical machines, and it compares a non-preemptive optimum against an
optimal preemptive distribution. \cite[p.~372]{VW14} give the scope as
best possible ``among all rounding algorithms for this LP''. Their
rounding and their relaxation are both different from these, so nothing
below is subsumed; what they anticipate is the strategy. On the $11/6$ family of \cite{Paper1}, of the eight valid
matchings the worst gives $11/6 - \delta$ and the best gives $1 = \OPT$.

\begin{proposition}\label{prop:variant}
Let the \emph{best-matching variant}, which we also call \emph{the oracle},
run Steps 1 and 2 unchanged and return a makespan-minimizing valid matching in
Step~3. Its makespan is at
most $\tfrac{11}{6}T$ on every instance. On the $11/6$ and six-job families of \cite{Paper1}, run on the feasible solutions
those proofs exhibit, it equals $\OPT$. The second claim is about those
instance-and-solution pairs and not about every solution Step~1 may return;
the proof says why the difference matters.
\end{proposition}

\begin{proof}
The analysis of \cite{WS16} bounds the makespan of \emph{every} valid
matching by $\tfrac{11}{6}T$ --- their proof uses no property of the
particular matching beyond validity, as the $11/6$ ceiling of \cite{Paper1}
re-derives --- and a minimizing valid matching is in particular a valid
one. The second claim needs no measurement, and the argument is two lines.
On the $11/6$ family of \cite{Paper1}, sending $q$ to machine~$1$, $h$ to
machine~$2$ and $r$ to slot~$2$ of machine~$0$ is a valid matching of makespan
$1 = \OPT$. On the six-job family of \cite{Paper1}, $r$ is adjacent only to slots $2$ and
$3$ of machine~$0$ and so is forced there; sending $q$ to machine~$0$ and $h$ to
machine~$2$'s slot then gives maximum load $\tfrac43 + \tfrac52\delta = \OPT$.
Both hold at every parameter of their families, where the Measurements above
sampled three.

Both are statements about those solutions, and the first of them does not
survive being quantified over the instance instead. Step~1 returns \emph{a}
feasible solution, and the instance of the $11/6$ family of \cite{Paper1} has another: at the endpoint
$x_{h0} = 0$, $x_{q0} = \tfrac23 + \delta$ of the segment in the remark
following that theorem, machine~$0$'s load is exactly $1$ and the big-job
budget there is $\tfrac23 + \delta \le 1$, so the point is feasible at
$T = 1$. There $x_{q0} \ge \beta$, so Step~2 assigns $q$ to machine~$0$;
$r$ is allowed on machine~$0$ alone and is matched there; and the makespan
is $p_q + p_r = \tfrac43 - \delta$ against $\OPT = 1$, with no choice left
at Step~3 to repair it. The same endpoint that makes the $11/6$ family of \cite{Paper1}
safe --- it destroys the $\tfrac{11}{6}$ behaviour, which is all that theorem
needs --- costs the oracle its optimality here. We therefore state both claims
for the exhibited solutions and make no claim about the others, including on
the six-job family of \cite{Paper1}, where its $O(\delta)$-box
leaves Step~1 much less room.
\end{proof}

Before the hardness we record a structural fact about slot graphs, which we
have not found stated and on which the half-open convention of
Section~\ref{sec:prelim} bears.

\begin{lemma}[realizability]\label{lem:realize}
Let $x$ satisfy $\sum_i x_{ij} = 1$ for every job, and form the slot
structure. Then every job--slot edge lies in some valid matching.
Consequently an edge is avoidable by some valid matching if and only if its
job is adjacent to at least two slots.
\end{lemma}

\begin{proof}
Fix a job $j_0$ and an adjacent slot $v$; we must match the rest after
deleting both. Suppose Hall fails for some
$S \subseteq \text{Jobs} \setminus \{j_0\}$: writing $N'$ for
neighbourhoods in the deleted graph, $|N'(S)| \le |S| - 1$. We show this is
impossible. The only slot deleted is $v$, so $|N(S)| \le |N'(S)| + 1 \le
|S|$; and $v \in N(S)$, since otherwise $N'(S) = N(S)$ and the original
graph would already violate Hall. Let $I$ be the set of
\emph{all} machines carrying positive $S$-mass and $\mu_i$ that mass; then
$|S| = \sum_{j \in S} \sum_i x_{ij} = \sum_{i \in I} \mu_i$. On machine
$i$ the $S$-intervals are disjoint, of total length $\mu_i$, and lie inside
the union of the slots they meet, $V_i = N(S) \cap \mathrm{slots}(i)$;
since each slot has length at most $1$, $|V_i| \ge \sum_{v' \in V_i}
\mathrm{len}(v') \ge \mu_i$, and $|V_i|$ is an integer, so
$|V_i| \ge \lceil \mu_i \rceil$. Summing,
$|S| \le \sum_i |V_i| = |N(S)| \le |S|$, so equality holds throughout: every
slot in $V_i$ has length exactly $1$ and is entirely covered by
$S$-intervals. In particular $v$ is entirely covered. But $j_0 \notin S$
meets $v$ in positive length and distinct jobs' intervals are disjoint, so
$v$ carries length greater than $1$, a contradiction. The same chain with
no deletion proves Hall's condition in the original graph, so the lemma is
self-contained. For the consequence, a job adjacent to one slot must use it
in every valid matching, and one adjacent to two can avoid either by the
first part.
\end{proof}

\noindent
The lemma uses only $\sum_i x_{ij} = 1$ and the slot construction, not the
load constraints, not the big-job budget, and not $|M(j)| \le 2$. Its consequence
is worth stating on its own, because it says which kind of argument a lower bound
in this setting can be.

\begin{corollary}[an overload is necessarily collective]\label{cor:collective}
In a Shmoys--Tardos slot graph, a job adjacent to at least two slots is matched
to no particular one of them by every valid matching. Hence an obstruction that
forces an overload using only jobs of slot-degree at least two cannot be
certified by a single job--slot edge that every valid matching contains.
\end{corollary}

\begin{proof}
By Lemma~\ref{lem:realize} each of the job's edges lies in some valid matching,
so for each of its adjacent slots there is a valid matching that sends it
elsewhere.
\end{proof}

\noindent
So the two mechanisms available are not on the same footing.
\emph{Single-machine forcing} --- the dedicated jobs $v_k$ and $D_k$, supported
on one machine and so placed there in every valid matching, even when they are
adjacent to two of its slots --- is local. A job supported on two machines has
no unavoidable slot edge once it is adjacent to two slots
(Corollary~\ref{cor:collective}), so an obstruction built from such jobs is
\emph{collective}: each individual obligation can be avoided, and
what cannot be avoided is meeting enough of them at once. The collision forcing
of Lemmas~\ref{lem:deficiency} and~\ref{lem:collision} is one such mechanism, and
the one both constructions here use; Corollary~\ref{cor:collective} rules out
single-job certificates, not every collective mechanism besides theirs. Both counterexamples above
are built from one of each, which is why no single unavoidable edge appears in
either, and it is what
Proposition~\ref{prop:twophase} exhibits from the other side.

\begin{theorem}[minimum-makespan completion is NP-hard for a supplied
fractional solution]\label{thm:step3hard}
Given a graph balancing instance, a feasible solution $x$ of the relaxation
at target $T$, and the induced slot structure, deciding whether some valid
matching has makespan at most $T$ is NP-complete.
\end{theorem}

\noindent
\emph{A limit on this theorem, stated before the proof because its title is
the easiest thing in the note to over-read.} The hardness is in the
supplied $x$. It does \emph{not} say that the best-matching variant of
Proposition~\ref{prop:variant} is NP-hard as a function of the graph balancing
instance alone: Step~1 is not adversarial, a different Step-1 output on the same
instance may make Step~3 trivial, and no claim is made here that a polynomial
selection rule for the solutions Step~1 actually returns cannot exist. The
variant's complexity is open, and this theorem does not close it.

\begin{proof}
Membership: a valid matching is a certificate of polynomial size, and
checking validity and makespan is immediate.

For hardness, reduce from \textsc{Partition} restricted to inputs of
$r \ge 4$ integers, which is NP-complete since the inputs with $r \le 3$
integers are decided in polynomial time. Given $q_1, \dots, q_r$ summing to $2Q$, put
$C = 1 + \max_k q_k$ and, for $k = 1, \dots, r$,
\[
a_k = C(r-k+1) + q_k, \qquad b_k = C(r-k+1),
\]
scaled so the total size is $2$. Then $a_k - b_k = q_k > 0$ and
$b_k - a_{k+1} = C - q_{k+1} > 0$, so
$a_1 > b_1 > \dots > b_r > 0$. Take two machines, allow every job on both,
set $T = 1$ and $x_{j,0} = x_{j,1} = 1/2$ for every job. The loads sum to
the total size $2$ over two machines, so no target below $1$ is feasible and
$T = 1$ is least. No job is big for $r \ge 4$, so Step~2 assigns nothing: the
unscaled total size is $Cr(r+1) + 2Q$ and the largest job is
$a_1 = Cr + q_1$, so after scaling to total size $2$ at $T = 1$ a job is big
exactly when its unscaled size exceeds a quarter of that total, and
\[
Cr(r+1) + 2Q - 4a_1 \;=\; Cr(r-3) + 2Q - 4q_1 \;>\; 0
\]
because $Cr(r-3) \ge 4C > 4q_1$ when $r \ge 4$, by the choice
$C = 1 + \max_k q_k$. This is the one place the restriction to $r \ge 4$ is
used: it makes every job small, $a_1$ included. Step~2 would assign nothing here
even at $r = 3$, since every share is $\tfrac12 < \beta$.
Each job contributes fraction $1/2$, so in the nonincreasing order the pair
$(a_z, b_z)$ exactly tiles slot $z$ on each machine: slot $z$ is adjacent to
$\{a_z, b_z\}$ and to nothing else. The slot graph is therefore $r$
disjoint copies of $K_{2,2}$ and there are exactly $2^r$ valid matchings,
one per choice of which machine receives $a_z$. A machine's load is
$\sum_z (\text{the member of } \{a_z,b_z\} \text{ it takes})$, and the two
loads sum to $2$; so some valid matching has makespan at most $T = 1$ if and
only if the loads can be equalised, i.e.\ if and only if some
$S \subseteq \{1,\dots,r\}$ has $\sum_{z \in S} q_z = Q$.
\end{proof}

\noindent
Three qualifications, none of which the theorem's statement conceals. The
hardness is for a \emph{given} $x$: it does not show Steps 1 and 3 are
jointly hard, since a different Step-1 output on the same instance may make
Step~3 trivial. It is hardness of \emph{exact} optimization and does not
obstruct an approximate rule: on this family the \emph{worst} valid matching
is at most $\tfrac76 T$, since with $r = 4$ it is
$\bigl(1 + Q/(\sum_z b_z + Q)\bigr) T$, so the construction says nothing
about the $11/6$ question of Section~\ref{sec:variant} or about
the threshold-relative candidate in Definition~\ref{conj:74}, where deciding ``$\le T$'' exactly
may be hard while ``some valid matching at $\le \tfrac74 T$'' is easy. And it is
\emph{weak} NP-completeness, from \textsc{Partition}, so this reduction alone
does not exclude a pseudo-polynomial algorithm for job sizes with a bounded
common denominator, or a fixed-parameter algorithm in the number of slots per
machine; Theorem~\ref{thm:twoslot} below excludes both. A fixed-parameter
algorithm in the number of machines alone is excluded unless P${}={}$NP,
because the reduction uses only two machines, joined by $2r$ parallel jobs. The
parallel jobs are necessary for that: when each pair of machines carries at
most one job, at most $\binom m2$ jobs have a choice of machine, and trying
every choice, rejecting those that exceed the target load, and testing the rest by bipartite matching decides completion in
$2^{\binom m2}$ times polynomial time. The reduction is nonetheless
polynomial-time many-one: the numbers produced have bit-length
$O(\log r + \log \max_k q_k)$, polynomial in the input's bit-length.

\medskip\noindent
The weak reduction leaves two questions open: whether bounded denominators admit
a pseudo-polynomial algorithm, and whether few slots per machine admit a
fixed-parameter one. Both answers are no, and two slots per machine are already
enough. The obstacle in a reduction from satisfiability is that a clause machine
must accept any two of its three occurrence jobs but never all three. Shares of
$\tfrac34$, $\tfrac12$ and $\tfrac34$ on that machine do it: the first job then
meets only slot $1$, the last only slot $2$, and the middle one both.

\begin{theorem}[two slots are already hard]\label{thm:twoslot}
The problem of Theorem~\ref{thm:step3hard} remains NP-complete in the strong
sense when the eligibility graph is simple, every machine has at most two slots,
every job size is $\tfrac12$ or $1$, every share is a multiple of $\tfrac14$,
$T = 1$ is the least feasible target, and Step~2 assigns nothing. With at most
one slot per machine it is solvable in polynomial time.
\end{theorem}

\noindent
So these restrictions together are not sufficient for tractability, and one
slot per machine is. Neither a bounded common denominator nor a bounded number
of slots per machine gives a pseudo-polynomial or fixed-parameter algorithm
unless P${}={}$NP. The proof is in Appendix~\ref{app:twoslot}, and
\texttt{p29\_extensions/\allowbreak p29lever\_twoslot\_check.py} runs the
reduction on $200$ small formulas: completion at makespan $1$ exists exactly
when the formula is satisfiable.
Figure~\ref{fig:twoslot} draws the clause machine that makes the reduction work.

\begin{figure}[htbp]\centering
\includegraphics[width=0.97\textwidth]{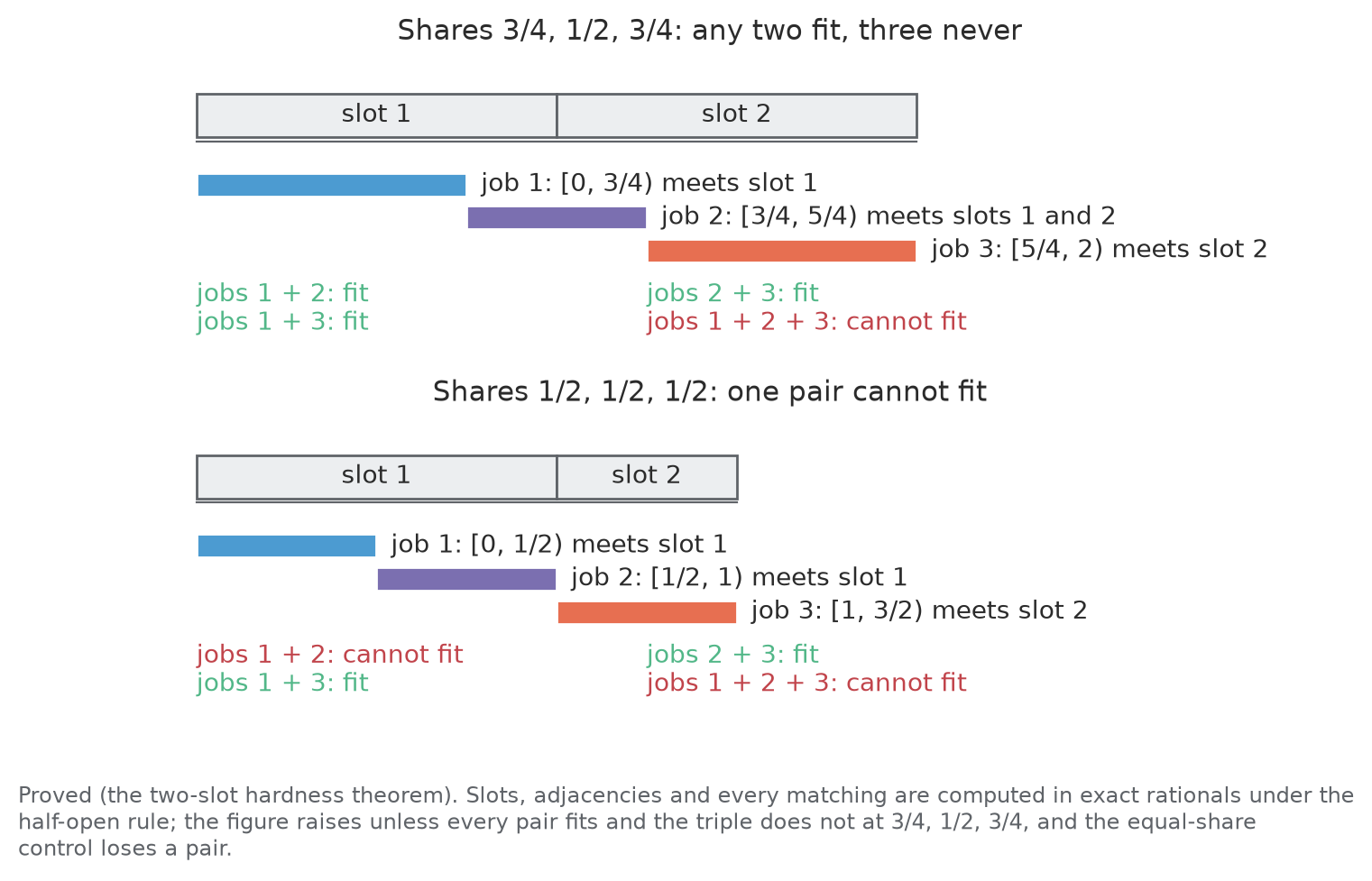}
\caption{The clause gadget of Theorem~\ref{thm:twoslot}. A clause machine has two
slots and room, at makespan $1$, for two of its three occurrence jobs of size
$\tfrac12$. \textbf{Top:} with shares $\tfrac34$, $\tfrac12$, $\tfrac34$ the
middle job's fraction meets both slots, so every pair of occurrence jobs can be
matched to the machine and the triple cannot; the clause must therefore send one
job to a literal machine, and that literal is the one made true. \textbf{Bottom:}
with equal shares the first two jobs both meet only slot~$1$, so the gadget would
forbid a pair it needs. Slots and matchings are computed under the half-open
convention.}
\label{fig:twoslot}
\end{figure}

The slot structure itself is not disposable. The measurement below exhibits one
polynomial-time completion rule that fails the $11/6$ guarantee; it does not
show that every polynomial-time rule must fail it, and it is not a complexity
claim.

\begin{measurement}\label{meas:greedy}
Replace Step~3 by the following polynomial-time completion: take the
unassigned jobs in nonincreasing order of size and put each on the least
loaded machine in its fractional support. On a $12$-machine, $14$-job
instance found by plateau walk, this rule returns makespan $4.054$ against
$\OPT = 2$, a ratio of $2.027$, so it does not meet the $11/6$ guarantee. That
instance is the one shipped in
\texttt{p29amb\_ratio\_\allowbreak audit\_\allowbreak smart\_\allowbreak walk\_\allowbreak 8000\_\allowbreak s30.json}.
The search is seeded from its own directory's contents, so rerunning it
overwrites that file with a different instance --- run from
\texttt{osf\_\allowbreak p29lever\_\allowbreak core.zip}, $12$ machines, $17$ jobs, ratio $2.2544$ ---
and either instance exceeds $11/6$;
the rule is \texttt{ws\_\allowbreak schedule\_\allowbreak smart} in mode
\texttt{"greedy"}, reached with \texttt{-\/-smart} on the command line.
\end{measurement}

\noindent
We state this as a measurement and not as a proposition on purpose. The
greedy rule is run on whatever feasible $x$ the solver returns, so the
number is solver-dependent, and an existence claim about an instance
should not be.

\noindent
The measurement is about \emph{this} completion rule only, and does not
extend to any completion that ignores the slot structure: a
completion may ignore the slot structure and still be good, for instance
one that solves the residual assignment problem exactly. With the slot cap removed, a machine
can receive arbitrarily many jobs in which it holds only a tiny share,
which is the failure mode Shmoys--Tardos rounding exists to prevent.

\section{Above \texorpdfstring{$\tfrac74$}{7/4}, a machine has one of two shapes}\label{sec:classify}

Both forms are refuted, so what follows is a classification rather than a route
to $\tfrac74$: given that the bound fails, which machines can fail it, and what
would have had to hold for them not to.

\noindent
The next theorem gives two conditions on where the matching puts a single job
that together deliver $\tfrac74$, and shows that in any valid matching a
machine above $\tfrac74$ has one of two shapes.

\begin{theorem}[reduction]\label{thm:reduction}
Let $x$ be feasible at $T = 1$ and let a valid matching satisfy
\begin{enumerate}
\item[(R1)] every big job that is matched (rather than assigned in
Step~2) sits on a machine where its fractional share is at least $1/2$,
and
\item[(R2$'$)] on every machine that received a Step-2 job of size
$b_i > 3/4$, the job matched to slot $1$, if any, has size at most
$3/4 - b_i/3$.
\end{enumerate}
Then its makespan is at most $7/4$. Consequently, in \emph{any} valid
matching, a machine of load above $7/4$ is either
\begin{itemize}
\item[(O1)] a machine with no Step-2 job whose matched big job has share
below $1/2$: its load is at most $2 - x_q/2$; or
\item[(O2)] a machine whose Step-2 job has size $b_i > 3/4$: its load
is at most $3/2 + b_i/3$, exceeding $7/4$ by at most $1/12$.
\end{itemize}
\end{theorem}

\begin{proof}
A machine carrying no slot job has load $b_i \le 1 < 7/4$. That is the
case the proof of the $11/6$ ceiling in \cite{Paper1} disposes of before its case
split, and it is easily dropped: the three cases
(N0), (N1), (S) are exhaustive and disjoint only among the machines that
do carry a slot job. On those, each case is already bounded above.
Case (N0) --- no Step-2 job and a small or absent slot-$1$ occupant ---
gives load $< 3/2 \le 7/4$ unconditionally.
Case (N1) gives load $< 2 - x_q/2$, which is at most $7/4$ exactly when
$x_q \ge 1/2$, i.e.\ under (R1).
Case (S) gives load $< 3/2 + b_i/3$, which is at most $7/4$ when
$b_i \le 3/4$. For $b_i > 3/4$ write $t_1$ for the size of the job
matched to slot $1$. Bounding slot $1$ by $t_1$ and slot $z \ge 2$ by
$s_z$, and using $\sum_{z \ge 2} s_z \le L'_i - w_{K_i}$ from
the chain bound of \cite{Paper1}, $w_{K_i} > 0$ and $L'_i \le 1 - \tfrac23 b_i$,
\[
\text{load} \ \le\ b_i + t_1 + L'_i - w_{K_i}
\ <\ 1 + \tfrac{b_i}{3} + t_1 \ \le\ \tfrac74
\]
by (R2$'$).
So every machine is at most $7/4$, which is the first claim. For the
second, a machine above $7/4$ must fail the case bound it falls under:
it carries a slot job, since otherwise its load is at most $1$; it
cannot be (N0); if it is (N1) it must have $x_q < 1/2$, which is
(O1); if it is (S) it must have $b_i > 3/4$, which is (O2), and the
excess over $7/4$ is at most $(3/2 + b_i/3) - 7/4 \le 1/12$ at
$b_i \le 1$.
\end{proof}

\noindent
Figure~\ref{fig:twoshapes} draws the two ceilings against the parameter each one
depends on.

\begin{figure}[htbp]\centering
\includegraphics[width=0.95\textwidth]{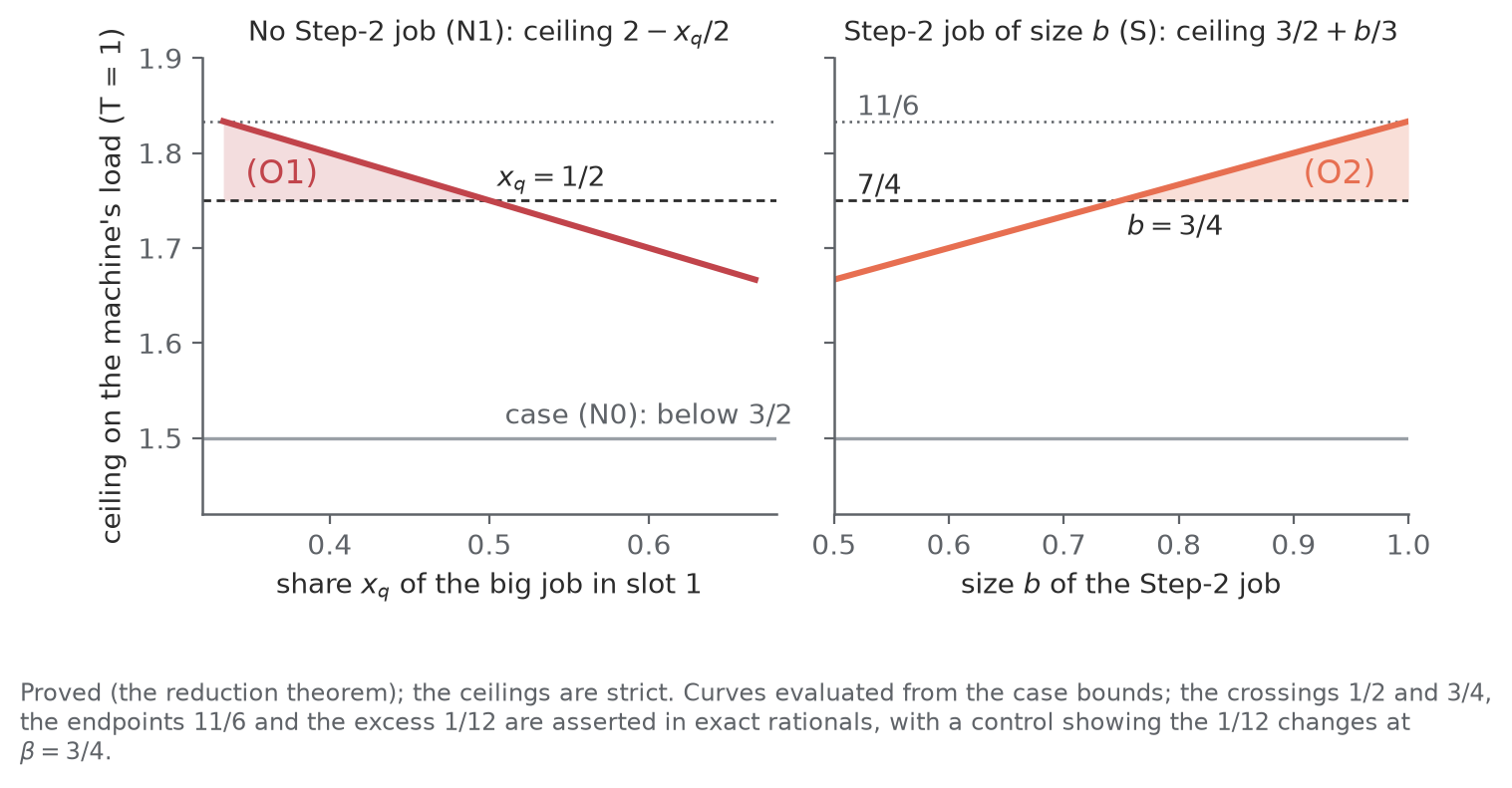}
\caption{The two shapes of Theorem~\ref{thm:reduction}, at $T = 1$. Each curve is
the strict ceiling one case of the proof places on a machine's load. Case (N0)
stays below $\tfrac32$. Case (N1) rises above $\tfrac74$ only when the matched big
job's share $x_q$ is below $\tfrac12$, which is shape (O1); case (S) rises above
it only when the Step-2 job's size $b$ exceeds $\tfrac34$, which is shape (O2).
Both ceilings end at $\tfrac{11}{6}$, at $x_q = \tfrac13$ and at $b = 1$, so
neither shape overshoots $\tfrac74$ by more than $\tfrac1{12}$.}
\label{fig:twoshapes}
\end{figure}

\begin{remark}[what the reduction achieves]
(R2$'$) replaces the more obvious hypothesis, that every
machine with $b_i > 3/4$ receive matched load at most $7/4 - b_i$. That
hypothesis is the theorem's own conclusion on exactly the machines where
the conclusion is in doubt, so the reduction it supports is close to
circular: it reduces the threshold-relative candidate in
Definition~\ref{conj:74} to (R1) together with itself
on the hard machines. (R2$'$) is a condition on where the matching puts a
single job, of the same kind as (R1), and it implies the old hypothesis
through the chain bound: $t_1 \le 3/4 - b_i/3$ gives matched load
$< t_1 + 1 - \tfrac23 b_i \le 7/4 - b_i$. It is therefore the stronger
hypothesis and the theorem under it the weaker statement. The reduction makes both hypotheses structural, so the candidate guarantee becomes a question about which slots a valid matching can be
made to use. The reduction makes no progress on satisfying them:
Proposition~\ref{prop:twophase} shows (R1) alone can fail for every valid
matching, and nothing here shows the shapes (O1) and (O2) can be routed
away \emph{at $\tfrac74$}, a constant the phrase needs, since
Theorem~\ref{thm:74false} has already shown they cannot always be, and
Theorem~\ref{thm:116} puts the least constant at which they always can be
at $\tfrac{11}{6}$.

The direction of the implication has to be said, because the two shapes are
threshold-relative while form~(b) is not. Routing both away is \emph{sufficient}
for form~(b), since a valid matching with no machine above $\tfrac74 T$ is
form~(a)'s conclusion and form~(a) implies form~(b); it is \emph{not necessary},
and Theorem~\ref{thm:74false}'s instance settles that on its own, since there
every valid matching leaves a machine above $\tfrac74 T$ --- an (O1), shape (O2)
being unavailable where Step~2 fires on nothing --- while $\ALG/\OPT = 88/75$
sits below $\tfrac74$ with $173/300$ to spare. The obvious strengthening that
would have supplied both hypotheses at once --- that some valid matching caps
every machine's slot-$1$ occupant at its residual budget --- is false: a
$14$-machine instance whose feasible region is a single point refutes it in exact
rational arithmetic, and a separate $9$-machine instance refutes its restriction
to the case with no Step-2 jobs.

\end{remark}

\begin{remark}[assuming (O2) away does not help the scheme]
\label{rem:o2}
The obvious way to simplify is to assume it away: suppose every Step-2 job
has size at most $3/4$, so shape (O2) cannot occur and only (O1) is left.
That hypothesis does not help the scheme at all as
\cite{WS16} state it. The $11/6$ family of \cite{Paper1} has no Step-2
job at all---Step~2 never fires on it, which the shares $\tfrac13 + \delta$ and $\tfrac23 - \delta$
were chosen to ensure---so it satisfies the
hypothesis vacuously while driving $\ALG/\OPT$ to $11/6$. Proposition~\ref{prop:twophase}'s
instance has no Step-2 job either, which is why that
proposition's refutation of (R1) survives the hypothesis too. So excluding
(O2) can only ever help the best-matching variant of
Section~\ref{sec:variant}, never the polynomial-time algorithm, and any
statement of that shape must say which of the two it is about.
\end{remark}

\begin{proposition}[the two-phase strategy fails]\label{prop:twophase}
Condition (R1) is not always satisfiable: there is an instance on seven
machines with six jobs and a feasible $x$ for which no valid matching
places every matched big job on a machine where its share is at least
$1/2$, although every valid matching has makespan $1$.
\end{proposition}

\begin{proof}
Machines $A_1,A_2,A_3$, $D$, $E_1,E_2,E_3$. Three big jobs $a_k$ of size
$1$ with $x_{A_k a_k} = 3/5$, $x_{E_k a_k} = 2/5$; three jobs $u_k$ of
size $1/2$ with $(x_{A_k u_k}, x_{D u_k})$ equal to $(2/5, 3/5)$,
$(2/5, 3/5)$, $(1/5, 4/5)$. The loads are $4/5$, $4/5$, $7/10$, $1$,
$2/5$, $2/5$, $2/5$, and the big fraction on any machine is $3/5$ or
$2/5$, so $x$ is feasible; every big share is below $\beta = 2/3$, so no
job is assigned in Step~2. Machine $D$ has mass $2$, hence two slots.
Machine $A_k$ has mass $1$, $1$, $4/5$ for $k = 1,2,3$, hence
$\lceil M_{A_k}\rceil = 1$ slot in every case. Both $a_k$'s and $u_k$'s
intervals are contained in that one slot; they are not in its interior
---$a_k$'s interval starts at the slot's left endpoint, and the interval of $u_k$ ends at its right endpoint---but
containment is all that is used, and it makes each of $a_k, u_k$ adjacent
to slot $1$ of $A_k$ and to no other slot of $A_k$. The mass exactly $1$
on $A_1$ and $A_2$ is the case that would break the argument: any more
fraction there and the machine has a second slot for $u_k$ to use.
Condition (R1) pins each $a_k$ to $A_k$, its only share of at least
$1/2$ being $x_{A_k a_k} = 3/5$, and that consumes $A_k$'s only slot. So
under (R1) the three jobs $u_1,u_2,u_3$ have available slots only on
$D$, of which there are two. Hall's condition fails, and no valid
matching satisfies (R1). Enumeration confirms it: $21$ valid matchings,
none satisfying (R1), every one of makespan $1$.
Figure~\ref{fig:twophase} draws the instance and one of those matchings.
\end{proof}

\begin{figure}[htbp]\centering
\includegraphics[width=0.97\textwidth]{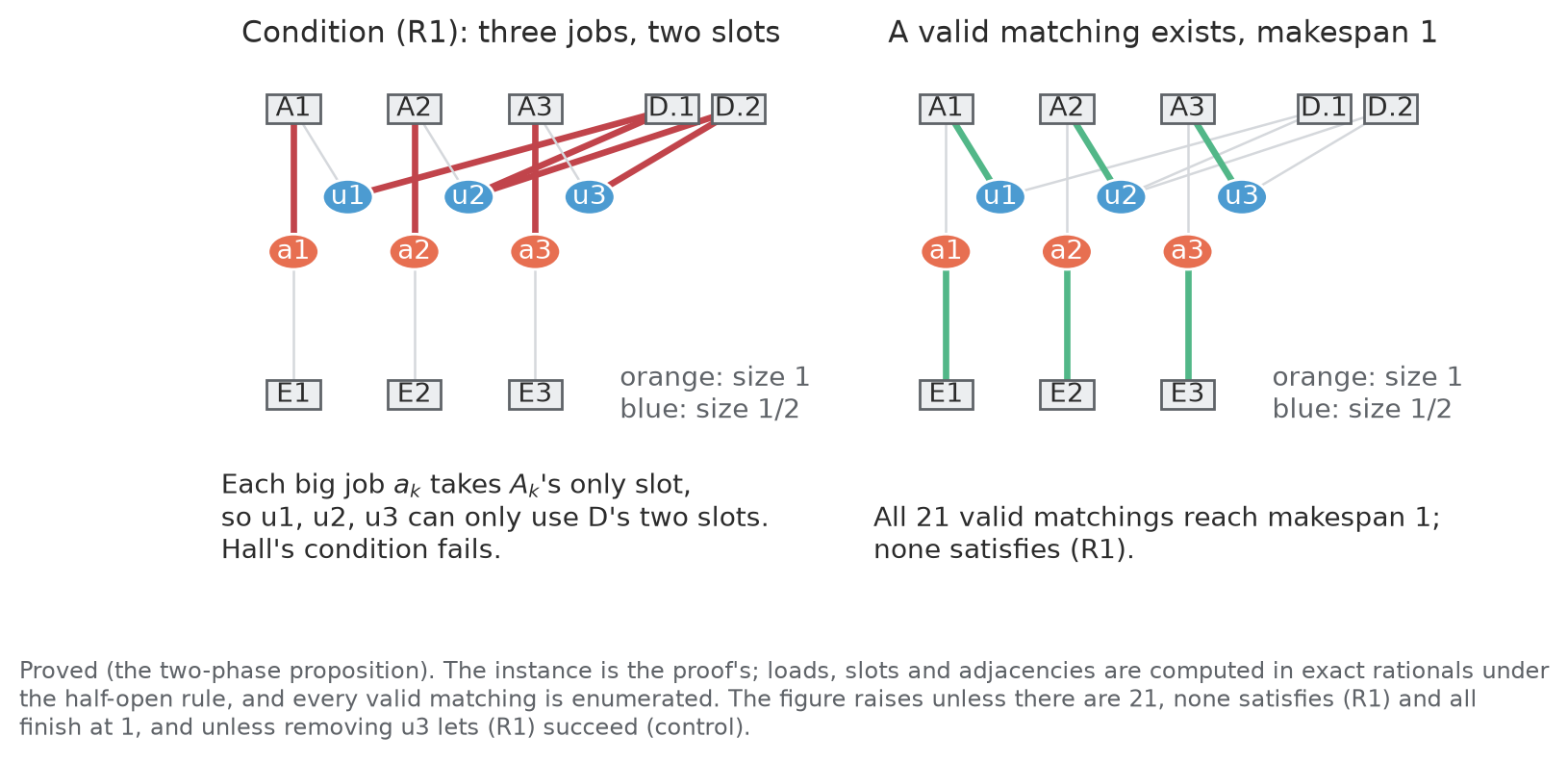}
\caption{The instance of Proposition~\ref{prop:twophase}. Grey lines join each job
to the slots its fraction touches; machines $A_1, A_2, A_3$ and $E_1, E_2, E_3$
have one slot each and $D$ has two. \textbf{Left:} condition (R1) pins each big
job $a_k$ to $A_k$, which uses up that machine's only slot, so the three jobs
$u_1, u_2, u_3$ can use only $D$'s two slots and Hall's condition fails.
\textbf{Right:} one of the $21$ valid matchings, which sends every big job to its
lighter side. All $21$ finish at makespan $1$ and none satisfies (R1), so failing
the repair condition is not failing the schedule.}
\label{fig:twophase}
\end{figure}

\noindent
So a proof of the threshold-relative candidate in Definition~\ref{conj:74}
could not proceed by
first pinning every big job to a machine where its share is at least $1/2$
--- condition~(R1) --- and then matching what remains: it would have
to be an exchange argument in which a big job pushed to its lighter side
draws slack from the deficiency that pushed it. The proposition refutes that one pinning rule, not pinning as a strategy; the proposition is stated in that
weak form and nothing here strengthens it.
Theorem~\ref{thm:74false} has since closed form~(a) outright, so this
proposition no longer constrains a live proof attempt. We keep it for the
obstruction it isolates --- that the deficiency and the slack that would
cover it sit on different machines --- which is what the counterexample turns
into a pigeonhole. It is \emph{not} a constraint on the
vertex-restricted question of Remark~\ref{rem:74bvertex}, tempting as that
reading is:
the point exhibited above is not a vertex --- its twelve coordinates are all
positive and the constraints active at it, six job equalities and machine
$D$'s load, have rank $7$ --- and by the reading this note applies elsewhere,
a non-vertex point bears on nothing vertex-restricted.

\section{The family the \texorpdfstring{$\tfrac74$}{7/4} bound is known to be tight against}\label{sec:threepath}

\paragraph{The family the $7/4$ bound is known to be tight against.}
The remaining question about $\tfrac74$ concerns not our families but the one
already in the literature: the three-path family of \cite{EKS14} named in
Section~\ref{sec:prelim}, which witnesses the relaxation's integrality gap of
$\tfrac74$. The next three results generalize it, bound it strictly below
$\tfrac74$, and say what that does and does not settle. The generalization is
ours; the family is theirs.

\begin{lemma}[the family's optimum]\label{lem:familyopt}
Let $F(r,k,S,c)$ be the generalized three-path family: machines $u, v$ and
inner vertices $v^p_1, \dots, v^p_{2k}$ for paths $p = 1, \dots, r$; jobs
$e^p_t$ joining $v^p_{t-1}$ to $v^p_t$, with $v^p_0 = u$ and
$v^p_{2k+1} = v$, of size $1$ for $t$ odd and $S$ for $t$ even; a dedicated
singleton of size $c$ on every machine. For $r \ge 3$, $k \ge 1$,
$0 < S \le 1$ and $c \ge 0$, $\OPT(F(r,k,S,c)) = 1 + S + c$.

\emph{Both bounds are what Section~\ref{sec:prelim} requires.} Job sizes
there are positive, so $S > 0$. (A size-zero job would still carry fractional
mass: mass is $\sum_j x_{ij}$, not $\sum_j p_j x_{ij}$.) At $c = 0$ the dedicated
singletons are \emph{absent} rather than present at size zero, which is the
reading used throughout and the one under which every job size stays positive.
\end{lemma}

\begin{proof}
($\le$) Fix $t_p \in \{1, \dots, 2k\}$ for each path and orient $e^p_s$ to
its right endpoint for $s \le t_p$ and to its left endpoint for
$s > t_p$. Then $u$ and $v$ receive no edge and carry load $c$; the vertex
$v^p_{t_p}$ receives $e^p_{t_p}$ and $e^p_{t_p+1}$, whose indices have
opposite parity, so its load is $1 + S + c$; every other inner vertex
receives one edge, of size at most $1$.

($\ge$) Encode an orientation of path $p$ by $\sigma \in \{L,R\}^{2k+1}$,
where $\sigma_t = L$ means $e^p_t$ goes to $v^p_{t-1}$ and $\sigma_t = R$ to
$v^p_t$. The inner vertex
$v^p_t$ receives both its edges exactly when
$\sigma_t \sigma_{t+1} = RL$, and then carries $1 + S + c$, the two indices
having opposite parity. If no such $t$ exists the word contains no factor
$RL$, so $\sigma = L^a R^b$ with $a + b = 2k+1$; then $u$ receives $e^p_1$
if $a \ge 1$ and $v$ receives $e^p_{2k+1}$ if $b \ge 1$, so path $p$
delivers a size-$1$ job to a pole either way. If that happens on all $r$
paths, one pole receives at least $\lceil r/2 \rceil \ge 2$ size-$1$ jobs
and carries at least $2 + c \ge 1 + S + c$.
\end{proof}

\noindent
We use the hypothesis $r \ge 3$ once, in the last step, and it is necessary:
at $r = 2$ the orientation $L^{2k+1}$ on one path and $R^{2k+1}$ on the
other gives makespan $1 + c$.

\begin{proposition}[the family defeats no matching]\label{prop:floor}
Let $r \ge 3$, $k \ge 1$, $0 < S \le \tfrac12$ and $0 \le c < 1$, let $x$ be
any feasible solution at any target $T \in [\max(1, 2c),\, 2)$, and take any
valid matching. (The bound $c < 1$ is what keeps that interval non-empty; the interval's right
endpoint $2$ is what keeps every size-$1$ job big, which the proof assumes
throughout.) Then the makespan on $F(r,k,S,c)$ is exactly $1 + S + c$.
\end{proposition}

\begin{proof}
Since $T < 2$, every size-$1$ job is big. First, $c \le T/2$: otherwise
every dedicated job is big with share $1$ on its own machine, exhausting
that machine's big-job budget, so no size-$1$ edge could be covered at all.
Hence $d_u$ is not big, Step~2 does not take it, and being a singleton it
must be matched to a slot of $u$. (At $c = 0$ there is no $d_u$: $u$'s mass is
then at most $1$, so it has a single slot and takes at most one edge, and the
bound below holds unchanged.) Next, every job on $u$ other than $d_u$ is
big, so the big-job constraint gives $\sum_p x_{e^p_1, u} \le 1$ and $u$'s
fractional mass is at most $2$: at most two slots. Finally, if Step~2
assigns some $e^{p_0}_1$ to $u$ then $x_{e^{p_0}_1,u} \ge \beta$, so
$x_{e^p_1,u} \le 1 - \beta$ for every other $p$, whence
$x_{e^p_1, v^p_1} \ge \beta$ and Step~2 assigns those edges to their inner
endpoints; no edge remains in $u$'s slot structure.

Now bound each machine. An inner vertex is incident to exactly three jobs
---its dedicated job, one size-$1$ edge and one size-$S$ edge---so under any
assignment whatever its load is at most $1 + S + c$. For $u$: if Step~2
placed an edge there, $u$ has a single slot, which $d_u$ takes, and the load
is $1 + c$; otherwise $u$ carries no Step-2 job and at most two slots, one
of them $d_u$'s, hence at most one matched edge, and the load is again at
most $1 + c$. The same for $v$. Since $S \ge 0$ every machine is at most
$1 + S + c$, and a valid matching is a schedule, so by
Lemma~\ref{lem:familyopt} the makespan is at least $\OPT = 1 + S + c$.
\end{proof}

\noindent
So on this family Step~3's freedom makes no difference: best matching, worst
matching and optimum coincide, for every feasible $x$, every target in the
stated range and every $\beta > 1/2$, and not only for the best matching.

The relaxation's least feasible threshold on the family is bounded below by
an explicit maximum. Write
\[
\theta_{\mathrm{cnt}} \;=\; \frac{k(1+S) + 1}{2k + 2/r},
\qquad
\theta_{\mathrm{inn}} \;=\; \frac{r\bigl(k(1+S) + 1\bigr) - 2}{2kr},
\]
the first the total work spread over all $2kr + 2$ machines, the second what
the $2kr$ inner vertices carry once each pole is capped at edge weight $1$.
Then, \emph{at every feasible target $T < 2$},
\[
T \;\ge\; L \;:=\; \max\Bigl(1,\ 2c,\ c + \tfrac{r}{r+1},\
c + \theta_{\mathrm{cnt}},\ c + \theta_{\mathrm{inn}}\Bigr) .
\]
The hypothesis is not decoration, though three of the five terms do not
need it. Those three hold at every feasible target, because none uses a
big-job row: the size-$1$ jobs give $1$; the cut gives $c + \tfrac{r}{r+1}$,
since its $r$ jobs can sit only on $u$ and the $v^p_1$, and each of those $r+1$
machines has room $T - c$ once its dedicated job is placed; and spreading the
total work over all machines gives $c + \theta_{\mathrm{cnt}}$. The term
$c + \theta_{\mathrm{inn}}$ rests on big-job rows that bind the size-$1$ edges
only while those edges are big, that is while $T < 2$. The term $2c$ also uses
those rows, but for $c < 1$ it holds at every $T \ge 2$ anyway, since then
$T \ge 2 > 2c$; it needs the hypothesis only when $c \ge 1$. Dropping
the hypothesis makes the display false, through $\theta_{\mathrm{inn}}$: at $r = 5$, $k = 1$,
$S = \tfrac12$, $c = \tfrac{99}{100}$ the least feasible target is
$\tfrac{1219}{600}$ while $L = \tfrac{51}{25}$ is larger by $\tfrac1{120}$
(\texttt{p29lever\_tlp\_bound\_check\_2026-09-08.py}). The only fact used below still follows from the hypothesis: if $L \ge 2$ then no
target under $2$ is feasible.
Whenever the relaxation is feasible at a target $T < 2$ this gives
$L \le T_{\mathrm{LP}} \le T < 2$; $L$ itself is not below $2$ for every
parameter, and at $r = 5$, $k = 1$, $S = \tfrac12$, $c = \tfrac{99}{100}$ it is
$\tfrac{51}{25}$, which is only to say that the family is infeasible below $2$
there.
The first four terms are the size-$1$ jobs; the bound $c \le T/2$ above; the
cut $\{e^p_1\}_p$; and the total work spread over $2kr + 2$ machines. The
third needs its argument written out, because the inequality it starts from
gives the term only in one of two cases. Write $\tau = T - c$ for what a
machine has left once its dedicated job is placed. The $r$ jobs of the cut
have total weight $r$ and meet only $u$ and the $v^p_1$; each $v^p_1$
absorbs at most $\tau$, and $u$ at most $\min(\tau, 1)$, the $1$ because
those jobs have size $1$ and are big for every $T < 2$, so the big-job
constraint caps their total share at $u$ by $1$. Hence
$r \le \min(\tau,1) + r\tau$. If $\tau \le 1$ that reads
$r \le \tau(1 + r)$, i.e.\ $\tau \ge r/(r+1)$; and if $\tau > 1$ then
$T > c + 1 > c + r/(r+1)$ outright. Either way $T \ge c + r/(r+1)$.

The fifth term sharpens the fourth. Only
$e^p_1$ reaches $u$ and only $e^p_{2k+1}$ reaches $v$; both have size $1$
and so are big for every $T < 2$, so the big-job constraint caps the edge
weight at each pole by $1$ and the rest of the work --- total
$r\bigl(k(1+S)+1\bigr) - 2$ at least --- is confined to the $2kr$ inner
vertices, which also carry their own dedicated jobs.

Neither counting term dominates: $\theta_{\mathrm{inn}} >
\theta_{\mathrm{cnt}}$ exactly when $\theta_{\mathrm{cnt}} > 1$, both
inequalities reducing to $r\bigl(k(1+S)+1\bigr) > 2kr + 2$. Inside the range
$r \ge 3$, $k \ge 1$, $0 < S \le \tfrac12$ used below, that happens only
for $k = 1$ with $rS > 2$; everywhere else the fifth term is slack and $L$
is the four-term maximum.

\begin{measurement}[the general lower bound is attained where it was sampled]\label{meas:familythreshold}
Lemma~\ref{lem:threepathexact} proves $T_{\mathrm{LP}} = L$ on the extremal line
and nothing depends on this measurement any longer; what it still says is that
the equality is not peculiar to that line.
$T_{\mathrm{LP}} = L$ at every point of a grid over $r \in \{3,4,5,6,7,9,
12\}$, $k \le 4$, $S \in \{\tfrac18, \allowbreak \tfrac14, \allowbreak \tfrac25, \allowbreak \tfrac12\}$ and
$c \in \{0, \allowbreak \tfrac18, \allowbreak \tfrac14, \allowbreak
\tfrac25\}$ with at most $120$ machines
(\texttt{p29amb\_\allowbreak family\_\allowbreak threshold.py}, which also evaluates
the formal boundary $S = 0$, outside the family by Lemma~\ref{lem:familyopt},
where the size-$S$ jobs have size zero); the four-term maximum, omitting
$\theta_{\mathrm{inn}}$, fails at $44$ of those points, all with $k = 1$ and $rS > 2$. Equality
also holds at every $k \le 20$ and at $k = 25, 30, 40$ along the extremal line
$r = 3$, $S = \tfrac12$, $c = \tfrac14$, agreeing there with the value
Lemma~\ref{lem:threepathexact} proves. We do not claim equality in general, and
nothing here does: the proposition's strict half uses only the displayed
inequality and its extremal half uses only the lemma.
\end{measurement}

\begin{lemma}[the threshold on the extremal line]\label{lem:threepathexact}
For every $k \ge 1$,
\[
  T_{\mathrm{LP}}\bigl(F(3,k,\tfrac12,\tfrac14)\bigr) \;=\; \frac{12k+7}{12k+4} .
\]
So the least feasible threshold on this line is fixed by $k$ alone, and it
exceeds $1$ for every $k$, approaching $1$ from above as $k$ grows.
\end{lemma}

\begin{proof}
\emph{The lower bound} is the displayed one, and on this line its fourth term is
the largest: $c + \theta_{\mathrm{cnt}} = \tfrac14 + \tfrac{9k+6}{12k+4} =
\tfrac{12k+7}{12k+4}$, against $1$, $\tfrac12$ and $1$ from the first three and
$1 + \tfrac1{6k}$ from the fifth, which is smaller for every $k \ge 1$.

\emph{The upper bound} needs a feasible solution at exactly that target, and one
is found rather than guessed by asking for the solution that loads every machine
equally. Write $\tau$ for the edge load every machine is to carry, so that with
its dedicated job its load is $T = \tfrac14 + \tau$, and let $y_t$ be the share
of $e^p_t$ placed on its left endpoint, the same on all three paths since the
paths are interchangeable. The left pole carries three shares $y_1$ of size-$1$
jobs, so $3y_1 = \tau$. An odd inner vertex carries the right share of a size-$1$
job and the left share of a size-$\tfrac12$ one, and an even inner vertex the
reverse, so
\[
  (1 - y_{2m+1}) + \tfrac12 y_{2m+2} \;=\; \tau ,
  \qquad
  \tfrac12(1 - y_{2m}) + y_{2m+1} \;=\; \tau ,
\]
that is $y_{2m+2} = 2(\tau - 1 + y_{2m+1})$ and
$y_{2m+1} = \tau - \tfrac12 + \tfrac12 y_{2m}$. Substituting the first into the
second, the odd shares advance by a constant,
\[
  y_{2m+3} - y_{2m+1} \;=\; 2\tau - \tfrac32 ,
\]
so the whole point is determined by $\tau$ --- and $\tau$ is determined by the
far end, since the right pole carries three shares $1 - y_{2k+1}$ and
$3(1 - y_{2k+1}) = \tau$ closes the system. It closes at
$\tau = \tfrac{9k+6}{12k+4}$, giving $T = \tfrac{12k+7}{12k+4}$ and
\[
  y_{2m+1} \;=\; \frac{3k+2+6m}{12k+4} \quad (0 \le m \le k),
  \qquad
  y_{2m} \;=\; \frac{3m-1}{3k+1} \quad (1 \le m \le k).
\]

Feasibility is then four checks. Every share lies strictly between $0$ and $1$:
the odd ones run from $\tfrac{3k+2}{12k+4}$ up to $\tfrac{9k+2}{12k+4}$ and the
even ones from $\tfrac{2}{3k+1}$ up to $\tfrac{3k-1}{3k+1}$. Every machine's load
is exactly $T$, by the three equalities the shares were built to satisfy. Every
job has size at most $1 < T$, so none is zeroed by the target. And $T > 1$ makes
$T/2 > \tfrac12$, so the size-$1$ edges are the only big jobs: an inner vertex
meets exactly one of them and its big-job budget is that single share, while a
pole meets three, of total share $\tau < 1$ for every $k \ge 1$. Hence
$T_{\mathrm{LP}} \le T$, and with the lower bound, equality.
\end{proof}

\begin{proposition}[the family's ceiling]\label{prop:familyceiling}
On $F(r,k,S,c)$ with $r \ge 3$ and $0 < S \le \tfrac12$, the ratio
$(1+S+c)/T_{\mathrm{LP}}$ is strictly below $\tfrac74$, and $\tfrac74$ is the
supremum: along $r = 3$, $S = \tfrac12$, $c = \tfrac14$ the ratio is exactly
$\tfrac74 \cdot \tfrac{12k+4}{12k+7}$, which increases to $\tfrac74$ as
$k \to \infty$. Only members with $r = 3$ can approach it: for $r \ge 4$ the cut
term holds the ratio at or below $\tfrac32 + \tfrac1{r+1} \le \tfrac{17}{10}$,
although the proof's own upper bound tends to $\tfrac74$ for every $r \ge 3$.
\end{proposition}

\begin{proof}
Only the first and fourth of the five lower-bound terms for $T_{\mathrm{LP}}$
carry the strict bound. The approach clause is separate: it evaluates the ratio
along one extremal line, where Lemma~\ref{lem:threepathexact} determines the
threshold exactly.

Write $\theta = \theta_{\mathrm{cnt}}$, and note
$\theta - \tfrac{1+S}{2} = \tfrac{r-1-S}{2rk+2} > 0$ for $r \ge 3$. Only
$T_{\mathrm{LP}} \ge \max(1, c + \theta)$ is used. If $c + \theta \le 1$ then
the ratio is at most $1 + S + c \le 2 + S - \theta$; if $c + \theta > 1$ it
is at most $1 + (1 + S - \theta)/(c + \theta) < 2 + S - \theta$, using
$\theta < 1 + S$. The bound $2 + S - \theta$ increases in $S$, since
$\partial_S \theta = k/(2k + 2/r) < 1/2$, so at $S = 1/2$
\[
\frac{1+S+c}{T_{\mathrm{LP}}} \;\le\; \tfrac52 - \frac{\tfrac32 k + 1}{2k + \tfrac2r}
\;<\; \tfrac52 - \tfrac34 \;=\; \tfrac74 ,
\]
the last step because $r > 3/2$. For the supremum, at $r = 3$, $S = 1/2$,
$c = 1/4$ Lemma~\ref{lem:threepathexact} gives
$T_{\mathrm{LP}} = (12k+7)/(12k+4)$ exactly, while
Lemma~\ref{lem:familyopt} gives $\OPT = 1 + S + c = \tfrac74$. So the ratio
there is $\tfrac74 \cdot (12k+4)/(12k+7)$, below $\tfrac74$ at every finite $k$
and increasing to it.
\end{proof}

\noindent
For $r \ge 4$ the cut term holds the ratio to at most $\tfrac32 + \tfrac1{r+1}$,
and the next result says exactly where that cap is reached. The intuition is a
balance. The short edges must be as long as allowed, the dedicated load must sit
where the cut and the size-$1$ jobs balance, and the paths must be long enough to
spread the remaining work.

\begin{proposition}[where the cut cap is attained]\label{prop:cutattained}
On $F(r,k,S,c)$ with $r \ge 3$, $k \ge 1$, $0 < S \le \tfrac12$ and $c \ge 0$,
\[
\frac{\OPT}{T_{\mathrm{LP}}} \;=\; \frac32 + \frac1{r+1}
\quad\text{if and only if}\quad
r \ge 4,\quad S = \tfrac12,\quad c = \tfrac1{r+1},\quad k(r-3) \ge 2(r-1),
\]
and at every such point $T_{\mathrm{LP}} = 1$.
\end{proposition}

\noindent
In words: the four conditions are each necessary and together sufficient. The
cap is reached exactly on the line $S = \tfrac12$, $c = \tfrac1{r+1}$, once each
path has at least $\lceil 2(r-1)/(r-3) \rceil$ pairs of edges, and never at
$r = 3$, where the length condition would read $0 \ge 4$.

\begin{proof}
Write $d = \tfrac1{r+1}$. \emph{Necessity.} The $r$ jobs $e^p_1$ have total
size $r$ and can sit only on $u$ and the $r$ machines $v^p_1$, each of which has
room $T - c$ after its dedicated job. So $r \le (r+1)(T-c)$ at every feasible
target, and with the size-$1$ jobs $T_{\mathrm{LP}} \ge \max(1,\, c + 1 - d)$.
By Lemma~\ref{lem:familyopt}, $\OPT = 1 + S + c$. If $c \le d$ the ratio is at
most $1 + S + c \le \tfrac32 + d$, with equality only when $S = \tfrac12$,
$c = d$ and $T_{\mathrm{LP}} = 1$. If $c > d$ it is at most
$(\tfrac32 + c)/(c + 1 - d)$, which decreases strictly in $c$ and equals
$\tfrac32 + d$ at $c = d$, so it stays strictly below the cap. At
$S = \tfrac12$, $c = d$ and $T = 1$, the total work
$r(\tfrac32 k + 1) + (2kr+2)d$ cannot exceed the capacity $2kr + 2$ of the
machines; multiplying out, that is $k(r-3) \ge 2(r-1)$.

\emph{Sufficiency.} Let $g = \tfrac12 - 2d$, positive for $r \ge 4$. On every
path put on the left endpoint of $e^p_t$ the share
\[
y_{2m+1} = \min(1-d,\ d + mg) \quad (0 \le m \le k), \qquad
y_{2m+2} = \min(1,\ 2mg) \quad (0 \le m < k),
\]
and the rest on the right endpoint. So the odd shares climb from $d$ in steps
of $g$ until they reach $1 - d$, and each even share is twice the climb so far,
capped at $1$. The pole $u$ carries $d + rd = 1$. An odd inner machine carries
$d + (1 - y_{2m+1}) + \tfrac12 y_{2m+2} \le 1$, because
$y_{2m+2} \le 2(y_{2m+1} - d)$; an even one carries
$d + \tfrac12(1 - y_{2m+2}) + y_{2m+3} \le 1$, because
$y_{2m+3} \le \tfrac12 - d + \tfrac12 y_{2m+2}$. The length condition is
exactly $d + kg \ge 1 - d$, so $y_{2k+1} = 1 - d$ and the pole $v$ also
carries $1$. At $T = 1$ only the size-$1$ jobs are big; an inner machine meets
one of them, and each pole holds a total share $rd < 1$ of them. So the
relaxation is feasible at $T = 1$, a size-$1$ job rules out anything smaller,
and the ratio is $1 + \tfrac12 + d$.
\end{proof}

\noindent
Exactly three members load every machine to $1$, so that total work equals
total capacity: $k(r-3) = 2(r-1)$ forces $r - 3$ to divide $4$, giving
$(r,k) = (4,6), (5,4), (7,3)$ and ratios $\tfrac{17}{10}$, $\tfrac53$ and
$\tfrac{13}{8}$. Every longer path still attains the cap, with capacity to
spare. \texttt{p29\_extensions/\allowbreak p29lever\_cutattained\_check.py}
checks the shares, both necessity steps and the three saturated members in
exact rational arithmetic. Figure~\ref{fig:cutcap} maps which pairs $(r,k)$ attain
the cap and outlines these three.

\begin{figure}[htbp]\centering
\includegraphics[width=0.97\textwidth]{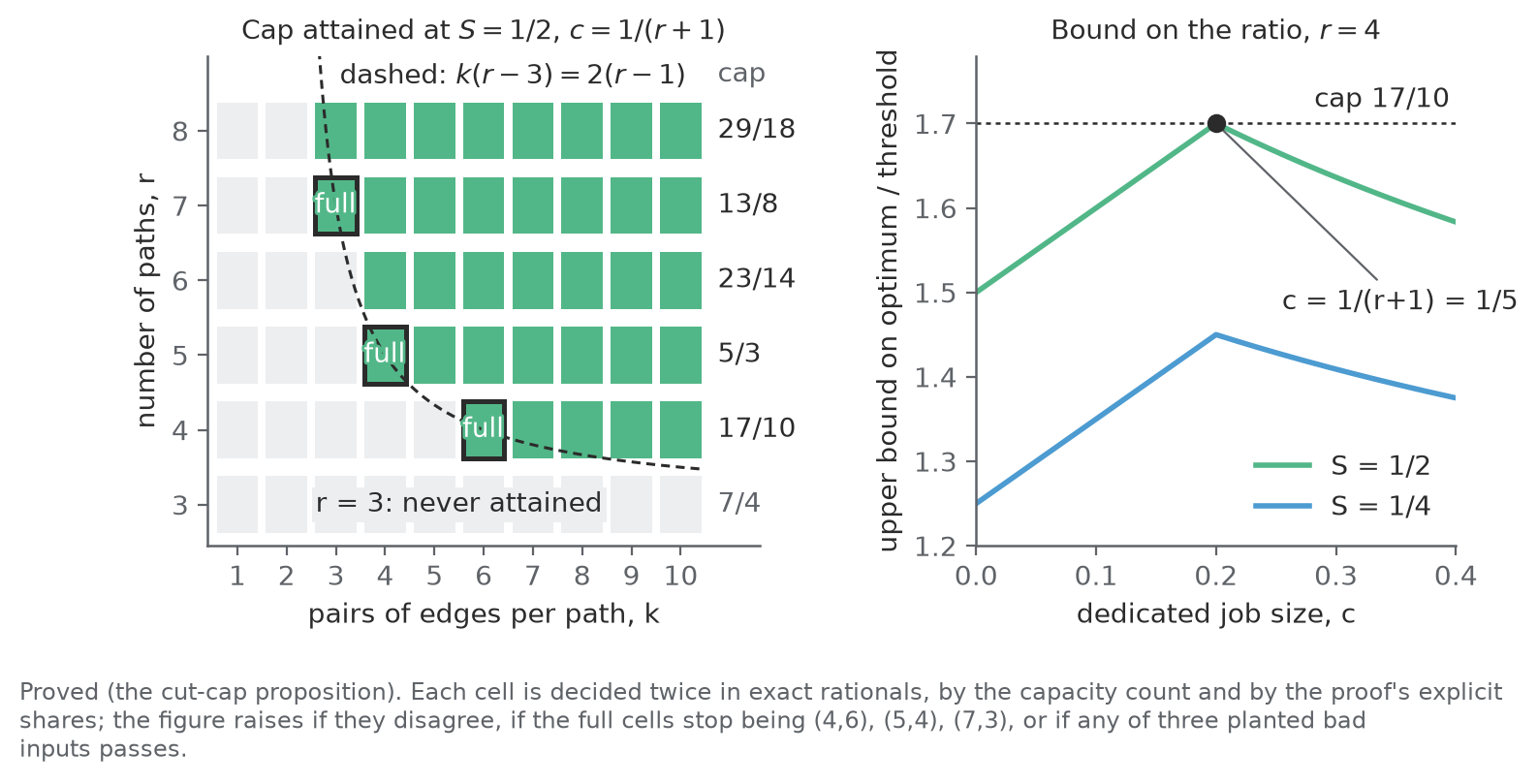}
\caption{The pairs $(r,k)$ at which the cut cap of Proposition~\ref{prop:cutattained} is attained.
\textbf{Left:} at $S = \tfrac12$ and $c = \tfrac1{r+1}$, a cell $(r,k)$ is green
when the relaxation is feasible at $T = 1$, so that $\OPT/T_{\mathrm{LP}}$ equals
the cap $\tfrac32 + \tfrac1{r+1}$ printed beside its row. Each cell is decided
twice in exact arithmetic, by the capacity count (necessary) and by the proof's
explicit shares (sufficient), and the two agree everywhere; the green region is
exactly the set the length condition $k(r-3) \ge 2(r-1)$ describes, and the row
$r = 3$ stays empty. Outlined cells are the three members that load every machine
to exactly $1$. \textbf{Right:} for $r = 4$, the upper bound the cut and the
size-$1$ jobs place on the ratio, as the dedicated size $c$ varies. It meets the
cap only at $S = \tfrac12$ and $c = \tfrac15$, which is why those two conditions
are necessary; the length condition on the left is what makes them sufficient.}
\label{fig:cutcap}
\end{figure}

\medskip\noindent
Three warnings about the obvious closed form. Reading the threshold off the
extremal line gives $\max(1, \tfrac12 + \tfrac S2 + c) + \Theta(1/k)$; that is
right at $r = 3$, $S = 1/2$, $c = 1/4$---where the correction is
$\tfrac{r-1-S}{2rk+2}$, positive, with coefficient of $1/k$ tending to $1/4$---and
wrong elsewhere, even on $r = 3$, $S = 1/2$: at $c = 0$ the threshold is exactly
$1$. Three ways stand out. The term $c + \tfrac{r}{r+1}$ is missing, and below
$S = \tfrac{(k-1)(r-1)}{k(r+1)}$ the error is a \emph{constant}, not
$O(1/k)$: at $r = 3$, $k = 4$, $S = \tfrac14$, $c = \tfrac12$ that form gives
$\tfrac98$, which is below $L = \tfrac54$ --- and $L$ is a lower bound, so the
form is not the threshold. There $L$ is carried by the very term the form omits,
$c + \tfrac{r}{r+1}$, and holds at $\tfrac54$ for every $S \le \tfrac38$. The term $2c$ is also missing. And where
$\max(1, \tfrac{1+S}{2} + c) = 1$ strictly exceeds $c + \theta_{\mathrm{cnt}}$
the threshold need not be $1$: at $r = 3$, $k = 4$, $S = \tfrac14$,
$c = \tfrac3{10}$ the form gives $1$ and $c + \theta_{\mathrm{cnt}} = \tfrac{129}{130}$,
yet the cut term forces $T_{\mathrm{LP}} \ge c + \tfrac{r}{r+1} = \tfrac{21}{20}$.
A measurement taken on the extremal line sees none of this.

Second, the extremal choice is $r = 3$, not any $r \ge 3$: for $r \ge 4$ the
cut term caps the ratio at $\tfrac32 + \tfrac1{r+1}$, which is $1.7$ at
$r = 4$, and the cap is attained, exactly where Proposition~\ref{prop:cutattained}
says: at $r = 4$, $k = 6$, $S = \tfrac12$, $c = \tfrac15$ the relaxation is
feasible at $T = 1$ while $\OPT = \tfrac{17}{10}$, and every longer path on the
same line does as well. Exact maximization over $r \le 8$, $k \le 40$ and a rational grid in
$(S,c)$ puts the global maximum at $r = 3$.

Third --- and this one is a correction to the four-term form above, not to
the draft before it --- that maximum was asserted as an equality, and it is
not one. At every feasible target below $2$ the four terms are each a valid
lower bound, but their maximum is
not always attained: at $r = 5$, $k = 1$, $S = \tfrac12$, $c = 0$ the
four-term value is $25/24$ while the relaxation is infeasible there and
first becomes feasible at $21/20 = c + \theta_{\mathrm{inn}}$. Both verdicts
are by exact rational simplex and are deposited
(\texttt{p29amb\_family\_\allowbreak threshold\_\allowbreak refutation\_\allowbreak r5k1S1\_2c0\_1.json}).
The fifth term repairs it. Nothing downstream moves: the four-term maximum
understates, never overstates, so every ratio computed from it is an
overestimate of the true one, and the maximization just described --- whose
optimum sits at $r = 3$, where the fifth term is slack and the four-term
form is exact --- keeps its conclusion with a little room to spare. The four-term
form is a valid lower bound on $T_{\mathrm{LP}}$ and not in general its value;
it is used here only where it is exact.

So the family that forces $7/4$ on the relaxation forces nothing more on
the variant: the extra load that would overshoot also inflates the
threshold, at exactly the compensating rate. The separation is now exact rather
than sampled. On the generalized three-path family, at every feasible target
below $2$ and for $c < 1$, every valid matching finishes at $\OPT$
(Proposition~\ref{prop:floor}), while $\OPT/T_{\mathrm{LP}}$ stays
strictly below $\tfrac74$ at every finite member and reaches it only in the
limit. That family is extremal for the \emph{relaxation} and costs the rounding
nothing; the families of Sections~\ref{sec:74} to~\ref{sec:vertex} are extremal
for the \emph{rounding} at points where the relaxation's own gap is comfortable.
The two are different phenomena, which is Remark~\ref{rem:gapscope}'s point met
from the other side.

\paragraph{Role of computation.}
The results above rest on explicit constructions and proofs, not on the absence
of counterexamples in a search. The deposited computations regenerate the
reported measurements and independently check the arithmetic, matching
enumerations, and vertex certificates used in those constructions.

\section{Open questions}\label{sec:open}

The constants are determined. Against the threshold the relaxation was solved at
and against the true optimum alike, optimizing Step~3 leaves the supremum at
$\tfrac{11}{6}$, restricting Step~1 to a vertex leaves it there, and imposing
both at once leaves it there as well; the strict ceiling of \cite{Paper1} rules
out attainment in every case. That is Theorem~\ref{thm:invariance}, and it is
what this note has to say about the two freedoms. The threshold-dependent guarantee is also determined in \cite{Paper1}: above $\tfrac12$ the guarantee is exactly
$\max\{\tfrac32 + \tfrac\beta2,\ \tfrac52 - \beta\}$, and at or below $\tfrac12$
exactly $\tfrac32 + (1-\beta)\lfloor 1/\beta\rfloor$, on both scales, so a lower
threshold does not improve the guarantee but destroys it, and $\tfrac23$ is the
unique best choice. Four questions remain.

\emph{The size of a near-extremal vertex instance.} Write $N(\varepsilon)$ for the fewest machines of an instance with $T = \OPT = 1$ whose relaxation has a vertex at which every valid matching has makespan at least $\tfrac{11}{6} - \varepsilon$, and $J(\varepsilon)$ for the fewest jobs. For
$0 < \varepsilon < \tfrac14$, a gap of $\tfrac1{4g}$ needs
$g \ge \tfrac1{4\varepsilon}$, and the least such $g$ is then at least $2$, as
the proposition requires; by
Proposition~\ref{prop:compressed} the member at that $g$ can be built on $5g+2$
machines with $12g+2$ jobs, so
\[
  N(\varepsilon) \;\le\; 5\Bigl\lceil \tfrac1{4\varepsilon} \Bigr\rceil + 2 ,
  \qquad
  J(\varepsilon) \;\le\; 12\Bigl\lceil \tfrac1{4\varepsilon} \Bigr\rceil + 2 ,
\]
whence $\limsup_{\varepsilon \downarrow 0}\varepsilon N(\varepsilon) \le \tfrac54$
and $\limsup_{\varepsilon \downarrow 0}\varepsilon J(\varepsilon) \le 3$. It is open whether $N(\varepsilon) = \Theta(1/\varepsilon)$. The explicit constant sharpens the
question: what is $\liminf_{\varepsilon \downarrow 0}\varepsilon N(\varepsilon)$,
and is the $\tfrac54$ this construction achieves optimal? A lower bound is the
missing half.

\emph{Can the two choices be made together in polynomial time?}
Theorem~\ref{thm:step3hard} decides the subproblem for a \emph{supplied} $x$ and
no more: Step~1 is not adversarial, and a different Step-1 output on the same
instance may make Step~3 trivial. We do not know whether a polynomial-time map exists taking an instance and a feasible threshold to a feasible relaxation solution together with a minimum-makespan valid matching for it. If none exists, everything here about
the oracle concerns a procedure of no direct algorithmic use.

\emph{The cost of completion at a structured $x$.} The same question with
Step~1 pinned: is minimum-makespan completion polynomial-time solvable when $x$
is a vertex returned by a specified polynomial-time linear-programming algorithm?
Theorem~\ref{thm:twoslot} answers the two questions the weak reduction of
Theorem~\ref{thm:step3hard} left: for a supplied $x$, neither bounded
denominators nor two slots per machine make completion easy. For a supplied
$x$, a fixed-parameter algorithm in the number of machines alone is excluded
unless P${}={}$NP, because the reduction uses only two machines joined by
parallel jobs; the reduction's $x$ is not a vertex, so this
exclusion does not reach the structured case asked about here.

\emph{Does a well-chosen Step~1 help?} Every lower bound here, and in
\cite{Paper1}, lets Step~1 return any feasible point or any vertex, chosen
against the scheme. On every optimum-scale family in the two notes the least
feasible threshold equals the optimum, so the integral optimum schedule is itself
a feasible point and a vertex, and at it every valid matching finishes at the
optimum: a Step-1 rule that found it would have ratio $1$ there. Finding it is
NP-hard, so this is no algorithm. Whether some polynomial-time rule for choosing
the Step-1 point improves $\tfrac{11}{6}$ is open.

\section{Summary of results}\label{sec:summary}

This section consolidates the results. Nothing here is new; every line identifies the result supporting it, and the three headings distinguish what is proved, what is measured at sampled parameters, and what is open.

\paragraph{Proved.} Three theorems, one per setting of the scheme, and they
agree: Section~\ref{sec:74} against the threshold, Section~\ref{sec:oracle}
against the optimum, Section~\ref{sec:vertex} at a vertex. Both $\tfrac74$
candidate guarantees are false, and they fall to different
instances. Against the threshold the relaxation was solved at,
Theorem~\ref{thm:74false} gives a $19$-machine instance on which every valid
matching finishes at $44/25$ or above, and Theorem~\ref{thm:116} fixes the
constant any such bound must reach at exactly $\tfrac{11}{6}$, which is what the
algorithm already guarantees. Against the true optimum,
Theorem~\ref{thm:74bfalse} gives each unit job of that family its own partner,
which leaves every forced load unchanged and drops the optimum to $1$, forcing
the oracle to $85/48 > \tfrac74$ at $g = 4$; Corollary~\ref{cor:oracle} then
puts the oracle's worst ratio at exactly $\tfrac{11}{6}$, approached and
attained by nothing. Removing Step~3's freedom entirely --- at a cost we do not know how to pay in
polynomial time, though Theorem~\ref{thm:step3hard} does not show it cannot be
paid --- does not improve the constant that freedom was assumed to account
for. For a supplied fractional solution, the completion problem of
Theorem~\ref{thm:step3hard}, which that theorem shows only weakly NP-hard, is
NP-complete in the strong sense already with two slots per machine, and
polynomial with one (Theorem~\ref{thm:twoslot}).

Restricting Step~1 to an adversarially chosen \emph{vertex} of the relaxation does not help either:
Theorem~\ref{thm:pinvertex} gives, for every $g \ge 2$, a vertex at which every
valid matching finishes at exactly $\tfrac{11}{6} - \tfrac1{4g}$. That family
sits at $T = \OPT = 1$, so it settles both scales at once
(Corollary~\ref{cor:vertexconst}) and binds every rule for choosing the matching,
randomized ones included (Corollary~\ref{cor:anyrule}). The four settings of the
two freedoms therefore carry one constant, $\tfrac{11}{6}$, attained in none of
them: Theorem~\ref{thm:invariance}. We also prove a reduction
confining every overload to two shapes (Theorem~\ref{thm:reduction}), show that
the natural two-phase repair cannot work (Proposition~\ref{prop:twophase}), and
establish that $\tfrac74$ strictly bounds the generalized three-path family and
is its exact supremum (Lemma~\ref{lem:threepathexact} and
Proposition~\ref{prop:familyceiling}), with the cut cap for $r \ge 4$ attained
exactly where Proposition~\ref{prop:cutattained} says.

\paragraph{Measured, and not proved.}
Nothing the body relies on. Measurement~\ref{meas:thrvertex} reports the
threshold-relative constant on a smaller dump-free construction only over the
range enumerated, and Corollary~\ref{cor:vertexconst} settles that constant
without it. Every statement labelled
\emph{Measurement} reports what the deposited code returned at the stated
parameters; none of them is a proof.

\paragraph{Open.}
Not the vertex restriction under an adversarial Step~1 (whether a well-chosen vertex
helps is the open question of Section~\ref{sec:open}), and not attainment: the ceiling of \cite{Paper1} is
strict at every feasible point, so no execution reaches $\tfrac{11}{6}$ anywhere,
and Theorem~\ref{thm:invariance} records all four suprema as unattained. Not the
threshold either, which \cite{Paper1} settles over its whole range. The size of a near-extremal vertex instance and the selection complexity remain open (Section~\ref{sec:open}).

Nothing in this note improves the approximation ratio for graph balancing,
which stands at $1.75$.

\appendix

\section{Computational record}\label{app:computational}

This appendix collects what the body does not need in order to be read: the
controls, the independent implementations, the search sizes, and the deposited
script behind each measured claim. Nothing here carries a proof. The deposited
bundles' READMEs map every claim to the command and artifact that establish it.

\paragraph{Theorem~\ref{thm:74false}'s instance, and Lemma~\ref{lem:ce74opt}.}
The exact, solver-free check \texttt{p29amb\_ce74\_deposit.py} rebuilds the
instance, verifies feasibility and least feasibility, decides infeasibility at
$\tfrac74$ and feasibility at $\tfrac{44}{25}$, and writes the attaining
schedule to \texttt{p29amb\_ce74\_witness.json}. The independent arithmetic
check \texttt{p29lever\_review4\_checks.py} verifies the support facts and load
inequalities used in Lemma~\ref{lem:ce74opt}.

\paragraph{Theorem~\ref{thm:116}'s family (Measurement~\ref{meas:116}).}
The exact builder \texttt{p29amb\_116\_deposit.py} regenerates the family, and
\texttt{p29amb\_partner\_scarcity\_referee.py} independently enumerates its
valid matchings. The structure is checked for $g = 1, \dots, 40$, and for every
job rather than a
sample: each $A_k$ meets exactly one slot at its blocker, each $D_k$ exactly
two, each crammer's share mass is exactly $2$, the sink's exactly $g - 1$, and
the deficiency is exactly $1$. The minimum over valid matchings is decided
exactly for $g = 1, \dots, 6$, each time infeasible strictly below
$\tfrac{11}{6} - \tfrac1{4g}$ and feasible at it, with the probe set to half the
instance's own load lattice rather than to a fixed constant; and every valid
matching is enumerated there --- $5$, $135$, $3{,}040$, $65{,}745$,
$1{,}407{,}275$ and $30{,}041{,}280$ of them --- minimum equal to maximum each
time. The tie-break sweep covers all orderings the nonincreasing-size rule
leaves free up to $g = 5$, $3{,}840$ of them at that parameter; $g = 6$ is
measured under the canonical order alone, because the sweep re-enumerates
everything once per ordering and there are $46{,}080$ of them, some hundreds of
hours at the rate the deposited file measures for itself. That one row is
therefore weaker than the other five, and the artifact carries the count and the
cost. The fifth member, at $107/60 = 1.78\overline{3}$, is the first to exceed
$16/9$. Both routines carry positive controls designed for other claims: the
enumerator returns exactly the eight valid matchings \cite{Paper1} enumerates,
worst $\tfrac{11}{6} - \delta$ and best $1$, and the optimum routine returns
$\tfrac43 + \tfrac52\delta$ on the six-job family. The deposit scripts
\texttt{p29amb\_\allowbreak 116\_\allowbreak deposit.py} and
\texttt{p29amb\_\allowbreak rigidity\_\allowbreak deposit.py} refuse to write when a recomputed value
disagrees with the claim.

\paragraph{Theorem~\ref{thm:74bfalse}'s family (Measurement~\ref{meas:74b}).}
\texttt{p29amb\_partner\_scarcity\_referee.py} rebuilds the slot structure from
Section~\ref{sec:prelim}'s definitions in exact rational arithmetic and
enumerates every valid matching. The deposited log records $g \le 4$; the
$g = 5$ and $g = 6$ rows come from re-running it, and $g = 7$'s forced load is
the closed form rather than a run. Three controls, each built to fail and each
failing for a named reason. The calibration run reproduces the two values
Theorem~\ref{thm:116} states in print at $g = 1$, $\OPT = \tfrac{13}{12}$ and a
forced $\tfrac{19}{12}$. The unmodified family, run through the identical code,
does not clear $\tfrac74$ at $g = 4$ --- $\tfrac{85}{61}$ --- despite an
identical forced value, so the partner change and not the harness is what clears
$\tfrac74$. And moving a little of one filler's share from its crammer to the
sink, which gives the sink one more slot and changes nothing else, drops the
forced load by exactly $\tfrac12$ at every $g$: the pigeonhole, isolated.
Reversing the free tie-break between the two equal-sized $A$ jobs on a crammer
changes nothing. The paired control on the unmodified family is measured rather
than predicted at $g \le 6$: run at $g = 6$ the deposited script returns
$\tfrac{31}{24}$ there, by the same exhaustive branch and bound over each job's
own support.

\texttt{p29amb\_partner\_vertex\_check.py} computes the rank of the tight
constraints by exact Gauss--Jordan, with a feasible integral point as a positive
control; \texttt{p29amb\_\allowbreak vertex\_\allowbreak forcing\_\allowbreak search.py} runs the purification
search of Remark~\ref{rem:74bvertex}, rejecting any endpoint that fails either
the feasibility or the vertex test rather than counting it; and
\texttt{p29amb\_vertex\_oracle\_search.py} runs the random-instance search with
the unpurified point of Theorem~\ref{thm:74bfalse} as its positive control,
which must return $\tfrac74$ at $g = 3$ and $\tfrac{85}{48}$ at $g = 4$, since
a low maximum over vertices would otherwise mean only that the pipeline cannot
find a large ratio anywhere. Both harnesses refuse to report any ratio at or
above $\tfrac{11}{6}$, which the ceiling of \cite{Paper1} forbids and which
would therefore be a defect in the harness rather than a result.

\paragraph{Theorem~\ref{thm:pinvertex}'s family.}
The general counts are reproduced at every checked $g$: $10g + 3$ machines,
$22g + 4$ variables, and a minimum over valid matchings of exactly
$\tfrac{11}{6} - \tfrac1{4g}$. The climb up $g$ stopped at a recursion limit rather than at a failure; the largest instance checked has $2{,}253$ machines and an active set of rank $4{,}954$. Vertexhood at $g = 4$ and $g = 8$
was decided independently by Bareiss fraction-free integer elimination in a
program sharing no code, giving rank $92$ of $92$ and $180$ of $180$, which the
general count $22g + 4$ reproduces. Five checks, each able to fail: the same
pipeline returns $\tfrac{85}{48}$ on Theorem~\ref{thm:74bfalse}'s own
\emph{non}-vertex point; buying the sink one extra slot collapses the lower
bound to $\tfrac{61}{48}$; a smaller family on $18$ machines reaches
$\tfrac{37}{21} > \tfrac74$ at a vertex of rank $79$ of $79$, so $43$ machines
is not what the construction needs; deleting one padding job leaves the rank
full while deleting two drops it, which is the single dependency a spanning-tree
support predicts; and the decider returns both verdicts at every $g$, infeasible
at $\tfrac74$ and feasible at the claimed value.

\paragraph{The two vertex searches, and the per-machine and closed-form checks.}
The three measurements below are reported here rather than in the body: none of
them carries a statement the argument depends on.

\begin{measurement}[the two vertex searches, and they are not one search]
\label{meas:vertexoracle}
Both searches build random instances around a job of size exactly $1$, so that
the least feasible threshold and the optimum are both pinned at $1$ and the
oracle's ratio is its makespan; both purify to a point whose active constraints
have full rank, which is what ``certified'' means here.
(a) \emph{One pool.} \texttt{p29amb\_vertex\_oracle\_search.py}, run at $900$
trials (\texttt{python3 p29amb\_vertex\_oracle\_search.py 900}) at the script's
own default seed, certifies $207$ vertices. Their ratios are
$\tfrac{17}{12}$ once, $\tfrac43$ five times, $\tfrac54$ nine times,
$\tfrac76$ seven times, $\tfrac{13}{12}$ twice, and exactly $1$ at the
remaining $183$; the tally sums to $207$. The largest is
$\tfrac{17}{12} \approx 1.417$. The script's deposited $400$-trial default
certifies $97$ vertices with $84$ at exactly $1$ and the same largest value.
Both runs are deterministic in the seed and were re-executed for this
revision.
(b) \emph{Four size bands.} \texttt{p29amb\_vertex\_oracle\_trend.py} runs
$700$ trials in each of four bands of increasing instance size and certifies
$162$, $134$, $71$ and $39$ vertices; the largest ratio is $\tfrac43$ in the
first three bands and $\tfrac76$ in the fourth. \textbf{No band produced
$\tfrac{17}{12}$}, which is why (a) and (b) are reported apart: the largest
value either search found comes from (a), and reading it off (b)'s bands would
attribute it to an experiment that did not produce it.
The search of (a) runs a control before it reports anything --- the same
pipeline on Theorem~\ref{thm:74bfalse}'s own unpurified point, which must
return $\tfrac{85}{48}$ at $g = 4$ --- and abandons the run if it fails; both
scripts stop outright on any ratio at or above $\tfrac{11}{6}$, which
the $11/6$ ceiling of \cite{Paper1} forbids and which would therefore be a defect
in the harness rather than a result.
\end{measurement}

\begin{measurement}\label{meas:permachine}
\texttt{p29amb\_permachine\_check.py} evaluates the three per-machine bounds
directly --- (N0) against $\tfrac32$, (N1) against $2 - x_q/2$, and (S) against
$\tfrac32 + p/3$ --- on \emph{every} valid matching of $4{,}000$ random feasible
instances: $13{,}174$ matchings enumerated in full, zero violations, and the
proposition's own conclusion (every machine strictly below $\tfrac{11}{6}$) holds
throughout. (N1) is checked against the bound the proof derives rather than
against $\tfrac{11}{6}$, since a check of the weaker number would pass even where
the derivation is wrong.

The margins say what this does and does not test. The tightest the random sample
comes is $\tfrac7{36}$ below (S), $\tfrac14$ below (N1) and $\tfrac5{12}$ below
(N0), so it exercises the bounds without approaching them; what approaches them
is Theorem~\ref{thm:116}'s family, which the same file runs as a control and
which sits $\tfrac1{16}$ under the ceiling at $g = 4$.
\end{measurement}

\begin{measurement}\label{meas:prop7}
(a) The closed form
$f(S,c) = (1 + S + c)/\max(1, \tfrac12 + \tfrac S2 + c)$ was maximized on
a $2000 \times 1200$ grid of intervals over $S \in [0, 1/2]$, $c \in [0,1]$: supremum
$1.75000000$, attained at $(S, c) = (0.5, 0.25)$. Its denominator is
not the threshold even at $r = 3$, $S = \tfrac12$, $c = \tfrac14$, where it is
the limit as $k \to \infty$ of the finite-$k$ value $\tfrac{12k+7}{12k+4}$ of
Lemma~\ref{lem:threepathexact}; the maximization is recorded because it located
that line, not because the form is the threshold.
(b) The sign and order of the finite-$k$ correction were measured
directly at $S = 0.499$, $c = 0.25$, where the limit value is $1$. The relaxation
threshold is $1.1871$, $1.1067$, $1.0745$, $1.0572$, $1.0390$, $1.0295$,
$1.0198$, $1.0148$ at $k = 1,2,3,4,6,8,12,16$; the excess over $1$ times
$k$ rises from $0.1871$ to $0.2371$ over that range. That range does not
show the asymptotics, which are not $+\Theta(1/k)$: $T_{\mathrm{LP}} = 1$ for
every $k \ge 500$, because the left-endpoint shares
$y_{2m+1} = \min(\tfrac34, \tfrac14 + \tfrac{m}{1000})$ and
$y_{2m+2} = \min(1, \tfrac{m}{499})$ load every machine to at most $1$, within every
big-job budget, once $k \ge 500$, and a size-$1$ job rules out anything smaller. The corresponding ratios
$(1 + S + c)/T_{\mathrm{LP}}$ rise from $1.4733$ to $1.7235$ over the reported
range and equal $1 + S + c = 1.749$ from $k = 500$ on, not $\tfrac74$. \texttt{p29amb\_meas\_prop7\_check.py} walks the grid
and the $r = 3$, $S = 0.499$, $c = 0.25$ line that these numbers come from.
Two conventions in the above need stating, and both are easier to get wrong
by reading than by running: the four-decimal rule is nearest with ties
toward zero, not truncation, and $(T-1) \cdot k$ dips by $0.0003$ at
$k = 16$, so it flattens rather than being nondecreasing.
\end{measurement}

\section{Proof of Theorem~\ref{thm:twoslot}}\label{app:twoslot}

The strategy is a reduction from satisfiability in which each kind of machine
has one role: a literal machine holds a size-$1$ job on whichever side is false,
a clause machine keeps at most two of its occurrence jobs, and the half-open
slots make ``at most two'' the only constraint a clause imposes.

\begin{proof}
\emph{One slot.} With at most one slot per machine, each machine receives at
most one job in Step~3, so a machine's final load is its Step-2 load plus that
one job's size. Delete every job--slot edge whose job would push its machine
above $T$. A matching covering every remaining job exists exactly when a
completion of makespan at most $T$ does, and bipartite matching is polynomial.

\emph{Two slots.} Membership in NP is as in Theorem~\ref{thm:step3hard}. Reduce
from 3-SAT. Replace the $t$ occurrences of each variable $v$ by fresh variables
$v_1, \dots, v_t$ and, when $t \ge 2$, add the clauses
$(\neg v_1 \vee v_2), \dots, (\neg v_t \vee v_1)$, which force the copies to
agree. The new formula $\Phi$ is satisfiable exactly when the old one is; its
clauses have two or three literals, and each literal occurs at most twice.

For each variable $v$ of $\Phi$ make machines $L_v^+$ and $L_v^-$ and a job
$Q_v$ of size $1$ with share $\tfrac12$ on each. For a three-literal clause make
a machine $C$ and three jobs of size $\tfrac12$, one to the machine of each
literal, with shares $\tfrac34$, $\tfrac12$, $\tfrac34$ on $C$ in tie order and
the rest on the literal. On $C$ they occupy $[0,\tfrac34)$,
$[\tfrac34,\tfrac54)$ and $[\tfrac54,2)$, so any two of them, but not all three,
can be matched to $C$'s two slots. For a two-literal clause make $C$, two
occurrence jobs of size $\tfrac12$ with share $\tfrac12$ on $C$, and a
size-$\tfrac12$ job with share $1$ on $C$ and $0$ on a private machine, ordered
occurrence, dedicated, occurrence; they occupy $[0,\tfrac12)$,
$[\tfrac12,\tfrac32)$ and $[\tfrac32,2)$, so the dedicated job and either
occurrence fit. At a literal machine $Q_v$ comes first and takes
$[0,\tfrac12)$. If the literal has two occurrences and the first has share
$\tfrac14$ there, insert between them a size-$\tfrac12$ padding job with share
$\tfrac14$ there and $\tfrac34$ on a private machine; the second occurrence then
starts at $1$, and the two occurrences can use slots $1$ and $2$ separately.
Every pair of machines carries at most one job, every load is at most $1$, every
machine has fractional mass at most $2$ and so at most two slots, and only the
$Q_v$ are big, each with share $\tfrac12 < \tfrac23$, so Step~2 assigns nothing.
The tie orders are those of one global order of job identifiers: clause by
clause, each dedicated or padding job between the two occurrence jobs it
separates.

If $\Phi$ is satisfied, send $Q_v$ to the machine of its false literal, each
occurrence to its literal if that literal is true and to its clause otherwise,
and each padding or dedicated job to its private machine or its clause. A clause
keeps at most two occurrence jobs, or its dedicated job and at most one; a true
literal receives at most its two occurrences, in separate slots; every load is
at most $1$. Conversely, in a completion of makespan at most $1$ no occurrence
job shares a machine with a $Q_v$, since $1 + \tfrac12 > 1$, and every clause
must send an occurrence away, since keeping all of them costs $\tfrac32$. That
occurrence goes to a literal machine without its $Q_v$; calling those literals
true satisfies every clause. All numbers are $\tfrac14$, $\tfrac12$,
$\tfrac34$ or $1$, so the hardness is in the strong sense.
\end{proof}

\section*{Code and data}
The code and artifacts needed to reproduce every computational result reported
in this note, except the solver-behaviour observation of
Measurement~\ref{meas:74b} and the regeneration of the $14$-job instance of
Measurement~\ref{meas:greedy}, are deposited on the Open Science Framework at
\url{https://doi.org/10.17605/OSF.IO/SKX86}, the same project that holds
the deposit of \cite{Shavit26}, as five bundles that each run on their own:
\texttt{osf\_\allowbreak p29lever\_\allowbreak core.zip} regenerates everything except four long runs,
and each of those has a bundle of its own,
\texttt{osf\_\allowbreak p29lever\_\allowbreak rigidity\_\allowbreak g4.zip},
\texttt{osf\_\allowbreak p29lever\_\allowbreak tiebreak\_\allowbreak g5.zip},
\texttt{osf\_\allowbreak p29lever\_\allowbreak crux\_\allowbreak g20.zip} and
\texttt{osf\_\allowbreak p29lever\_\allowbreak extend\_\allowbreak g.zip}. Each holds the claim-supporting scripts
it runs, the modules they import, a driver and pinned dependency versions,
with the stored artifacts its README lists; some checks report to the console
only. The deposit is the replication set, not the
project's working directory. Each bundle's README maps its claims to the command
and artifact that establish them and lists every file it holds.

Theorem~\ref{thm:pinvertex}'s scripts are deposited with the JSON they
write: the builder
\texttt{p29amb\_\allowbreak pinned\_\allowbreak vertex\_\allowbreak family.py}
and the second opinion
\texttt{p29amb\_\allowbreak independent\_\allowbreak verify.py}, which shares
no code with it by design, together with the proof checks, the smaller
$18$-machine witness, the climb up $g$, the perturbed family of
Measurement~\ref{meas:pinvertex}, and the counting-crux programs run under both
boundary readings. All of them carry a different prefix from the rest, because
the family they build post-dates the note's own naming. The claim maps of the
bundles name each with the command that runs it.

The deposit reproduces every numerical measurement reported here except the
solver-behaviour fact inside Measurement~\ref{meas:74b}, which that measurement
marks as having no regenerating script, and the $14$-job instance of
Measurement~\ref{meas:greedy}, which the deposit ships but a rerun replaces, as
that measurement says; each bundle's
README lists its files, including the modules it carries from the deposit of
\cite{Shavit26}.

All other computational results reported here were regenerated from the
deposited code and artifacts. The public package contains the final claim-supporting
checks; exploratory proof attempts and the internal revision diary are not part
of the replication set.

\section*{Acknowledgements}

I thank Baruch Schieber (NJIT) for introducing me to the
restricted-assignment and graph-balancing gap problems; this note began
from reading Wang and Sitters within the problem area he opened up for
me. I thank JMS and CS for reading and commenting on drafts.

\section*{Use of AI tools}

AI tools assisted under the author's direction. Claude (Anthropic) wrote search,
verification and figure code to the author's specification. Claude and
ChatGPT (OpenAI) gave editorial support. Perplexity was used for literature search and citation verification.
ChatGPT (OpenAI) and Gemini (Google) were used for adversarial review and
independent mathematical checking, and Wolfram Mathematica for checking
derivations.

\vfill
\begin{center}
\begin{minipage}{0.72\textwidth}
\itshape\small
We chose the visible steps\\
the best to ever tread\\
but we were pinned \& didn't move\\
even with the corners of the polytope\\
the same not there.\\
The lever is in some other place\ldots
\end{minipage}
\end{center}

\end{document}